\def\llncs{0}
\def\fullpage{1}
\def\anonymous{0}

\def\notxfont{0}
\def\submission{0}
\def\cameraready{0}

\ifnum\submission=1
\def\anonymous{1}
\def\llncs{1}

\else
\fi

\ifnum\cameraready=1
\def\submission{1}
\def\llncs{1}

\def\anonymous{0}
\else
\fi

\ifnum\anonymous=1
\def\llncs{0}

\else
\fi

\ifnum\llncs=1
	\documentclass[envcountsect,a4paper,runningheads,11pt]{llncs}
\else
	\documentclass[letterpaper,hmargin=1.05in,vmargin=1.05in,11pt]{article}
			\ifnum\fullpage=1
		\usepackage{fullpage}
		\fi
\fi

\usepackage[%
  colorlinks=true,
  citecolor=blue,
  pagebackref=true
]{hyperref}

\usepackage{amsmath, amsfonts, amssymb, mathtools,amscd}

\usepackage{amsthm}

\usepackage{fancyhdr}
\usepackage{lmodern}
\usepackage[T1]{fontenc}
\usepackage[utf8]{inputenc}

\usepackage{arydshln} % In order to use \hdashline
\usepackage{url}
\usepackage{ifthen}
\usepackage{bm}
\usepackage{multirow}
\usepackage[dvips]{graphicx}
\usepackage[usenames]{color}
\usepackage{xcolor,colortbl} % Leave out in case the usepackage cannot be found. Not critical
\usepackage{threeparttable}
\usepackage{comment}
\usepackage{paralist,verbatim}
\usepackage{cases}
\usepackage{booktabs}
\usepackage{braket}
\usepackage{cancel} 
\usepackage{ascmac} 
\usepackage{framed}
\usepackage{physics}
\usepackage{authblk}
\usepackage{pifont}
\definecolor{darkblue}{rgb}{0,0,0.6}
\definecolor{darkgreen}{rgb}{0,0.5,0}
\definecolor{maroon}{rgb}{0.5,0.1,0.1}
\definecolor{dpurple}{rgb}{0.2,0,0.65}

\usepackage[capitalise,noabbrev]{cleveref}
\usepackage[absolute]{textpos}
\usepackage[final]{microtype}
\usepackage[absolute]{textpos}
\usepackage{everypage}

\newtheoremstyle{thicktheorem}%
{\topsep}
{\topsep}
{\itshape}{}%
{\bfseries}%
{.}
{ }%
{\thmname{#1}\thmnumber{ #2}%
		\thmnote{ (#3)}%
}

\newtheoremstyle{remark}%name
{\topsep}
{\topsep}
	{}%body font
	{}%indent amount
	{}%theorem head font
	{.}%punctuation after theorem head
	{ }%space after theorem head
	{\textit{\thmname{#1}}\thmnumber{ #2}%theorem head specs
			\thmnote{ (#3)}%
	}

\ifnum\llncs=0
	\theoremstyle{thicktheorem}
	\newtheorem{theorem}{Theorem}[section]
	\newtheorem{lemma}[theorem]{Lemma}
        \newtheorem{assumption}[theorem]{Assumption}
	\newtheorem{corollary}[theorem]{Corollary}
	
	\newtheorem{definition}[theorem]{Definition}
	\newtheorem{game}[theorem]{Game}

	\theoremstyle{remark}
	\newtheorem{claim}[theorem]{Claim}
	\newtheorem{remark}[theorem]{Remark}

\else
\fi
	\crefname{theorem}{Theorem}{Theorems}
	\crefname{assumption}{Assumption}{Assumptions}
	\crefname{construction}{Construction}{Constructions}
	\crefname{corollary}{Corollary}{Corollaries}
	\crefname{conjecture}{Conjecture}{Conjectures}
	\crefname{definition}{Definition}{Definitions}
	\crefname{exmaple}{Example}{Examples}
	\crefname{experiment}{Experiment}{Experiments}
	\crefname{counterexample}{Counterexample}{Counterexamples}
	\crefname{lemma}{Lemma}{Lemmata}
	\crefname{observation}{Observation}{Observations}
	\crefname{proposition}{Proposition}{Propositions}
	\crefname{remark}{Remark}{Remarks}
	\crefname{claim}{Claim}{Claims}
	\crefname{fact}{Fact}{Facts}
	\crefname{note}{Note}{Notes}

\ifnum\llncs=1
 \crefname{appendix}{App.}{Appendices}
 \crefname{section}{Sec.}{Sections}
\else
\fi

\ifnum\llncs=1
\renewcommand*{\backref}[1]{}
\else
	\renewcommand*{\backref}[1]{(Cited on page~#1.)}
	\ifnum\notxfont=1
	\else
		\usepackage{newtxtext}
	\fi
\fi

\newcommand{\taiga}[1]{$\ll$\textsf{\color{red} Taiga: { #1}}$\gg$}
\newcommand{\mor}[1]{$\ll$\textsf{\color{blue} Tomoyuki: { #1}}$\gg$}

\newcommand{\Samp}{\algo{Samp}}
\newcommand{\Ver}{\algo{Vrfy}}
\newcommand{\ans}{\mathsf{ans}}
\newcommand{\puzz}{\mathsf{puzz}}

\newcommand{\la}{\leftarrow}
\newcommand{\ra}{\rightarrow}

\newcommand{\seteq}{\coloneqq}

\newcommand{\cA}{\mathcal{A}}
\newcommand{\cB}{\mathcal{B}}
\newcommand{\cC}{\mathcal{C}}
\newcommand{\cD}{\mathcal{D}}

\newcommand{\cF}{\mathcal{F}}
\newcommand{\cG}{\mathcal{G}}

\newcommand{\cL}{\mathcal{L}}
\newcommand{\cM}{\mathcal{M}}

\newcommand{\cO}{\mathcal{O}}
\newcommand{\cP}{\mathcal{P}}
\newcommand{\cQ}{\mathcal{Q}}
\newcommand{\cR}{\mathcal{R}}
\newcommand{\cS}{\mathcal{S}}
\newcommand{\cT}{\mathcal{T}}
\newcommand{\cU}{\mathcal{U}}
\newcommand{\cV}{\mathcal{V}}

\newcommand{\cZ}{\mathcal{Z}}

\def\makeuppercase#1{
\expandafter\newcommand\csname tl#1\endcsname{\widetilde{#1}}
}

\def\makelowercase#1{
\expandafter\newcommand\csname tl#1\endcsname{\widetilde{#1}}
}

\newcommand{\N}{\mathbb{N}}

\newcommand{\R}{\mathbb{R}}

\newcommand*{\algo}[1]{\ensuremath{\mathsf{#1}}}

\newenvironment{boxfig}[2]{\begin{figure}[#1]\fbox{\begin{minipage}{0.97\linewidth}
                        \vspace{0.2em}
                        \makebox[0.025\linewidth]{}
                        \begin{minipage}{0.95\linewidth}
            {{
                        #2 }}
                        \end{minipage}
                        \vspace{0.2em}
                        \end{minipage}}}{\end{figure}}

\newcommand{\bit}{\{0,1\}}

\newcommand{\Gen}{\algo{Gen}}

\newcommand{\Vrfy}{\algo{Vrfy}}

\newcommand{\negl}{{\mathsf{negl}}}

\newcommand{\poly}{{\mathrm{poly}}}

\usepackage{tikz}

\ifnum\llncs=1
\let\oldvec\vec% Store \vec in \oldvec
\let\vec\oldvec% Restore \vec from \oldvec
\makeatletter
\renewcommand*\l@author[2]{}
\renewcommand*\l@title[2]{}
\makeatletter

\else
\fi    
\theoremstyle{remark}
\title{
\textbf{
Hardness of Quantum Distribution Learning 
and Quantum Cryptography
}}

\begin{document}
%\author{}
%\institute{}

\ifnum\anonymous=1
\author{\empty}
%\institute{\empty}
\else
%
%  For camera ready version.
%
\ifnum\llncs=1
\index{Taiga, Hiroka}
\author{
	Taiga Hiroka\inst{1} 
}
\institute{
	Yukawa Institute for Theoretical Physics, Kyoto University, Japan  \and NTT Corporation, Tokyo, Japan
}
\else
%
%   For full/eprint version, etc.
%

\author[1,2]{Taiga Hiroka}
\author[1]{Min-Hsiu Hsieh}
\author[2,1]{Tomoyuki Morimae}
\affil[1]{{\small Hon Hai Research Institute, Taipei, Taiwan}\authorcr{\small \{taiga.hirooka,min-hsiu.hsieh\}@foxconn.com}}
\affil[2]{{\small Yukawa Institute for Theoretical Physics, Kyoto University, Kyoto, Japan}\authorcr{\small tomoyuki.morimae@yukawa.kyoto-u.ac.jp}}

\renewcommand\Authands{, }
\fi %%%%% END OF LNCS branch
\fi

\ifnum\llncs=1
\date{}
\else
\date{}
\fi

\maketitle

%\ifnum\llncs=0
%\thispagestyle{fancy}
%\rhead{YITP-24-125}
%\else
%\fi

\pagenumbering{gobble} % Turn off page numbering temporarily

\begin{abstract}
The existence of one-way functions (OWFs) forms the minimal assumption in classical cryptography.
However, this is not necessarily the case in quantum cryptography [Kretschmer 2021; Morimae and Yamakawa 2022; Ananth, Qian and Yuen 2022].
One-way puzzles (OWPuzzs) introduced by [Khurana and Tomer 2024] provide a natural quantum analogue of OWFs.
The existence of OWPuzzs implies $\mathbf{PP}\neq\mathbf{BQP}$ [Cavalar, Goldin, Gray, Hall, Liu, and Pelecanos 2025], while the converse remains
an open question.
In classical cryptography, the analogous problem -- whether OWFs can be constructed from $\mathbf{P}\neq\mathbf{NP}$ --
is one of the most central open problems, and has long been studied from the viewpoint of hardness of learning.
Hardness of learning in various frameworks (including PAC learning) has been connected to the existence of OWFs
or to $\mathbf{P}\neq\mathbf{NP}$.
In contrast, no such characterization previously existed for OWPuzzs.
In this paper, we establish the first complete characterization of OWPuzzs based on the hardness of a well-studied learning model:
distribution learning.
Specifically, we prove that OWPuzzs exist if and only if 
proper quantum distribution learning is hard on average.
A natural question that follows from our result is whether the worst-case hardness of
proper quantum distribution learning can be derived from $\mathbf{PP}\neq\mathbf{BQP}$,
because if this is the case and the worst-case to average-case hardness reduction within proper quantum distribution learning is achieved,
it would imply OWPuzzs solely from $\mathbf{PP}\neq\mathbf{BQP}$!
However, we show that answering the question positively would be extremely difficult:
if worst-case hardness of proper quantum distribution learning is $\mathbf{PP}$-hard (in a black-box reduction), 
then $\mathbf{SampBQP}\neq\mathbf{SampBPP}$ is derived solely from the infiniteness of the polynomial-time hierarchy,
which solves a long-standing open problem in the field of quantum advantage.
Despite that, we show that $\mathbf{PP}\neq\mathbf{BQP}$ is equivalent to another standard and well-studied notion of hardness of learning,
namely agnostic.
We show that $\mathbf{PP}\neq\mathbf{BQP}$ if and only if agnostic quantum distribution learning with respect to the KL divergence is hard.
Finally, as a byproduct, we also obtain an interesting implication for quantum advantage: we show that
hardness of agnostic quantum distribution learning with respect to the statistical distance
against ${\rm PPT}^{\Sigma_3^\mathbf{P}}$ learner implies $\mathbf{SampBQP}\neq\mathbf{SampBPP}$.
This provides the first construction of sampling-based quantum advantage from
a hardness assumption on a standard and natural framework in learning theory.
Moreover, hardness of agnostic learning is of the worst-case type, and therefore this result is the first construction of
sampling-based quantum advantage solely from a worst-case hardness assumption.
\end{abstract}

\ifnum\cameraready=1
\else
\ifnum\llncs=1
\else
\newpage
  \setcounter{tocdepth}{2}      % sections in table if depth < i
  \setcounter{secnumdepth}{2}   % sections numbered if depth < i
  \tableofcontents
  \pagenumbering{arabic}
  \setcounter{page}{0}          % set the table contents page as 0-th page
  \thispagestyle{empty}
  \clearpage
\fi
\fi

%kokonifile
\section{Introduction}

In classical cryptography, the existence of one-way functions (OWFs) is the minimal assumption~\cite{FOCS:ImpLub89}:
OWFs are known to be existentially equivalent to a wide range of cryptographic primitives, including
pseudorandom generators (PRGs)~\cite{Hill99}, 
pseudorandom functions (PRFs)~\cite{JACM:GolGolMic86}, 
secret-key encryption (SKE)~\cite{GM84}, message authentication codes (MACs)~\cite{C:GolGolMic84}, 
digital signatures~\cite{STOC:Rompel90}, and commitment schemes~\cite{C:Naor89}.
Moreover, nearly all cryptographic primitives (including public-key encryption (PKE) and multiparty computation) imply OWFs. 

In contrast, in quantum cryptography, OWFs would not necessarily be the minimal assumption~\cite{TQC:Kre21,C:MorYam22,C:AnaQiaYue22}.
Several quantum analogues of OWFs, PRGs, and PRFs have been introduced, 
such as pseudorandom unitaries (PRUs)~\cite{C:JiLiuSon18}, pseudorandom state generators (PRSGs)~\cite{C:JiLiuSon18}, one-way state generators (OWSGs)~\cite{C:MorYam22}, 
one-way puzzles (OWPuzzs)~\cite{STOC:KhuTom24}, and EFI pairs~\cite{ITCS:BCQ23}. 
They could be weaker than OWFs~\cite{TQC:Kre21,STOC:KQST23,STOC:LomMaWri24,myKreQiaTal24},
yet they support a variety of important applications including private-key quantum money~\cite{C:JiLiuSon18}, 
SKE~\cite{C:AnaQiaYue22}, MACs~\cite{C:AnaQiaYue22}, digital signatures~\cite{C:MorYam22}, 
commitments~\cite{C:MorYam22,C:AnaQiaYue22,AC:Yan22}, and multiparty computation~\cite{C:MorYam22,C:AnaQiaYue22,C:BCKM21b,EC:GLSV21}.

In particular, 
one-way puzzles (OWPuzzs) introduced by Khurana and Tomer \cite{STOC:KhuTom24} are
a natural quantum analogue of OWFs, and one of the most fundamental quantum cryptographic primitives.
A OWPuzz is a pair $(\Samp,\Ver)$ of two algorithms. $\Samp$ is a quantum polynomial-time (QPT) algorithm that, on input the security parameter $1^n$, outputs
two classical bit strings, $\ans$ and $\puzz$. $\Ver$ is an unbounded algorithm that, on input $(\puzz,\ans')$, outputs $\top/\bot$.
Correctness requires that $\Ver$ accepts $(\puzz,\ans)\gets\Samp(1^n)$ with high probability.
Security requires that no QPT adversary that receives $\puzz$ can output $\ans'$ such that $(\puzz,\ans')$ is accepted by $\Ver$ with high probability.

The existence of OWPuzzs implies $\mathbf{PP}\neq\mathbf{BQP}$~\cite{arXiv:CGGHLP23}.
On the other hand, in a recent breakthrough paper by Khurana and Tomer~\cite{STOC:KT25}, OWPuzzs were constructed from $\mathbf{PP}\neq\mathbf{BQP}$ in conjunction with 
an additional standard assumption in the field of 
quantum advantage~\cite{STOC:AarArk11,BreMonShe16}, which we refer to as the ``quantum advantage assumption'' for 
simplicity\footnote{This is so-called ``anti-concentration'' and ``average-case $\mathbf{\#P}$-hardness''. For the definition, see \cref{def:QAA}.}.
Although the quantum advantage assumption has been intensively investigated for over a decade, and some partial progress has been made~\cite{NatPhys:BFNV19,Ramis}, 
it remains unproven. Moreover, the assumption was a newly introduced one, devised specifically for this purpose, 
and has not been explored in other areas.
As a result, the following question posed by \cite{STOC:KT25} remains one of the most important open problems in quantum cryptography.
\begin{center}
\emph{Can OWPuzzs be constructed solely from $\mathbf{PP}\neq\mathbf{BQP}$?}
\end{center}

This question serves as a quantum analogue of the long-standing open problem in classical cryptography~\cite{DifHel76}:
Can OWFs be constructed from $\mathbf{P}\neq\mathbf{NP}$ (or $\mathbf{NP}\not\subseteq\mathbf{BPP}$)? 
In the classical setting, OWFs and $\mathbf{P}\neq\mathbf{NP}$ (or $\mathbf{NP}\not\subseteq\mathbf{BPP}$)
have been extensively studied through the lens of learning hardness~\cite{STOC:Valiant84,ACM:PV88,FOCS:ImpLev90,C:BFKL93,STOC:Regev05,C:NaoYog15,FOCS:HirNan23}.
Impagliazzo and Levin \cite{FOCS:ImpLev90} first constructed OWFs based on a specific type of learning hardness, 
and subsequent works~\cite{C:BFKL93,ACM:NR06,C:NaoYog15,FOCS:HirNan23} extended this approach to more standard and general learning models.
In particular, OWFs can be constructed from the average-case hardness of proper PAC learning~\cite{STOC:Valiant84,C:BFKL93},
while $\mathbf{NP}\nsubseteq\mathbf{BPP}$ implies the worst-case hardness of proper PAC learning~\cite{ACM:PV88}.
These characterizations suggest a promising pathway toward resolving the open problem of constructing
OWFs from $\mathbf{NP}\not\subseteq\mathbf{BPP}$, namely,
establishing the worst-case to average-case hardness reduction within proper PAC learning!
%Indeed, \cite{FOCS:HirNan23} made partial progress in this direction by establishing a partial worst-case to average-case reduction for certain learning problems.\mor{aimai}
Learning theory is closely connected to meta-complexity~\cite{Ko91,FOCS:HirNan21,CCC:GKLO22,CCC:GolKab23}, 
which is another compelling
direction for addressing the open problem~\cite{FOCS:Hirahara18,FOCS:LiuPas20,C:LiuPas21,IRS21,C:LiuPas23,STOC:HILNO23,STOC:Hirahara23}.
Therefore, characterizing OWFs with hardness of learning also deepens our understanding of OWFs in terms of meta-complexity. 

Characterizing cryptographic primitives with hardness of learning has another important implication: it can offer
valuable insights into identifying concrete hardness assumptions on which cryptographic primitives can be based.
Indeed, well-known assumptions such as Learning Parity with Noise (LPN) and Learning with Errors (LWE)~\cite{C:BFKL93,STOC:Regev05}
are rooted in the hardness of specific learning problems.
Establishing foundations for quantum cryptographic primitives based on concrete mathematical hard problems remains one of the central goals in the field.

%In this way, characterizing cryptographic primitives with hardness of learning has a long history and many importance.
However, to the best of our knowledge, no such characterization with learning hardness has been established for OWPuzzs.

\subsection{Our Results}
In this paper, we establish, for the first time, complete characterizations of OWPuzzs and $\mathbf{PP}\neq\mathbf{BQP}$ with hardness of learning.

In the classical setting, characterizing cryptographic primitives through the hardness of PAC learning has a long history~\cite{C:BFKL93,ABX08,Nanashima20}.
In (classical) PAC learning, the learner receives samples $((x_1,y_1),(x_2,y_2),...,(x_T,y_T))$, where $x_1,x_2,...,x_T$ are independently sampled from
a distribution $\cD$, $y_i=f(x_i)$, and $f$ is a function computable in classical deterministic polynomial-time.
The learner's goal is to output a hypothesis function $h$ such that
$\Pr_{x\gets \cD}[f(x)\neq h(x)]$ is small.
A natural quantum analogue is that the function $f$ is evaluated in deterministic QPT.
However, OWPuzzs will not be properly characterized with hardness of such quantum PAC learning, because
hardness of such quantum PAC learning is broken with access to a $\mathbf{QCMA}$ oracle,
while OWPuzzs will not be broken with access to a $\mathbf{QCMA}$ oracle~\cite{TQC:Kre21}.
On the other hand, OWPuzzs are broken with access to a $\mathbf{PP}$ oracle~\cite{arXiv:CGGHLP23}.
Thus, to obtain a correct characterization of OWPuzzs, we must consider a learning model whose hardness is not broken by a
$\mathbf{QCMA}$ oracle but broken by a $\mathbf{PP}$ oracle.

Our key idea is to focus on distribution learning~\cite{STOC:KMRRSS94,FOCS:HirNan23}.
Distribution learning is a learning framework in which the learner receives samples from a distribution 
and must ``learn'' the distribution. 
(For example, in proper distribution learning, the learner receives samples from a distribution with a hidden parameter,
and the learner has to recover the parameter.)
Unlike in PAC learning, distribution learning lacks an efficient verification, and therefore its hardness will not be broken with access to a $\mathbf{QCMA}$ oracle.
On the other hand, because the power of $\mathbf{PP}$ is enough to estimate quantum probability distributions~\cite{FR99}, 
hardness of quantum distribution learning should be broken with access to a $\mathbf{PP}$ oracle.

Moreover, distribution learning is among the most standard learning models, and has long been 
studied in the field of learning theory~\cite{Yatracos85,ML:AW92,STOC:KMRRSS94,Devroye_Lugosi,diakonikolas2019,bousquet2020,FOCS:HirNan23}.
This makes it a natural and well-motivated candidate for exploring connections to OWPuzzs.

\if0
\mor{main result no corollary nanodewa?}
Strangely, as far as we know, there is no explicitly written proof that quantum distribution learning is broken with the $\mathbf{PP}$ oracle.
Therefore for our goal, we first have to confirm that distribution learning is indeed broken with the $\mathbf{PP}$ oracle.
Our first result is the following.
\begin{theorem}
If $\mathbf{PP}=\mathbf{BPP}$ (resp. $\mathbf{NP}\subseteq \mathbf{BPP}$), then
for any QPT (resp. PPT) sampleable distribution $\cD$, a PPT algorithm can learn $\cD$.
\end{theorem}
This theorem is informal, and we do not specify in what sense we can ``learn'' $\cD$.
Details will be given later, but now it is enough to know that all types of hardness of learning that we use in this paper
are broken with the $\mathbf{PP}$ oracle.
As far as we know, no upper-bound of the hardness of agnostic distribution learning
was known before even in the classical case.
This theorem gives the first upper-bound.
\fi

\paragraph{OWPuzzs and distribution learning.}
Our first main result is the complete characterization of OWPuzzs with hardness of distribution learning.

\begin{theorem}\label{thm:OWPuzzs_distlearn}
OWPuzzs exist if and only if
proper quantum distribution learning 
is average-case hard.
\end{theorem}
Proper quantum distribution learning is a quantum analogue of
proper distribution learning introduced in \cite{STOC:KMRRSS94}.\footnote{
The word proper means that the learner has to output a hypothesis that is inside of the learning class (which is $\{\cD(z)\}_z$ in the current case).
If the learner can output any distribution that is close to the learning distribution, it is called improper.}
Let $\cD$ be a QPT algorithm that, on input a bit string $z$, outputs a bit string. 
In proper quantum distribution learning, the learner is given samples $x_1,...,x_t$ independently sampled from the distribution $\cD(z)$ without knowing $z$,
and has to output a hypothesis $z^*$ such that 
$\mathsf{SD}(\cD(z), \cD(z^*)) \leq \epsilon$, where $\mathsf{SD}$ is the statistical distance. 
We consider the \emph{average-case hardness} of this task, namely, 
$z$ is sampled from a QPT algorithm $\cS$, and the probability (over $\cS$, $\cD(z)$, and the learner's algorithm)
that the learner outputs $z^*$ such that 
$\mathsf{SD}(\cD(z), \cD(z^*)) \leq \epsilon$ should be small for any QPT learner.

\paragraph{$\mathbf{PP}\neq\mathbf{BQP}$ and distribution learning.}
In our first result above, \cref{thm:OWPuzzs_distlearn},
we have shown that {\it average-case} hardness of proper quantum distribution learning is equivalent to the existence of OWPuzzs.
A natural question raised from this result is whether the {\it worst-case} hardness of proper quantum distribution learning is implied by
$\mathbf{PP}\neq\mathbf{BQP}$.
If this is the case, and if the worst-case to average-case hardness reduction within proper quantum distribution learning is achieved,
OWPuzzs is obtained solely from $\mathbf{PP}\neq\mathbf{BQP}$! 
Unfortunately, showing worst-case hardness of proper quantum distribution learning from $\mathbf{PP}\neq\mathbf{BQP}$ seems to be extremely difficult.
This is suggested from the following our second result.
\begin{theorem}\label{thm:PP_worstdistlearn}
If the worst-case hardness of proper quantum distribution learning is $\mathbf{PP}$-hard in a PPT-black box reduction, then
$\mathbf{PP}\not\subseteq\mathbf{BPP}^{\Sigma_3^\mathbf{P}}$ implies $\mathbf{SampBQP}\neq\mathbf{SampBPP}$.
\end{theorem}
It is a long-standing open problem in the field of quantum advantage whether sampling-based quantum advantage is derived
solely from the infiniteness of the polynomial-time hierarchy. (In fact, \cite{CCC:AarChe17} showed that such a proof should be non-relativized.)
Hence showing the worst-case hardness of proper quantum distribution learning from $\mathbf{PP}\neq\mathbf{BQP}$ seems to be extremely challenging
(at least in a black-box reduction).

We are thus led to consider an alternative question: can $\mathbf{PP}\neq\mathbf{BQP}$ be characterized with any standard hardness notion of distribution learning?
Interestingly, $\mathbf{PP}\neq\mathbf{BQP}$ can be characterized with another well-studied notion of hardness, namely, 
\emph{agnostic}.
We show the following result.
\begin{theorem}\label{thm:PP_agnostic}
$\mathbf{PP}\neq \mathbf{BQP}$
if and only if agnostic quantum distribution learning with respect to the KL divergence is hard.
\end{theorem}
Let us explain what is agnostic quantum distribution learning.
Let $\cD$ be a QPT algorithm that takes a bit string $z$ as input and outputs a bit string.
In agnostic quantum distribution learning, 
the learner receives samples $x_1,...,x_T$ independently sampled from an unknown distribution $\cT$, 
which is not necessarily contained in $\{\cD(z)\}_z$.
The goal of the learner is to output a hypothesis $z^*$ such that the KL divergence \(D_{\mathrm{KL}}(\cT \,\|\, \cD(z^*))\) is close to the optimal value, 
\(\min_{z} D_{\mathrm{KL}}(\cT \,\|\, \cD(z))\). 
The hardness means that no QPT learner can, with high probability, output such a near-optimal \(z^*\).
Agnostic distribution learning has been well studied in the learning theory~\cite{Yatracos85,ML:AW92,STOC:KMRRSS94,Devroye_Lugosi,diakonikolas2019,bousquet2020,FOCS:HirNan23}.
%\mor{KL yori TD no houga standard.}
%\mor{KL wa standard de naikedo, nanode, PP to no toukasei gawakatte omoshiroi}

In summary, we have characterized OWPuzzs and $\mathbf{PP}\neq\mathbf{BQP}$ with two standard and well-studied hardness notions of distribution learning,
namely, proper and agnostic, respectively.
A natural and important question is whether these two notions of hardness can be connected.
If such a connection could be established, it would lead to a construction of OWPuzzs solely from $\mathbf{PP}\neq\mathbf{BQP}$!
However, establishing this connection is extremely challenging and lies beyond the scope of the present paper. 
Even in the classical setting, only characterizing cryptographic primitives with hardness of learning has been considered a significant achievement~\cite{C:BFKL93,ACM:NR06,C:NaoYog15,FOCS:HirNan23}, and, to date,
no one has succeeded in connecting different notions of hardness except for a very recent partial progress in ~\cite{FOCS:HirNan23}.
Therefore we believe that our first complete characterizations of OWPuzzs and $\mathbf{PP}\neq\mathbf{BQP}$
with hardness of distribution learning
constitute a meaningful contribution to the field of quantum cryptography, and
represent an important first step toward the ultimate goal of constructing OWPuzzs from $\mathbf{PP}\neq\mathbf{BQP}$.

As previously discussed,
studying the connection between the hardness of learning and quantum cryptography has another important implication: it may provide
valuable insights into identifying concrete hardness assumptions on which quantum cryptographic primitives can be based
like Learning Parity with Noise (LPN) and Learning with Errors (LWE) assumptions~\cite{C:BFKL93,STOC:Regev05} in classical cryptography.
Establishing quantum cryptographic primitives based on concrete mathematical hard problems is one of the central goals in the field.
We hope that our characterizations of OWPuzzs with hardness of learning
will contribute to the discovery of such concrete assumptions for quantum setting.

In addition to the main results above, we also obtain several interesting additional results, which we will explain in the following.
\cref{figure} also provides a summary of our results.

\paragraph{Sampling-based quantum advantage from hardness of learning.}
Sampling-based quantum advantage~\cite{STOC:AarArk11,BreMonShe16} refers to the existence of a distribution that can be
sampled by a QPT algorithm but not by any PPT algorithm.
Sampling-based quantum advantage has been shown for several ``sub-universal'' quantum computing models including random circuits~\cite{NatPhys:BFNV19}, Boson sampling~\cite{STOC:AarArk11}, IQP~\cite{BreJozShe10,BreMonShe16}, 
and the one-clean-qubit model~\cite{FKMNTT18,MorDQC1additive}.
So far, sampling-based quantum advantage is obtained from the combination of two assumptions. One is $\mathbf{PP}\not\subseteq\mathbf{BPP}^{\mathbf{NP}}$.
The other is a more complicated assumption, which we call ``quantum advantage assumption'' for 
simplicity\footnote{That is so-called ``anti-concentration'' and ``average-case $\mathbf{\#P}$-hardness''. See \cref{def:QAA} for details.}.
Although the quantum advantage assumption has been extensively studied for over a decade and some progress has been made~\cite{NatPhys:BFNV19,Ramis}, 
the assumption itself remains unproven. Moreover, it is a 
newly introduced assumption devised specifically for this purpose, and
has not been explored in any other areas.
Therefore, an important open problem in the field is
to derive sampling-based quantum advantage from a more standard and natural assumption.

We show the following two results. 
%(See also Fig.~\ref{figure}.)
\begin{theorem}
If the quantum advantage assumption holds and $\mathbf{PP}\not\subseteq\mathbf{BQP}^{\Sigma_3^{\mathbf{P}}}$,
then agnostic quantum distribution learning with respect to the statistical distance
is hard against QPT algorithms with access to a $\Sigma_3^{\mathbf{P}}$ oracle.
\end{theorem}
\begin{theorem}\label{thm:qad}
Suppose that agnostic quantum distribution learning with respect to the statistical distance is hard against PPT algorithms with access to a $\Sigma_3^{\mathbf{P}}$ oracle.
Then, $\mathbf{SampBQP} \neq \mathbf{SampBPP}$.
\end{theorem}
These two theorems mean that hardness of agnostic quantum distribution learning is implied by the quantum advantage assumption 
(plus $\mathbf{PP}\not\subseteq\mathbf{BQP}^{\Sigma_3^{\mathbf{P}}}$),
and implies
$\mathbf{SampBQP}\neq\mathbf{SampBPP}$.
(See also \cref{figure}.)
This in particular means that
sampling-based quantum advantage can be derived from potentially weaker and more natural assumption, namely, hardness of agnostic quantum distribution learning. 
As far as we know, this is the first time that sampling-based quantum advantage is derived from
a hardness assumption on a standard and natural framework of learning.
Moreover, hardness of agnostic learning 
is a worst-case hardness,
and therefore our result shows (for the first time) that
sampling-based quantum advantage can be derived
solely from a worst-case hardness assumption.
This is an important advantage, because in general showing worst-case hardness is easier than average-case hardness.

\paragraph{Relation between statistical distance and KL divergence.}
In this paper, we have studied two different notions of hardness of agnostic quantum distribution learning, namely,
with respect to the statistical distance and to the KL divergence.
It is natural to ask how they are related.
We show the following.
\begin{theorem}
The hardness of agnostic quantum distribution learning with respect to the statistical distance implies 
that with respect to the KL divergence.
\end{theorem}
This result is actually obtained as a corollary of the following technical result, which is an
upper bound of the complexity of agnostic quantum distribution learning with respect to the statistical distance.
\begin{theorem}
\label{thm:upperbound_SD}
If $\mathbf{PP}\subseteq\mathbf{BQP}$, then there exists a QPT algorithm which achieves agnostic quantum distribution learning with respect to the statistical distance.
\end{theorem}

How about the opposite direction? 
Does the hardness of agnostic quantum distribution learning with respect to the KL divergence
imply that with respect to the statistical distance?
Showing it seems to be extremely difficult because 
from the following theorem, \cref{thm:bb_SD},\footnote{\cref{thm:PP_worstdistlearn} is actually a corollary of \cref{thm:bb_SD}}
if the hardness of agnostic quantum distribution learning with respect to the statistical distance implies
$\mathbf{PP}\neq\mathbf{BQP}$, then $\mathbf{SampBQP}\neq\mathbf{SampBPP}$ can be obtained solely from
the infiniteness of the polynomial-time hierarchy.
\begin{theorem}
\label{thm:bb_SD}
If the hardness of agnostic quantum distribution learning with respect to the statistical distance is $\mathbf{PP}$-hard in a PPT-black box reduction, then
$\mathbf{PP}\not\subseteq\mathbf{BPP}^{\Sigma_3^\mathbf{P}}$ implies $\mathbf{SampBQP}\neq\mathbf{SampBPP}$.
\end{theorem}

\paragraph{Classical results.}
Although our primary focus is on quantum learning, our proof techniques are also applicable to the classical setting, and in fact
yield several new classical results.
Our first result is the equivalence between the existence of OWFs and hardness of classical distribution learning. 
\begin{theorem}
\label{thm:c}
The following three conditions are equivalent:
\begin{enumerate}
\item OWFs exist.
\item Average-case hardness of proper classical distribution learning holds.
\item Average-case hardness of improper classical distribution learning holds.
\end{enumerate}
\end{theorem}
The equivalence between 1 and 3 is actually known result~\cite{FOCS:HirNan23}, and the direction from 3 to 2 is trivial.
(However, for the convenience of readers, here we have combined all results as a single theorem.)
The direction from 2 to 1 is our new result, which is an improvement of the direction from 3 to 1 shown in \cite{FOCS:HirNan23}. 
The technique of \cite{FOCS:HirNan23} cannot be directly used to show 1 from 2, but here we introduce a new technique that overcomes the issue.
Moreover, also for the direction from 3 to 1, we provide a simpler proof than that of \cite{FOCS:HirNan23}.
In the classical setting, as \cref{thm:c} shows, proper and improper learning are equivalent. On the other hand, in the quantum case, we do not know how to show the equivalence of proper and improper learning, because PRFs are used to show 3 from 1 in the classical proof.

Our second classical result is the characterization of $\mathbf{NP}\nsubseteq\mathbf{BPP}$.
\begin{theorem}
The following three conditions are equivalent:
\begin{enumerate}
\item $\mathbf{NP}\nsubseteq \mathbf{BPP}$.
\item Hardness of agnostic classical distribution learning with respect to KL divergence holds.
\item Worst-case hardness of classical maximum likelihood estimation holds.
\end{enumerate}
\end{theorem}
The equivalence between 2 and 3, and the direction from 1 to 3 were shown in \cite{ML:AW92}. 
Moreover, \cite{ML:AW92} cited a personal communication with Osamu Watanabe as the source of a proof of 1 from 3, 
but to the best of our knowledge, no explicit proof has appeared in the literature.
In this paper, we provide such a proof and thereby complete the logical equivalence.
In addition, we provide simpler proofs than those of \cite{ML:AW92}.

Our final classical result is an upper bound of the complexity of agnostic classical distribution learning.
\begin{theorem}\label{thm:classical_upperbound}
There exists a PPT algorithm with access to a $\Sigma_3^\mathbf{P}$ oracle which achieves agnostic classical distribution learning with respect to the statistical distance.
\end{theorem}
In the field of classical learning, an upper bound of agnostic classical distribution learning with respect to the statistical distance was unknown.

\if0
In \cref{sec:total_variation_distance}, we study the computational complexity of agnostic quantum distribution learning with respect to total variation distance.
In this setting, the learner has sample access to an unknown distribution $\cT$ and is given a description of a family of QPT distributions $\{\cD(z)\}_{z \in \bit^n}$.
We say that an algorithm $\cA$ agnostically learns $\cD$, if the algorithm $\cA$ can find a hypothesis $z^*$ such that the total variation distance $\mathsf{SD}(\cT , \cD(z^*))\leq C\cdot\mathsf{SD}(\cT,\cD(z^*) )+\epsilon$ for some small constant $C$.
Our first result in \cref{sec:total_variation_distance} is the following theorem.
\fi

\usetikzlibrary{positioning} % for position relative to node
\usetikzlibrary{calc} % for computing coordinates
\usetikzlibrary {quotes}
\tikzset{>=latex} % set default arrow head as latex

% TIKZ STUFF
\tikzstyle{mysmallarrow}=[->,black,line width=1.6]
\tikzstyle{myblackbotharrow}=[<->,black,line width=1.6]
\tikzstyle{myredbotharrow}=[<->,red,line width=1.6]
\tikzstyle{newarrow}=[->,red,line width=1.6]
\tikzstyle{newsinglearrow}=[->,red,line width=1.6]
\tikzstyle{carrow}=[->,red,line width=1.6]
% JET CATEGORIES with straight lines
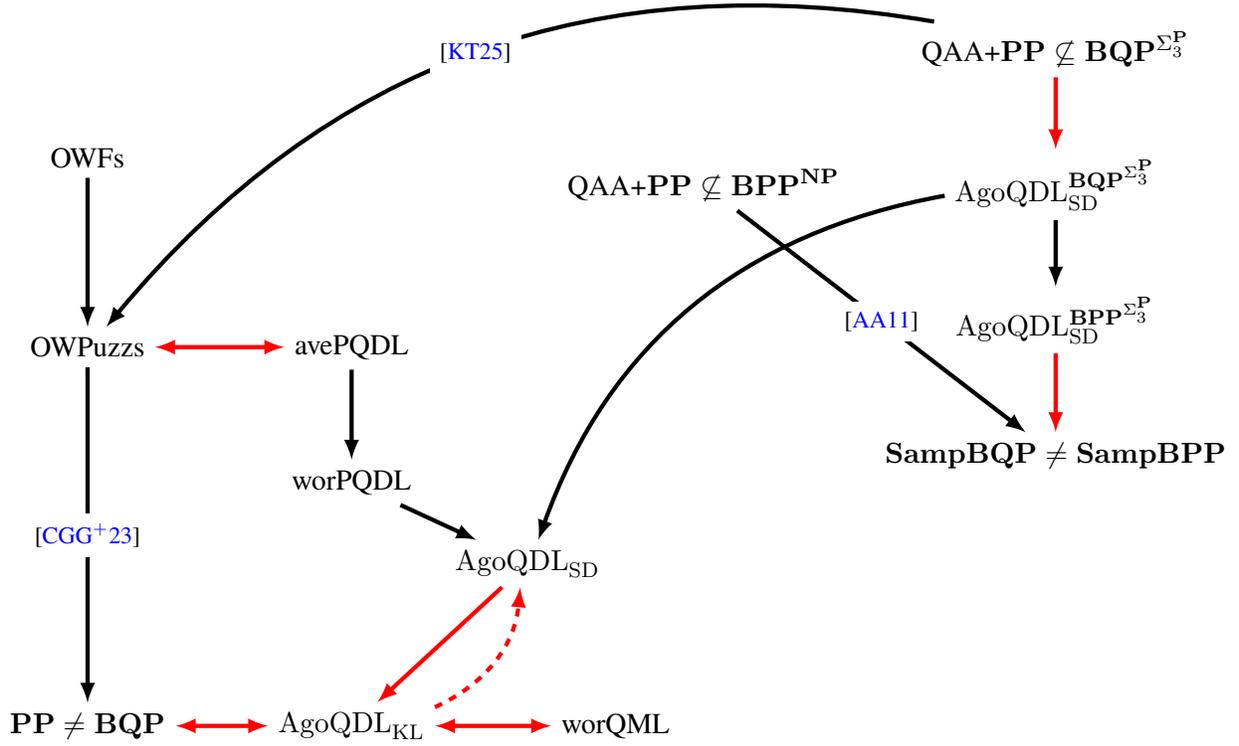
\begin{figure}
\begin{center}
    \begin{tikzpicture}[scale=0.9,every edge quotes/.style = {font=\footnotesize,fill=white}]
    % Nodes
      \def\h{-2.0} % vertical space between rows
      \def\w{2.6} % horizontal space between two main branches

       \node[] (OWFs) at (1*\w,2.8*\h) {OWFs};
       
       \node[] (OWPuzzs) at (1*\w,4.2*\h) {OWPuzzs};
       \node[] (avePQDL) at (2.5*\w,4.2*\h) {avePQDL};
       \node[] (worPQDL) at (2.5*\w,5.2*\h) {worPQDL};
%       \node[] (aveQML) at (4*\w,4.2*\h) {aveQML};

       \node[] (AgoQDL_SD) at (3.5*\w,5.8*\h) {${\rm AgoQDL}_{\rm SD}$};

       \node[] (PP) at (1*\w,7*\h) {$\mathbf{PP}\neq \mathbf{BQP}$};
       \node[] (AgoQDL_KL) at (2.5*\w,7*\h) {${\rm AgoQDL}_{\rm KL}$};
       \node[] (worQML) at (4*\w,7*\h) {worQML};

       \node[] (QAA1) at (4.5*\w,3*\h) {QAA+$\mathbf{PP}\not\subseteq\mathbf{BPP}^{\mathbf{NP}}$};
       \node[] (QAA2) at (6.5*\w,2*\h) {QAA+$\mathbf{PP}\not\subseteq\mathbf{BQP}^{\Sigma_3^\mathbf{P}}$};
       \node[] (AgoQDL_SDQ) at (6.5*\w,3*\h) {${\rm AgoQDL}_{\rm SD}^{\mathbf{BQP}^{\Sigma_3^\mathbf{P}}}$};
       \node[] (AgoQDL_SDP) at (6.5*\w,4*\h) {${\rm AgoQDL}_{\rm SD}^{\mathbf{BPP}^{\Sigma_3^\mathbf{P}}}$};
       \node[] (Samp) at (6.5*\w,5*\h) {$\mathbf{SampBQP}\neq\mathbf{SampBPP}$};

        \draw[mysmallarrow] (QAA2)[bend right=30] edge["\cite{STOC:KT25}"] (OWPuzzs);
        \draw[mysmallarrow] (AgoQDL_SDQ)[bend right=30] edge[] (AgoQDL_SD);

        \draw[mysmallarrow] (avePQDL) edge[] (worPQDL);
        \draw[newarrow] (AgoQDL_SD) edge[] (AgoQDL_KL);
        \draw[newarrow,dashed,color=red][bend right=30] (AgoQDL_KL) edge[] (AgoQDL_SD);
%        \draw[mysmallarrow] (aveQML) edge[] (worQML);
        
        \draw[mysmallarrow] (QAA1) edge["\cite{STOC:AarArk11}"] (Samp);
        \draw[newarrow] (QAA2) edge[] (AgoQDL_SDQ);
        \draw[mysmallarrow] (AgoQDL_SDQ) edge[] (AgoQDL_SDP);
        \draw[newarrow] (AgoQDL_SDP) edge[] (Samp);

        \draw[myredbotharrow] (OWPuzzs) edge[] (avePQDL);
%        \draw[myredbotharrow] (avePQDL) edge[] (aveQML);
        
        \draw[myredbotharrow] (PP) edge[] (AgoQDL_KL);
        \draw[myredbotharrow] (AgoQDL_KL) edge[] (worQML);

        \draw[mysmallarrow] (OWPuzzs) edge["\cite{arXiv:CGGHLP23}"] (PP);
        \draw[mysmallarrow] (OWFs) edge[] (OWPuzzs);
        
        \draw[mysmallarrow] (worPQDL) edge[] (AgoQDL_SD);
    \end{tikzpicture}
\end{center}
\caption{A summary of results. 
Black lines are known results or trivial implications.
Red lines are new in our work. 
avePQDL is average-case hardness of proper quantum distribution learning.
worPQDL is worst-case hardness of proper quantum distribution learning.
%aveQML is average-case hardness of quantum maximum likelihood.
worQML is worst-case hardness of quantum maximum likelihood estimation.
${\rm AgoQDL}_{\rm KL}$ is hardness of agnostic quantum distribution learning with respect to the KL divergence.
${\rm AgoQDL}_{\rm SD}^{X^{\Sigma_3^{\mathbf{P}}}}$ is hardness of agnostic quantum distribution learning with respect to the statistical distance against
$X=\{\mbox{PPT},\mbox{QPT}\}$ learners with access to a $\Sigma_3^\mathbf{P}$ oracle.
QAA stands for the quantum advantage assumption.
The dotted line suggests that showing that implication seems to be difficult.
(If it is shown, then $\mathbf{SampBQP}\neq\mathbf{SampBPP}$ is derived from the infiniteness of the polynomial-time hierarchy.)
}
\label{figure}
\end{figure}

\subsection{Related Works}
\label{sec:relatedworks}

\paragraph{Distribution learning.}
The study of distribution learning has a long history and we cannot cover all results here.
\if0\footnote{
Distribution learning has also been studied over real number spaces. In these communities, distribution learning is called as hypothesis selection or density estimation.}\fi
We only mention four related results.

%Our formalizm of distribution learning is based on \cite{STOC:KMRRSS94}.
%Let us emphasize that it has been an open problem if we can show the NP-hardness of distribution learning with respect to \emph{total variation distance}
%while cryptographic assumptions such as pseudorandom functions directly imply the hardness of distribution learning.\mor{Kono bun yokuwakaran}

%Independent from \cite{STOC:KMRRSS94}, \cite{ML:AW92} also formulates distribution learning. 
%They consider KL-divergence as a measure, and show the NP-hardness of learning Hidden Markov distribution in their learning model. 
%Our formulation of agnostic distribution learning with respect to \emph{KL-divergence} is based on \cite{ML:AW92}.

\cite{FOCS:HirNan23} showed the equivalence between OWFs and the average-case hardness of several learning problems including distribution learning.
Their proof for constructing OWFs from hardness of distribution learning can be directly used to construct OWPuzzs from hardness of quantum ditribution learning, but
in that case we obtain OWPuzzs from hardness of \emph{improper} learning. We use a completely different technique so that we can construct OWPuzzs from 
hardness of \emph{proper} learning, which is stronger.
On the other hand, their proof for showing hardness of distribution learning from OWFs cannot be used in our case, because their technique completely relies on classical properties.
We therefore use another technique, but we do not know how to show hardness of \emph{improper} learning from OWPuzzs.

\cite{SSHE21,myPRL:HINH23,myPJSE24} considered distribution learning in the quantum regime. 
\cite{myPRL:HINH23} showed that specific families of quantum circuits are hard to learn under the LPN assumption.
\cite{SSHE21,myPJSE24} showed existence of some quantum advantage in distribution learning task.

\paragraph{Quantum cryptography and hardness of learning quantum states.}
\cite{myBHHP24,HH24,myFGSY25,PQS24} investigated relationships between quantum cryptography and hardness of learning quantum states
with the motivation of basing quantum cryptographic primitives on hardness of learning. 
\cite{myFGSY25,HH24} showed the equivalence between the existence of OWSGs with pure-state outputs and average-case hardness of learning pure quantum states.
\cite{myBHHP24} 
constructed OWSGs from
hardness of obtaining classical descriptions of given pure states.
\cite{PQS24} introduced a hardness assumption what they call \emph{leaning stabilizers with noise} and constructed EFIs from it.

\paragraph{Quantum cryptography from concrete assumptions.}
The construction of OWPuzzs in \cite{STOC:KT25} from quantum advantage assumption plus $\mathbf{PP}\neq\mathbf{BQP}$
suggests several concrete assumptions for OWPuzzs implementable with sub-universal models such as random circuits and IQP.

\paragraph{Quantum advantage and hardness of estimating quantum probability.}
\cite{STOC:MSY25} showed that average-case hardness of estimating quantum probability implies
$\mathbf{SampBQP}\neq\mathbf{SampBPP}$,
and therefore their hardness could also be considered as another assumption for sampling-based quantum advantage.
However, our assumption, namely, hardness of agnostic quantum distribution learning, is worst-case hardness assumption, and therefore
more advantageous than that of \cite{STOC:MSY25}.

\if0
\paragraph{Sampling-based quantum advantage.}
\taiga{Add more references.}
\cite{STOC:AarArk11,BreMonShe16,myKT24,myMSY24} studies\mor{study} quantum advantage for sampling tasks.
\cite{myMSY24} formalizes the notion of quantum advantage samplers (QASs), which captures an average-case version of $\mathbf{SampBQP} \neq \mathbf{SampBPP}$. 
To demonstrate quantum advantage in sampling, several computational assumptions have been proposed, although none have been proven so far. 
These assumptions, together with the conjecture $\mathbf{P^{\#P}} \nsubseteq \mathbf{BPP^{NP}}$, imply the existence of QASs~\cite{myMSY24} (and therefore $\mathbf{SampBQP}\neq\mathbf{SampBPP}$). 
\cite{myKT24} shows that all of these assumptions imply \cref{def:QAA}, and that \cref{def:QAA} combined with $\mathbf{P^{\#P}} \nsubseteq \mathbf{(io)BQP}$ yields quantumly secure OWPuzzs. 
This suggests that proving these computational assumptions is at least as hard as constructing quantum cryptographic primitives based on worst-case hardness assumptions.

Our assumptions~(\cref{assumption:PP_hardness_of_Dis_learn}) are implied by \cref{def:QAA} together with $\mathbf{P^{\#P}} \nsubseteq \mathbf{(io)BQP^{\Sigma_3}}$. 
While our assumptions imply $\mathbf{SampBQP} \neq \mathbf{SampBPP}$, it remains unclear whether they also imply QASs.
\fi

\paragraph{Other characterizations of OWPuzzs.}
\cite{EC:CGGH25} and \cite{myHM24} showed the equivalence between the existence of
OWPuzzs and hardness of estimating Kolmogorov complexity.
These two papers and \cite{STOC:KT25} also showed the equivalence between the existence of OWPuzzs and
hardness of estimating quantum probability.

\subsection{Technical Overview}

In this subsection, we provide a high-level overview of the proofs of our main results.

\paragraph{Proof of \cref{thm:OWPuzzs_distlearn}.}
\if0
In the classical case, \cite{FOCS:HirNan23} shows the equivalence between OWFs and the average-case hardness of improper classical distribution learning.
One natural approach is to follow the classical approach.
Unfortunately, their technique of showing the hardness of improper distribution learning from OWFs depends on the classical case.
On the other hand, their construction of OWFs from  hardness of learning depends on the improper case.
learning, and it is unclear how to generalize it to the proper learning case~\footnote{In their proof, they assume that no OWFs exist, then show how to estimate classical probability on average.
From it, they construct a learning algorithm, which contradicts to average-case hardness of \emph{improper} distribution learning.
Crucially, their learning algorithm outputs an arbitrary PPT algorithm that needs to compute classical probability.
Therefore, their technique inherently can be applied to improper learning.}.
Therefore, we develop a new idea, and give a construction of OWPuzzs.
Our construction can be applied to the classical case as well, and we can construct OWFs from the average-case hardness of proper classical distribution learning.
\taiga{koko nakutemo yoi}
\fi
Let us first explain our proof of \cref{thm:OWPuzzs_distlearn}, which is the equivalence between
OWPuzzs and average-case hardness of proper quantum distribution learning.
We first explain how to construct OWPuzzs from average-case hardness of proper quantum distribution learning.
To show it, we first assume average-case hardness of proper quantum distribution learning. This means that there exist a QPT algorithm $\cS$ that samples a hidden parameter $z$, and a QPT algorithm $\cD$ that takes $z$ as input, and outputs a bit string $x$.
The hardness of learning requires that, for any polynomial $t$ and any QPT algorithm that receives $x_1, \dots, x_t$ 
independently sampled from $\cD(z)$, it is hard to find $z^*$ such that $\mathsf{SD}(\cD(z^*),\cD(z))$ is small.

We construct a OWPuzz, $(\Samp, \Vrfy)$, from such $\cS$ and $\cD$ as follows. $\Samp$ algorithm samples $z$ from $\cS$ and independently samples $x_1, \dots, x_t$ from $\cD(z)$. It then outputs $\ans \coloneqq z$ and $\puzz \coloneqq \{x_i\}_{i \in [t]}$.
Next, $\Vrfy$ algorithm takes $\ans^* = h$ and $\puzz = \{x_i\}_{i \in [t]}$ as input. It first computes $z^*$ such that
\begin{align}
\Pr[\{x_i\}_{i \in [t]} \leftarrow \cD(z^*)] = \max_{a} \Pr[\{x_i\}_{i \in [t]} \leftarrow \cD(a)].
\end{align}
Then, $\Vrfy$ outputs $1$ if $\cD(z^*)$ is statistically close to $\cD(h)$, and outputs $0$ if $\cD(z^*)$ is statistically far from $\cD(h)$.

We can show that, $\cD(z^*)$ and $\cD(z)$ are statistically close with high probability over $z\la\cS$ and $x_1, \dots, x_t \la\cD(z)^{\otimes t}$.
Therefore, in order to output $h$ such that $1 \leftarrow \Vrfy(\{x_i\}_{i \in [t]}, h)$, the adversary must produce $h$ such that $\cD(h)$ is close to $\cD(z)$. However, this is prohibited by the average-case hardness of proper quantum distribution learning. Therefore, the security of the OWPuzz follows.

Next, let us explain how to show 
average-case hardness of proper quantum distribution learning from OWPuzzs.
\cite{STOC:KhuTom24,C:ChuGolGra24} showed that OWPuzzs imply non-uniform QPRGs (nuQPRGs).
Therefore, it is sufficient to show average-case hardness of proper quantum distribution learning from nuQPRGs.
A nuQPRG is a QPT algorithm $\Gen$ that takes an advice string $\mu\in[n]$ as input, and outputs $n$-bit strings.
The security guarantees that there exists a $\mu^*\in[n]$ such that 
$\Gen(\mu^*)$ is statistically far but computationally indistinguishable from the uniform distribution $U_n$.

We show the average-case hardness of proper quantum distribution learning from nuQPRGs, $\Gen$.
We construct a pair $(\cS,\cD)$ of QPT algorithms as follows. 
$\cS$ samples $\mu\gets[n]$ and $b\gets\bit$,
and outputs $(\mu,b)$.
$\cD$ takes $(\mu,b)$ as input, and outputs $(\mu,x)$, where
$x\gets U_n$ if $b=1$, and $x\gets\Gen(\mu)$ if $b=0$.
For the sake of contradiction, suppose that $(\cS,\cD)$ is not average-case hard to learn. 
This means that there exists a QPT algorithm $\cA$, given samples $(\mu,x_1),...,(\mu,x_t)$ independently sampled from $\cD(\mu,b)$ (with $(\mu,b)\gets\cS$), 
that can output $(\mu^*,b^*)$ such that $\cD(\mu^*,b^*)$ is statistically close to $\cD(\mu,b)$.

Because the first output of $\cD(\mu,b)$ is $\mu$, the first output $\mu^*$ of $\cA$ has to be equal to $\mu$: otherwise,
$\mathsf{SD}(\cD(\mu,b),\cD(\mu^*,b^*))$ cannot be small.
Moreover, if $\mu$ is such that $\Gen(\mu)$ is statistically far from the uniform distribution,
$\cA$'s second output $b^*$ has to be equal to $b$, because otherwise
$\mathsf{SD}(\cD(\mu,b),\cD(\mu,b^*))$ cannot be small.

Therefore, from such $\cA$, we can construct a QPT algorithm $\cB$ that breaks the security of $\Gen$
as follows.
Let $\mathsf{Good}$ be the set of all $\mu\in[n]$ such that
$\Gen(\mu)$ is statistically far from $U_n$.
Our construction of $\cB$ is as follows: 
On input $x_1,...,x_t$ independently sampled from $\Gen(\mu^*)$ or $U_n$,
$\cB$ runs $\cA((\mu, x_1),...,(\mu,x_t))$ for all $\mu\in\mathsf{Good}$. 
$\cB$ outputs 1(0) if $\cA$ outputs 1(0) for all $\mu\in\mathsf{Good}$.
As we have explained above, $\cA$ outputs 1(0) for all $\mu\in\mathsf{Good}$ with high probability if
$x_1,...,x_t$ are sampled from $U_n$ ($\Gen$), and therefore,
$\cB$ breaks the security of $\Gen$.

However, one issue in the construction is that we do not know how to obtain $\mathsf{Good}$ in QPT.
We therefore define another set $\mathsf{Good}'$ that is the set of all $\mu\in[n]$
such that
$\Pr_{(y_1,...,y_t)\gets U_n^{\otimes t}}[1\gets\cA((\mu,y_1),...,(\mu,y_t))]$ is large.
Such a set can be obtained in QPT, and therefore $\cB$ can find it by itself.
It is easy to verify that $\mathsf{Good}\subseteq\mathsf{Good}'$.
Therefore, if we modify $\cB$ as follows, $\cB$ can correctly distinguish $\Gen$ and $U_n$:
$\cB$ first computes the set $\mathsf{Good}'$.
On input $x_1,...,x_t$ independently sampled from $\Gen(\mu^*)$ or $U_n$,
$\cB$ runs $\cA((\mu,x_1),...,(\mu,x_t))$ for all $\mu\in\mathsf{Good}'$.
$\cB$ outputs 1 if $\cA$ outputs 1 for all $\mu\in\mathsf{Good}'$.
$\cB$ outputs 0 if $\cA$ outputs 0 for at least one $\mu\in\mathsf{Good}'$.

\paragraph{Proof of \cref{thm:PP_agnostic}.}
%In the following, we explain how to show $\mathbf{BQP} \neq \mathbf{PP}$ from the hardness of agnostic quantum distribution learning with respect to KL divergence.
Let us next explain our proof of \cref{thm:PP_agnostic}, which is the equivalence between $\mathbf{BQP}\neq \mathbf{PP}$ and hardness of agnostic quantum distribution learning with respect to KL divergence.
Our key idea is to introduce the worst-case hardness of quantum maximum likelihood (QML) estimation as an intermediate object.
Let $\cD$ be a QPT algorithm that, on input a bit string $z$, outputs a bit string $x$.
Worst-case hardness of QML means that there exists a bit string $x$ such that
no QPT learner given $x$ can find $z^*$ such that
$\Pr[x\la\cD(z^*)]$ is close to $\max_{z}\Pr[x\la\cD(z)]$.
Intuitively, the worst-case hardness of QML is equivalent to the hardness of agnostic quantum distribution learning,
but we focus on the former, because it is more convenient to show the equivalence to $\mathbf{BQP}\neq\mathbf{PP}$.\footnote{Although
they are equivalent at the intuition level, showing the equivalence is non-trivial,
and we carefully show their equivalence, which is also our technical contribution.}

Hence our goal is to show the equivalence between the worst-case hardness of QML and $\mathbf{BQP}\neq\mathbf{PP}$.
\cite{HH24} already showed $\mathbf{BQP}\neq\mathbf{PP}$ from the worst-case hardness of QML. 
Therefore the remaining task is to show the other direction.

Let $\cL$ be a language in $\mathbf{PP}=\mathbf{PostBQP}$.
Then, there exists a QPT algorithm $\cM$ such that $\Pr_{(b,b^*)\la\cM(x)}[b=1|b^*=1]\geq \frac{2}{3}$ 
for all $x\in\cL$, and $\Pr_{(b,b^*)\la\cM(x)}[b=1|b^*=1]\leq \frac{1}{3}$ for all $x\not\in\cL$.
For the QPT algorithm $\cM$, we consider $\cM^*(x,c)$ that runs $(b,b^*)\la\cM(x)$, and outputs $x$ if $b=c$ and $b^*=1$, and outputs $\bot$ otherwise.
Crucially, if $x\in\cL$, then $\Pr[x\la\cM^*(x,1)]$ is much larger than $\Pr[x\la\cM^*(x,0)]$ while if $x\not\in\cL$, then $\Pr[x\la\cM^*(x,0)]$ is much larger than $\Pr[x\la\cM^*(x,1)]$.
Therefore, if QML is easy in the worst-case setting, we can check if $x\in\cL$ or $x\not\in\cL$ in QPT by doing the maximum likelihood estimation of $\cM^*$.
%Since $\mathbf{PostBQP}=\mathbf{PP}$~\cite{myAaronson05}, this completes the proof.

\paragraph{Proof of \cref{thm:classical_upperbound}}
\cref{thm:classical_upperbound} is the construction of a PPT algorithm with access to a $\Sigma_3^{\mathbf{P}}$ oracle that can
achieve agnostic classical distribution learning with respect to the statistical distance.
Let $\cT$ be an unknown algorithm that outputs classical bit strings, and let $\cD$ be a known QPT algorithm that takes an $n$-bit string $z$ as input, and outputs a classical bit string $x$.
The learner receives $x_1,...,x_t$ independently sampled from $\cT$, and has to find $z^*$ such that $\mathsf{SD}(\cD(z^*),\cT)$ is small.

In order to make our explanation simpler, let us first assume that $\cD$'s input $z$ is just a single bit.
In this setting, access to an $\mathbf{NP}$ oracle, not a $\Sigma_3^{\mathbf{P}}$ oracle, is enough.
Given $x_1,...,x_t$ independently sampled from $\cT$, the learner has to find $b\in\bit$ such that $\cT$ is statistically closer to $\cD(b)$
than $\cD(b\oplus1)$.
How can we do that?
By the definition of the statistical distance, there exists an unbounded algorithm $\mathsf{Dis}$ such that
\begin{align}
\abs{\Pr[1\la\mathsf{Dis}(x):x\la\cD(1)]-\Pr[1\la\mathsf{Dis}(x):x\la\cD(0)]}=\mathsf{SD}(\cD(0),\cD(1)).
\end{align}
Given $x_1,...,x_t$ as input,
which is independently sampled from $\cT$,
the learner first estimates the value of 
$\Pr[1\la\mathsf{Dis}(x):x\la\cT]$. 
It then outputs 0 (1) if
the estimated value of $\Pr[1\la\mathsf{Dis}(x):x\la\cT]$ is 
closer to that of 
$\Pr[1\la\mathsf{Dis}(x):x\la\cD(0)]$ ($\Pr[1\la\mathsf{Dis}(x):x\la\cD(1)]$).
The reason why this learning algorithm works can be easily understood from the inequality
\begin{align}
\abs{\Pr[1\la\mathsf{Dis}(x):x\la\cT]-\Pr[1\la\mathsf{Dis}(x):x\la\cD(b)]}\le\mathsf{SD}(\cT,\cD(b)).
\end{align}
This learning algorithm can be implemented with access to an $\mathbf{NP}$ oracle,
because 
$\mathsf{Dis}$ is the algorithm that takes $x$ as input and outputs $b$ if $x$ is more likely from $\cD(b)$,
and such an algorithm can be implemented with the Stockmeyer theorem~\cite{STOC:Stockmeyer83}:
a PPT algorithm with access to an $\mathbf{NP}$ oracle can approximate classical probability distributions.
Hence, in this simplified setting, a PPT algorithm with access to an $\mathbf{NP}$ oracle can achieve agnostic distribution learning.

Next let us consider the actual case where $\cD$'s input $z$ is an $n$-bit string. 
Again, from the definition of the statistical distance, 
there exists an unbounded algorithm $\mathsf{Dis}$ such that
\begin{align}
\abs{\Pr[1\la\mathsf{Dis}(a,b,x):x\la\cD(a)]-\Pr[1\la\mathsf{Dis}(a,b,x):x\la\cD(b)]}=\mathsf{SD}(\cD(a),\cD(b))
\end{align}
for all $a$ and $b$.
Given $x_1,...,x_t$ independently sampled from $\cT$, the learner first estimates the value of
$\Pr[1\la\mathsf{Dis}(a,b,x):x\la\cT]$ for all $a,b$.  
The learner outputs $a^*$ such that
\begin{align}
\abs{\Pr[1\la\mathsf{Dis}(a^*,b,x):x\la\cT] - \Pr[1\la\mathsf{Dis}(a^*,b,x):x\la\cD(a^*)]}
\end{align}
is small for all $b$.
The reason why this learning algorithm works can be understood from the following inequality:
\begin{align}
\mathsf{SD}(\cT,\cD(a))
&\leq \mathsf{SD}(\cT,\cD(b)) + \mathsf{SD}(\cD(a),\cD(b))\\
&=  \mathsf{SD}(\cT,\cD(b))+ \abs{\Pr[1\la\mathsf{Dis}(a,b,x):x\la\cD(b)] - \Pr[1\la\mathsf{Dis}(a,b,x):x\la\cD(a)]}\\
&\leq \mathsf{SD}(\cT,\cD(b)) 
+\abs{\Pr[1\la\mathsf{Dis}(a,b,x):x\la\cD(b)] 
-
\Pr[1\la\mathsf{Dis}(a,b,x):x\la\cT]}\\
&\,\,\,\,\,+\abs{\Pr[1\la\mathsf{Dis}(a,b,x):x\la\cT] 
-
\Pr[1\la\mathsf{Dis}(a,b,x):x\la\cD(a)]}\\
&\le 2\mathsf{SD}(\cT,\cD(b))
+\abs{\Pr[1\la\mathsf{Dis}(a,b,x):x\la\cT] 
-
\Pr[1\la\mathsf{Dis}(a,b,x):x\la\cD(a)]},
\end{align}
which suggests that in order to minimize $\mathsf{SD}(\cT,\cD(a))$, 
\begin{align}
\abs{\Pr[1\la\mathsf{Dis}(a,b,x):x\la\cT] 
-
\Pr[1\la\mathsf{Dis}(a,b,x):x\la\cD(a)]}
\end{align}
has to be minimized over $a$.

Finally, with a careful investigation, we show that the learning algorithm can be implemented with access to a $\Sigma_3^{\mathbf{P}}$ oracle.\footnote{We do not know whether $\Sigma_3^{\mathbf{P}}$ is optimal. This could be improved.}

\paragraph{Proof of \cref{thm:qad}.}
In the following, we explain our proof of \cref{thm:qad},
which is $\mathbf{SampBQP}\neq\mathbf{SampBPP}$ from hardness of agnostic quantum distribution learning with respect to the statistical distance
against $PPT^{\Sigma_3^{\mathbf{P}}}$ learner.
Let $\cT$ be an unknown algorithm that outputs classical bit strings, 
and let $\cD$ be a known QPT algorithm that takes a bit string $z$ as input and outputs a classical bit string $x$. The learner who receives $x\gets\cT$
has to find $z^*$ such that $\cD(z^*)$ is statistically close to $\cT$.  

For the sake of contradiction, assume that
$\mathbf{SampBQP} = \mathbf{SampBPP}$.  
Then, there exists a PPT algorithm $\cC$ such that $\mathsf{SD}(\cD(z),\cC(z))$ is small for all $z$.  
Hence, given samples $x_1, \ldots, x_t$ from $\cT$, it is sufficient to find $z^*$ such that $\cC(z^*)$ is statistically close to $\cT$,  
which is 
agnostic {\it classical} distribution learning.
From \cref{thm:classical_upperbound},
it is possible
with access to a $\Sigma_3^{\mathbf{P}}$ oracle.

\if0
\mor{kokomo soudan}
We explain how to show \cref{thm:PP_worstdistlearn}.
\cref{thm:PP_worstdistlearn} states that, if we can show the $\mathbf{PP}$-hardness of worst-case hardness of proper quantum distribution learning and $\mathbf{PP}\nsubseteq \mathbf{BPP}^{\Sigma_3^{\mathbf{P}}}$, then we have $\mathbf{SampBQP}\neq\mathbf{SampBPP}$.
For showing \cref{thm:PP_worstdistlearn}, it is sufficient to show that a PPT algorithm with the $\Sigma_3^{\mathbf{P}}$ oracle can solve agnostic classical distribution learning with respect to the statistical distance.
The reason is as follows:
If $\mathbf{SampBQP}=\mathbf{SampBPP}$, then for any QPT algorithm $\cD$, which takes $z$ and outputs $x$, there exists a PPT algorithm $\cD^*$ such that $\cD^*(z)$ is close to $\cD(z)$ for all $z$.
Therefore, if we can show that a PPT algorithm with the $\Sigma_3^{\mathbf{P}}$ oracle can solve agnostic classical distribution learning with respect to the statistical distance, then $\mathbf{SampBQP}\neq\mathbf{SampBPP}$ implies that a PPT algorithm with the $\Sigma_3^{\mathbf{P}}$ oracle can solve worst-case proper quantum distribution learning as well.
On the other hand, if we can show the $\mathbf{PP}$-hardness of worst-case proper quantum distribution learning, then no PPT algorithms with the $\Sigma_3^{\mathbf{P}}$ oracle can solve worst-case proper quantum distribution learning assuming $\mathbf{PP\neq \mathbf{BPP}}^{\Sigma_3^\mathbf{P}}$.
This is a contradiction.
\fi

\paragraph{Proof of \cref{thm:PP_worstdistlearn}.}
\cref{thm:PP_worstdistlearn}
states that
if the worst-case hardness of proper quantum distribution learning is $\mathbf{PP}$-hard (via a black-box reduction), 
then $\mathbf{PP} \nsubseteq \mathbf{BPP}^{\Sigma_3^{\mathbf{P}}}$ implies $\mathbf{SampBQP} \neq \mathbf{SampBPP}$. 
In order to show it,
assume that the worst-case hardness of proper quantum distribution learning can be reduced to the $\mathbf{PP}$-hardness.
Let us also assume that $\mathbf{SampBQP} = \mathbf{SampBPP}$. 
It is straightforward that the worst-case hardness of proper quantum distribution learning implies
the hardness of agnostic quantum distribution learning with respect to the statistical distance.
From the assumption of 
$\mathbf{SampBQP} = \mathbf{SampBPP}$, 
the latter implies
the hardness of agnostic {\it classical} distribution learning with respect to the statistical distance.
From \cref{thm:classical_upperbound}, this is broken by a PPT algorithm with access to a $\Sigma_3^{\mathbf{P}}$ oracle.
Therefore, $\mathbf{PP}\subseteq \mathbf{BPP}^{\Sigma_3^{\mathbf{P}}}$.

\section{Preliminaries}
\subsection{Basic Notations}
We introduce basic notations and mathematical tools used in this paper.

We use the standard notations of cryptography and quantum information.
$[n]$ denotes the set $\{1,2,...,n\}$.
$\negl$ is a negligible function.
$\poly$ is a polynomial.
QPT stands for quantum polynomial time,
and PPT stands for classical probabilistic polynomial time.
For an algorithm (or a Turing machine) $\cA$, $y\gets\cA(x)$ means that $\cA$ runs on input $x$ and outputs $y$.
When we explicitly show that $\cA$ uses randomness $r$, we write $y\la \cA(x;r)$.
The notation $\{y_i\}_{i\in[N]}\la\cA^{\otimes N}(x)$ means that $\cA$ is run on input $x$ independently $N$ times, and $y_i$ is the $i$th result.
For any distribution $D$, $a\gets D$ means that $a$ is sampled according to the distribution $D$.
For any set $S$, $a\gets S$ means that $a$ is sampled uniformly at random from $S$.
We denote $U_n$ to mean a uniform distribution over $\bit^n$.
For two distributions $D$ and $E$, we write
$\mathsf{SD}(D,E)$ or $\mathsf{SD}(x_{x\la D}, x_{x\la E})$ to mean the statistical distance between two distributions $D$ and $E$, and we write $D_{KL}(D\,\|\,E)\seteq  \sum_{a}\Pr[a\la D]\log\left(\frac{\Pr[a\la D]}{\Pr[a\la E]} \right)$ to mean the KL-divergence from $D$ to $E$.
If $A$ is an algorithm that outputs bit strings, we often denote $A$ by the output distribution of $A$ when it is clear from the context.

\subsection{Lemmas}
We introduce lemmas, which we use in this paper.

\if0
\begin{lemma}\label{lem:Hoeffding}
    Let $n\in\N$ and $M\in\R$.
    Let $\cF\coloneqq\{
    f:\bit^{\poly(n)}\ra [0,M]
    \}$ be a family of functions. 
    Let $\epsilon,\delta>0$, and
    \begin{align}
        T\geq \frac{M^2}{\epsilon^2}\abs{\log{\abs{\cF}}+\log(1/\delta)}.
    \end{align}
    Then, for any distribution $\cD$ over $\cZ$ with probability at least $1-\delta$ over $\{z_i\}_{i\in[T]}\la \cD^{\otimes T}$, we have
    \begin{align}
        \abs{\frac{1}{m}\sum_{i\in[T]}X(z_i)-  \mathbb{E}_{z\la\cD}[X(z)] }\leq \epsilon
    \end{align}
    for all $X\in\cF$ at the same time.
    \mor{@koko yoku wakaran}
\end{lemma}
\fi

\begin{lemma}\label{lem:Hoeffding}
    Let $M,\epsilon,\delta>0$ and let
    $\cF_M\coloneqq\{
    f:\bit^{*}\ra [0,M]
    \}$ be a family of functions such that $|\cF_M|$ is finite. 
   Let $T$ be an integer such that
    \begin{align}
        T\geq \frac{M^2}{\epsilon^2}\abs{\log{\abs{\cF_M}}+\log(1/\delta)}.
    \end{align}
     Let $\cD$ be a distribution over $\bit^*$.
     We have
    \begin{align}
        \Pr\left[\abs{\frac{1}{T}\sum_{i\in[T]}f(x_i)-  \mathbb{E}_{x\la\cD}[f(x)] }\leq \epsilon \mbox{\,\,\,for\, all\,\,} f\in\cF_M:\{x_i\}_{i\in[T]}\la\cD^{\otimes T} \right]\geq 1-\delta.
    \end{align}
    Here, $\{z_i\}_{i\in[T]}\gets\cD^{\otimes T}$ means that each $z_i$ is independently sampled from $\cD$. 
\end{lemma}
\cref{lem:Hoeffding} directly follows from Hoeffding's inequality and the union bound.
Therefore, we omit the proof.

The following theorem guarantees that for all PPT algorithms $\cD$, which takes $1^n$ as input, and outputs a classical string $x$, there exists a PPT algorithm 
with access to an $\mathbf{NP}$ oracle that can estimate $\Pr[x\la\cD(1^n)]$ within $1/\poly$-multiplicative error for all $x\in\bit^*$ and for all sufficiently large $n\in\N$.
\begin{theorem}[\cite{STOC:Stockmeyer83}]\label{claim:stockmeyer}
For any polynomial $\epsilon$, and any PPT algorithm $\cD$ which takes $1^n$ as input, and outputs $x\in\bit^{\poly(n)}$, there exists a PPT algorithm $\mathsf{Estimate}$ with access to an $\mathbf{NP}$ oracle such that, for all $x\in\bit^n$, we have
\begin{align}
\Pr[
\abs{\mathsf{Estimate}^{\mathbf{NP}}(1^n,x)-\Pr[x\la\cD(1^n)]}\leq\frac{\Pr[x\la\cD(1^n)]}{\epsilon(n))}]\geq 1-1/\epsilon(n)
\end{align}
for all sufficiently large $n\in\N$.
\end{theorem}

In this paper, we will use the following \cref{thm:estimate}, which generalizes \cref{claim:stockmeyer}.
The difference is that, in \cref{thm:estimate}, the PPT algorithm $\cD$ takes a classical bit string $z\in\bit^n$ in addition to the security parameter $1^n$
as input.

\begin{theorem}\label{thm:estimate}
    For any polynomial $\epsilon$, and for any PPT algorithm $\cD$ that takes $1^n$ and $z\in\bit^n$ as input, and outputs $x\in\bit^{\poly(n)}$, there exists a PPT algorithm $\mathsf{Estimate}$ with access to an $\mathbf{NP}$ oracle such that, for all $(z,x)\in\bit^n\times\bit^{\poly(n)}$, we have
    \begin{align}
       \Pr[
       \abs{\mathsf{Estimate}(1^n,z,x)-\Pr[x\la\cD(1^n,z)] }\leq \frac{\Pr[x\la\cD(1^n,z)]}{\epsilon(n)} ]\geq 1-1/\epsilon(n) 
    \end{align}
    for all sufficiently large $n\in\N$.
\end{theorem}

\begin{proof}[Proof of \cref{thm:estimate}]

For an arbitrary PPT algorithm $\cD$, which takes $1^n$ and $z\in\bit^n$ as input, and outputs $x\in\bit^{\poly(n)}$, we consider another PPT algorithm $\cD^*$ that takes $1^n$ as input, and uniformly randomly samples $z\la\bit^n$ and outputs $(z,\cD(1^n,z))$.
From the construction of $\cD^*$, we have $\Pr[x\la\cD(1^n,z)]=2^n\Pr[(z,x)\la\cD^*(1^n)]$.
On the other hand, from \cref{claim:stockmeyer}, there exists a PPT algorithm with access to an $\mathbf{NP}$ oracle that can estimate $\Pr[(z,x)\la\cD^*(1^n)]$ for all $x\in\bit^{\poly(n)}$ and $z\in\bit^n$.
Therefore, a PPT algorithm with access to an $\mathbf{NP}$ oracle can estimate $\Pr[x\la\cD(1^n,z)]$ for all $x\in\bit^{\poly(n)}$ and $z\in\bit^n$.
\end{proof}

\subsection{Complexity Classes}
We introduce basic complexity classes, which we consider in this paper.

\begin{definition}[\cite{myAaronson05}]
$\mathbf{PostBQP}$ is the class of language $\cL\subseteq \bit^*$ for which there exist a polynomial $p$, and a QPT algorithm $\cM$, which takes $x\in\bit^*$ as input and outputs $(b,b^*)\in\bit\times \bit$ such that for all $x\in\bit^*$, the following three conditios are satisfied:
\begin{enumerate}
\item $\Pr[b^*=1:(b,b^*)\la\cM(x)]\geq 2^{-p(\abs{x})}$.
\item If $x\in\cL$, then we have $\Pr[b=1:(b,1)\la\cM(x)]\geq \frac{2}{3}$.
\item If $x \notin\cL$, then we have $\Pr[b=1:(b,1)\la\cM(x)]\leq \frac{1}{3}$.
\end{enumerate}
\end{definition}

\begin{definition}
$\Sigma_3^{\mathbf{P}}$ is the class of language $\cL\subseteq\bit^*$ for which there exist polynomials $p$, $q$, and $r$, and a deterministic polynomial time algorithm $\cM$, which takes $x\in\bit^*$, $a\in\bit^{p(\abs{x})}$, $b\in\bit^{q(\abs{x})}$, and $c\in\bit^{r(\abs{x})}$ such that, for all 
$x\in\bit^*$, the following two conditions are satisfied:
\begin{enumerate}
\item If $x\in\cL$, then there exists $a\in\bit^{p(\abs{x})}$, such that for all $b\in\bit^{q(\abs{x})}$, there exists $c\in\bit^{r(\abs{x})}$ such that we have $\cM(x,a,b,c)=1$.
\item If $x\notin\cL$, then for all $a\in\bit^{p(\abs{x})}$, there exists $b\in\bit^{q(\abs{x})}$ such that for all $c\in\bit^{r\abs{x}}$, $\cM(x,a,b,c)=0$.
\end{enumerate}
\end{definition}

\begin{definition}
$\Sigma_2^{\mathbf{PP}}$ is the class of language $\cL\subseteq\bit^*$ for which there exist polynomials $p$ and $q$, and a deterministic polynomial time algorithm $\cM$ with access to a $\mathbf{PP}$ oracle, which takes $x\in\bit^*$, $a\in\bit^{p(\abs{x})}$, and $b\in\bit^{q(\abs{x})}$ such that for all 
$x\in\bit^*$, the following two conditions are satisfied:
\begin{enumerate}
\item If $x\in\cL$, then there exists $a\in\bit^{p(\abs{x})}$, such that for all $b\in\bit^{q(\abs{x})}$, we have $\cM^{\mathbf{PP}}(x,a,b)=1$.
\item If $x\notin\cL$, then for all $a\in\bit^{p(\abs{x})}$, there exists $b\in\bit^{q(\abs{x})}$ such that $\cM^{\mathbf{PP}}(x,a,b)=0$.
\end{enumerate}
\end{definition}

%\paragraph{Sampling Complexity}

\begin{definition}[Sampling Problems~\cite{myAaronson14,myABK24}]
A sampling problem  is a collection of probability distributions $\{D_{x}\}_{x\in\bit^*}$, where each $D_x$ is a probability distribution over $\bit^{p(x)}$ for some fixed polynomial $p$. 
\end{definition}

\begin{definition}[SampBQP~\cite{myAaronson14,myABK24}]
$\mathbf{SampBQP}$ is a class of sampling problems $\{D_x\}_{x\in\bit^*}$ for which there exists a uniform QPT algorithm $\cQ$ such that for all $x\in\bit^*$ and for all $\epsilon>0$, $\mathsf{SD}(\cQ(1^{\lfloor 1/\epsilon\rfloor},x),D_x)\leq \epsilon$.
\end{definition}

\begin{definition}[SampBPP~\cite{myAaronson14,myABK24}]
$\mathbf{SampBPP}$ is a class of sampling problems $\{D_x\}_{x\in\bit^*}$ for which there exists a uniform PPT algorithm $\cC$ such that for all $x\in\bit^*$ and for all $\epsilon>0$, $\mathsf{SD}(\cC(1^{\lfloor 1/\epsilon\rfloor},x),D_x)\leq \epsilon$.
\end{definition}

\subsection{Quantum Advantage Assumption}
We introduce quantum advantage assumption.

\begin{assumption}[Quantum Advantage Assumption~\cite{STOC:AarArk11,BreMonShe16,STOC:KT25}]\label{def:QAA}
We say that quantum advantage assumption holds if both of the following two conditions are satisfied:
\begin{enumerate}
\item There exists a family $\cC=\{\cC_n\}_{n\in\N}$ of distributions such that for each $n\in\N$, $\cC_n$ is a (uniform) QPT sampleable distribution over quantum circuits $C$, which takes $1^n$ as input, and outputs $n$-bit classical bit strings.
\item
There exist polynomials $p$ and $\gamma$ such that:
\begin{enumerate}
\item For all sufficiently large $n\in\N$,
\begin{align}
\Pr\left[\Pr[x\la C(1^n)]\geq \frac{1}{p(n)2^n}:C\la\cC_n,x\la\bit^n\right]\geq\frac{1}{\gamma(n)}.
\end{align}
\item For any oracle $\mathcal{O}$ satisfying that for all sufficiently large $n\in\N$,
\begin{align}
\Pr\left[\abs{\mathcal{O}(C,x)-\Pr[x\la C(1^n)]}\leq\frac{\Pr[x\la C(1^n)]}{p(n)}:C\la\cC_n,x\la\bit^n \right]\geq \frac{1}{\gamma(n)}-\frac{1}{p(n)},
\end{align}
we have that $\mathbf{P^{\#P}}\subseteq \mathbf{BPP}^{\mathcal{O}}$.
\end{enumerate}
\end{enumerate}
\end{assumption}

\subsection{Cryptographic Primitives}
We introduce cryptographic primitives used in this paper.

\begin{definition}[One-Way Functions (OWFs)]\label{def:OWFs}
    A function $f:\bit^*\to\bit^*$ that is computable in classical deterministic polynomial-time is a classically-secure one-way function (OWF) if for any PPT adversary $\cA$,
    \begin{align}
    \Pr[f(x')=f(x): x\gets\bit^{n}, x'\gets\cA(1^{n},f(x))] \le\negl(n).
    \end{align}
\end{definition}

\begin{definition}[Classical One-Way Puzzles (OWPuzzs)~\cite{STOC:KhuTom24}]
Let $p$ be a polynomial.
A classical one-way puzzle (OWPuzz) with $\left(1-\frac{1}{p(n)}\right)$-security is a pair $(\Samp, \Ver)$ of algorithms with the following syntax:
\begin{itemize}
    \item $\Samp(1^n) \rightarrow (\ans,\puzz)$: It is a PPT algorithm that, on input $1^n$, outputs two classical bit strings $(\ans,\puzz)$. 
    \item $\Ver(\ans',\puzz) \rightarrow \top/\bot$: It is an unbounded algorithm that, on input $(\ans',\puzz)$, outputs $\top/\bot$.
\end{itemize}
We require the following correctness and security.

\paragraph{Correctness:}
\begin{align}
\Pr [\top\gets\Ver(\ans,\puzz):(\ans,\puzz)\leftarrow \Samp(1^n)] \ge 1 - \negl(n).
\end{align}

\paragraph{Security:}
For any uniform PPT algorithm $\cA$,
\begin{align}
\Pr[\top\gets\Ver(\cA(1^n,\puzz),\puzz):(\ans,\puzz) \leftarrow \Samp(1^n)] \le 1-\frac{1}{p(n)}.
\end{align}
\end{definition}
We call OWPuzzs with $\negl(n)$ security as OWPuzzs for simplicity.

\begin{theorem}[\cite{C:ChuGolGra24}]\label{thm:amplification_c_OWPuzz}
For any polynomial $p$, classical OWPuzzs exist if and only if classical OWPuzzs with $\left(1-\frac{1}{p(n)}\right)$-security exist.
\end{theorem}

It is shown that classical OWPuzzs exist if and only if OWFs exist.
\begin{theorem}[\cite{STOC:KhuTom24}]\label{thm:classical_OWPuzz}
Classical OWPuzzs exist if and only if OWFs exist.
\end{theorem}

\begin{definition}[One-Way Puzzles (OWPuzzs) \cite{STOC:KhuTom24}]
\label{def:OWPuzz}
Let $p$ be a polynomial.
A one-way puzzle (OWPuzz) with $\left(1-\frac{1}{p(n)}\right)$-security is a pair $(\Samp, \Ver)$ of algorithms with the following syntax:
\begin{itemize}
    \item $\Samp(1^n) \rightarrow (\ans,\puzz)$: It is a QPT algorithm that, on input $1^n$, outputs two classical bit strings $(\ans,\puzz)$. 
    \item $\Ver(\ans',\puzz) \rightarrow \top/\bot$: It is an unbounded algorithm that, on input $(\ans',\puzz)$, outputs $\top/\bot$.
\end{itemize}
We require the following correctness and security.

\paragraph{Correctness:}
\begin{align}
\Pr [\top\gets\Ver(\ans,\puzz):(\ans,\puzz)\leftarrow \Samp(1^n)] \ge 1 - \negl(n).
\end{align}

\paragraph{Security:}
For any uniform QPT algorithm $\cA$,
\begin{align}
\Pr[\top\gets\Ver(\cA(1^n,\puzz),\puzz):(\ans,\puzz) \leftarrow \Samp(1^n)] \le 1-\frac{1}{p(n)}.
\end{align}
\end{definition}
We call OWPuzzs with $\negl(n)$ security as OWPuzzs for simplicity.
\begin{theorem}[\cite{C:ChuGolGra24}]\label{thm:amplification_OWPuzz}
For any polynomial $p$, OWPuzzs exist if and only if OWPuzzs with $\left(1-\frac{1}{p(n)}\right)$-security exist.
\end{theorem}

\if0
\taiga{Remove this}
\color{red}
\begin{definition}[Distributional OWPuzzs~\cite{C:ChuGolGra24}]
A distributional one-way puzzle is a uniform QPT algorithm $\Samp$, which takes $1^n$ as input and outputs $(\puzz,\ans)$ that satisfies the following.
There exists a polynomial $p$ such that for all uniform QPT algorithms $\cA$, 
\begin{align}
\mathsf{SD}\left((\puzz,\cA(\puzz))_{(\puzz,\ans)\la\Samp(1^n)},(\puzz,\ans)_{(\puzz,\ans)\la\Samp(1^n)}\right)\ge \frac{1}{p(n)}
\end{align}
for all sufficiently large $n\in\N$
\end{definition}

It is known that OWPuzzs are existentially equivalent to distributional OWPuzzs.
\begin{theorem}[\cite{C:ChuGolGra24}]\label{thm:DisOWPuzz_OWPuzz}
OWPuzzs exist if and only if distributional OWPuzzs exist.    
\end{theorem}
\color{black}
\fi

%We also introduce a non-uniform version of QEFID.

\begin{definition}[non-uniform QPRGs (nuQPRGs)] 
A non-uniform QPRG is a QPT algorithm $\Gen(1^n,\mu)$ that takes a security parameter $1^n$ and $\mu\in [n]$ as input, and outputs a classical 
bit string $x\in\bit^n$.
We require that there exists $\mu^*\in[n]$ such that the following two conditions hold:
    \paragraph{Statistically far:} 
    \begin{align}
        \mathsf{SD}(\Gen(1^n,\mu^*),U_n)\geq 1-\negl(n).
    \end{align}
    Here, $U_n$ is the uniform distribution over $\bit^n$.
    \paragraph{Computationally indistinguishable:} 
    For any QPT algorithm $\cA$ and any polynomial $t$,
    \begin{align}
        &\Bigg|\Pr[1\la\cA(1^n,\{x_i\}_{i\in[t(n)]}):\{x_i\}_{i\in [t(n)]}\la\Gen(1^n,\mu^*)^{\otimes t(n)}]\\
        &\hspace{2cm}-\Pr[1\la\cA(1^n,\{x_i\}_{i\in[t(n)]}):\{x_i\}_{i\in[t(n)]}\la U_n^{\otimes t(n)}]\Bigg|\leq\negl(n).
    \end{align}
\end{definition}

\if0
\begin{definition}[Non-Uniform QEFID (nuQEFID)] 
A non-uniform QEFID is a QPT algorithm $\Gen(1^n,\mu,b)$ that takes the security parameter $1^n$, $\mu\in [n]$, and $b\in\bit$ as input, and outputs a classical 
bit string $x\in\bit^n$.
We require that there exists $\mu^*\in[n]$ such that the following two conditions hold:
    \paragraph{Statistically far:} 
    \begin{align}
        \mathsf{SD}(\Gen(1^n,\mu^*,0),\Gen(1^n,\mu^*,1))\geq 1-\negl(n).
    \end{align}
    \paragraph{Computationally indistinguishable:} 
    For any uniform QPT algorithm $\cA$,
    \begin{align}
        \abs{\Pr[1\la\cA(1^n,x):x\la\Gen(1^n,\mu^*,0)]-\Pr[1\la\cA(1^n,x):x\la \Gen(1^n,\mu^*,1)]}\leq\negl(n).
    \end{align}
\end{definition}
\fi

It is known that OWPuzzs are existentially equivalent to nuQPRGs.
\begin{theorem}[\cite{STOC:KhuTom24,C:ChuGolGra24}]\label{thm:QEFID_OWPuzz}
OWPuzzs exist if and only if nuQPRGs exist.
\end{theorem}

\section{Average-Case Hardness of Proper Distribution Learning}\label{Sec:Average}

\subsection{Quantum Case}

We introduce average-case hardness of proper quantum distribution learning, and show that it is equivalent to the existence of OWPuzss.
We first define average-case hardness of proper quantum distribution learning.
Our formalism is based on \cite{STOC:KMRRSS94,FOCS:HirNan23}.

\begin{definition}[Average-Case Hardness of Proper Quantum Distribution Learning]\label{def:q_avg_learn}
We say that average-case hardness of proper quantum distribution learning holds if the following holds.
There exist some polynomials $\epsilon$ and $\delta$, a QPT algorithm $\cS$, which takes $1^n$ as input, and outputs $z\in\bit^n$, and a QPT algorithm $\cD$, which takes $1^n$ and $z\in\bit^n$ as input, and outputs $x\in\bit^{\poly(n)}$, such that, for all QPT algorithms $\cA$ and any polynomial $t$,
\begin{align}
\label{ineq:q_avg_learn}
\Pr[ 
\mathsf{SD}(\cD(1^n,z),\cD(1^n,h))\leq 1/\epsilon(n):
\begin{array}{ll}
&z\la\cS(1^n)\\
&\{x_i\}_{i\in[t(n)]}\la\cD(1^n,z)^{\otimes t(n)} \\
& h\la \cA(1^n,\{x_i\}_{i\in[t(n)]})
\end{array}
]\leq 1-1/\delta(n)
\end{align}
for all sufficiently large $n\in\N$.
Here, $\{x_i\}_{i\in [t(n)]}\gets\cD(1^n,z)^{\otimes t(n)}$ means that $\cD(1^n,z)$ is run $t(n)$ times, and $x_i$ is the $i$th output.
\end{definition}
\begin{remark}
In the above definition, we require that \cref{ineq:q_avg_learn} holds \emph{for all sufficiently large security parameters} instead of \emph{for infinitely many security parameters}.
This is because we want to show its equivalence to OWPuzzs instead of infinitely-often OWPuzzs.
If we use the ``for infinitely many security parameters'' security, we obtain the equivalence to infinitely-often OWPuzzs
in a similar proof.
\end{remark}

%Previos definition
\if0
\begin{definition}[Average-Case Hardness of Proper Quantum Distribution Learning]
There exists a QPT algorithm $\cS$, which takes as input $1^n$, and outputs $z\in\bit^n$, and a QPT algorithm $\cD$, which takes as input $1^n$ and $z\in\bit^n$, and outputs $x\in\bit^n$, such that, for all QPT algorithms $\cA$ and all polynomials $T$, there exists some $\epsilon>0$ and $\delta>0$ such that
\begin{align}
\Pr_{z\la\cS(1^n)}\left[\Pr[ 
\mathsf{SD}(\cD(1^n,z),\cD(1^n,h))\leq \epsilon:
\begin{array}{ll}
&\{x_i\}_{i\in[T]}\la\cD(1^n,z)^{\otimes T} \\
& h\la \cA(1^n,1^{\lfloor 1/\epsilon\rfloor},1^{\lfloor1/\delta\rfloor},\{x_i\}_{i\in[T]})
\end{array}
]\leq 1-\delta\right]\geq 1-\delta
\end{align}
for all sufficiently large $n\in\N$.
\end{definition}
\fi

\begin{theorem}\label{thm:OWPuzz}
The following two conditions are equivalent:
\begin{enumerate}
\item OWPuzzs exist.
\item Average-case hardness of proper quantum distribution learning holds.
\end{enumerate}
\end{theorem}

\cref{thm:OWPuzz} follows from the following \cref{lem:OWPuzz,lem:OWPuzz_to_Learn_avg}.
\begin{lemma}\label{lem:OWPuzz}
If average-case hardness of proper quantum distribution learning holds, then OWPuzzs exist.
\end{lemma}

\begin{lemma}\label{lem:OWPuzz_to_Learn_avg}
If OWPuzzs exist, then average-case hardness of proper quantum distribution learning holds.
\end{lemma}

In the following, we give the proofs.

\begin{proof}[Proof of \cref{lem:OWPuzz}]
From \cref{thm:amplification_OWPuzz}, it is sufficient to construct OWPuzzs with $\left(1-\frac{1}{p(n)}\right)$ security for some polynomial $p$.

Assume that average-case hardness of proper quantum distribution learning holds.
Then, there exist polynomials $\epsilon$ and $\delta$, a QPT algorithm $\cS$, which takes $1^n$ as input, and outputs $z\in\bit^n$, 
and a QPT algorithm $\cD$, which takes $1^n$, and $z\in\bit^n$ as input, and outputs $x\in\bit^{\poly(n)}$ 
such that for any polynomial $t$, and any QPT algorithm $\cA$, 
\begin{align}\label{ineq_learn_A}
\Pr[ 
\mathsf{SD}(\cD(1^n,z),\cD(1^n,h))\leq 1/\epsilon(n):
\begin{array}{ll}
&z\la\cS(1^n)\\
&\{x_i\}_{i\in[t(n)]}\la\cD(1^n,z)^{\otimes t(n)} \\
& h\la \cA(1^n,\{x_i\}_{i\in[t(n)]})
\end{array}
]\leq 1-1/\delta(n)
\end{align}
for all sufficiently large $n\in\N$.
From $\cS$ and $\cD$, we construct OWPuzzs $(\Samp,\Vrfy)$ with $1-\frac{1}{9\delta(n)^2}$ security.

Before describing the construction of OWPuzzs, let us introduce notations and a claim, which we will use for showing correcntess and security.
Let $t(n)\seteq 16\epsilon(n)^2\cdot n$.
Let us denote $\mathsf{Eval}$ to mean a deterministic unbounded-time algorithm that takes $1^n$, and $\{x_i\}_{i\in[t(n)]}$ as input, and outputs an arbitrary $z^*$ such that
\begin{align}
\Pr[\{x_i\}_{i\in[t(n)]}\la\cD(1^n,z^*)^{\otimes t(n)}]
=\max_{a\in\bit^n}\left\{\Pr[\{x_i\}_{i\in[t(n)]}\la \cD(1^n,a)^{\otimes t(n)}]\right\}.
\end{align}
\begin{claim}\label{claim:useful}
For all $z\in\bit^n$,
we have
\begin{align}
&\Pr\left[
\mathsf{SD}(\cD(1^n,z),\cD(1^n,z^*))\leq \frac{1}{2\epsilon(n)}:
\begin{array}{ll}
     & \{x_i\}_{i\in[t(n)]}\la \cD(1^n,z)^{\otimes t(n)} \\
     & z^*\la\mathsf{Eval}(1^n,\{x_i\}_{i\in[t(n)]})
\end{array}
\right]\geq 1-2^{-n+1}
\end{align}
for all $n\in\N$.
\end{claim}
We defer the proof of \cref{claim:useful} to the end of the proof.
Our construction of OWPuzzs $(\Samp,\Vrfy)$ with $1-\frac{1}{9\delta(n)^2}$ security is as follows.
\begin{description}
\item[$\Samp(1^n)$:] $ $
\begin{enumerate}
\item Run $z\la\cS(1^n)$.
\item Run $\{x_i\}_{i\in[t(n)]}\la\cD(1^n,z)^{\otimes t(n)}$.
\item Output $\puzz\coloneqq\{x_i\}_{i\in[t(n)]}$ and $\ans\coloneqq z$.
\end{enumerate}
\item[$\Vrfy(\puzz,\ans)$:]$ $
\begin{enumerate}
\item Parse $\puzz=\{x_i\}_{i\in[t(n)]}$ and $\ans=z$.
\item Compute $z^*\la\mathsf{Eval}(1^n,\{x_i\}_{i\in[t(n)]})$.
\item Output $\top$ if
\begin{align}
\mathsf{SD}(\cD(1^n,z^*),\cD(1^n,z))\leq \frac{3}{2\epsilon(n)}.
\end{align}
Otherwise, output $\bot$.
\end{enumerate}
\end{description}

\paragraph{Correctness.}
First, let us show the correctness.
From the construction of $\Vrfy(\{x_i\}_{i\in[t(n)]},z)$, it runs $z^*\la\mathsf{Eval}(1^n,\{x_i\}_{i\in[t(n)]})$,
computes $\mathsf{SD}(\cD(1^n,z^*),\cD(1^n,z))$, and outputs $\top$ if $\mathsf{SD}(\cD(1^n,z^*),\cD(1^n,z))\leq \frac{3}{2\epsilon(n)}$, and outputs $\bot$ otherwise.
On the other hand, from \cref{claim:useful}, we have
\begin{align}
\Pr[ 
\mathsf{SD}(\cD(1^n,z^*),\cD(1^n,z))\leq 1/2\epsilon(n):
\begin{array}{ll}
&z\la\cS(1^n)\\
&\{x_i\}_{i\in[t(n)]}\la\cD(1^n,z)^{\otimes t(n)} \\
& z^*\la \mathsf{Eval}(1^n,\{x_i\}_{i\in[t(n)]})
\end{array}
]\geq 1-2^{-n+1}
\end{align}
for all $n\in\N$.
This implies that we have
\begin{align}
\Pr\left[\top\la\Vrfy(\puzz,\ans):(\puzz,\ans)\la\Samp(1^n)\right]\geq 1-\negl(n).
\end{align}

\paragraph{Security.}
Next, let us show the security.
For contradiction, let us assume that the OWPuzz is not secure.
This means that there exists a QPT algorithm $\cB$ such that
\begin{align}
\Pr[\top\la\Vrfy(\puzz,\ans^*):
\begin{array}{ll}
&(\puzz,\ans)\la\Samp(1^n)\\
& \ans^*\la\cB(1^n,\puzz)
\end{array}
]\geq 1-\frac{1}{9\delta^2(n)}
\end{align}
for infinitely many $n\in\N$.
Note that from the standard average argument, this implies that with probability 
at least $1-\frac{1}{3\delta(n)}$ over $z\la\cS(1^n)$ and $\{x_i\}_{i\in[t(n)]}\la\cD(1^n,z)^{\otimes t(n)}$, we have
\begin{align}\label{ineq:B}
\Pr[ \top\la\Vrfy(\{x_i\}_{i\in[t(n)]},h):
 h\la\cB(1^n,\{x_i\}_{i\in[t(n)]})
]\geq 1-\frac{1}{3\delta(n)}
\end{align}
for infinitely many $n\in\N$.

From $\cB$, we construct a QPT algorithm $\cA$ that contradicts to \cref{ineq_learn_A}.
$\cA$ receives $1^n$ and $\{x_i\}_{i\in[t(n)]}$, which is sampled by $\{x_i\}_{i\in[t(n)]}\la\cD(1^n,z)^{\otimes t(n)}$, where $z\la\cS(1^n)$, runs $h\la\cB(1^n,\{x_i\}_{i\in[t(n)]})$, and outputs $h$. 
In the following, we analyze the $\cA$.

From \cref{claim:useful}, with probability at least $1-2^{-n+1}$ over 
$\{x_i\}_{i\in[t(n)]}\la\cD(1^n,z)^{\otimes t(n)}$, 
we have
%\begin{align}\label{ineq:Eval}
%\Pr_{z^*\la\mathsf{Eval}(1^n,\{x_i\}_{i\in[T]})}\left[\mathsf{SD}(\cD(1^n,z),\cD(1^n,z^*))\leq \frac{1}{2\epsilon(n)} \right]=1.
%\end{align}
\begin{align}\label{ineq:Eval}
\mathsf{SD}(\cD(1^n,z),\cD(1^n,z^*))\leq \frac{1}{2\epsilon(n)},
\end{align}
where $z^*$ is the output of $\mathsf{Eval}(1^n,\{x_i\}_{i\in[t(n)]})$.
Let $\mathsf{Good}_{n,z,\epsilon}$ be a set of $\{x_i\}_{i\in[t(n)]}$ that satisfies the following two conditions:
\begin{enumerate}
    \item  
    The output $z^*$ of $\mathsf{Eval}(1^n,\{x_i\}_{i\in[t(n)]})$ satisfies
    \begin{align}
    \mathsf{SD}(\cD(1^n,z),\cD(1^n,z^*))\leq \frac{1}{2\epsilon(n)}.
    \end{align}
    \item
    \begin{align}
    \Pr[ \top\la\Vrfy(\{x_i\}_{i\in[t(n)]},h):
     h\la\cB(1^n,\{x_i\}_{i\in[t(n)]})
    ]\geq 1-\frac{1}{3\delta(n)}.
    \end{align}
\end{enumerate}
From the union bound, \cref{ineq:Eval,ineq:B} imply that 
\begin{align}
\Pr[\{x_i\}_{i\in[t(n)]}\in\mathsf{Good}_{n,z,\epsilon}:
\begin{array}{ll}
     & z\gets\mathcal{S}(1^n) \\
     & \{x_i\}_{i\in[t(n)]}\gets\cD(1^n,z)^{\otimes t(n)}
\end{array}
]\geq1-\frac{2}{3\delta(n)}.\label{ineq:Good}
\end{align}

We will use the following \cref{claim:Good}, which we defer the proof of \cref{claim:Good} to the end of the proof.
\begin{claim}\label{claim:Good}
For all $\{x_i\}_{i\in[t(n)]}\in\mathsf{Good}_{n,z,\epsilon}$,
if we have
\begin{align}
\top\la\Vrfy(\{x_i\}_{i\in[t(n)]},h),     
\end{align}
then $h$ satisfies
\begin{align}
\mathsf{SD}(\cD(1^n,h),\cD(1^n,z))\leq  \frac{1}{\epsilon(n)}.
\end{align}
\end{claim}

From \cref{claim:Good}, for all $\{x_i\}_{i\in[t(n)]}\in\mathsf{Good}_{n,z,\epsilon}$, we have 
\begin{align}
    &\Pr[ \mathsf{SD}(\cD(1^n,z),\cD(1^n,h))\leq \frac{1}{\epsilon(n)}:
    \begin{array}{ll}
     h\la\cB(1^n,\{x_i\}_{i\in[t(n)]})
    \end{array}
    ]\label{ineq:Vrfy_B}\\
    &\geq\Pr[ \top\la\Vrfy(\{x_i\}_{i\in[t(n)]},h):
    \begin{array}{ll}
     h\la\cB(1^n,\{x_i\}_{i\in[t(n)]})
    \end{array}
    ]\geq 1-\frac{1}{3\delta(n)}.
\end{align}
Here, in the final inequality, we have used the definition of $\mathsf{Good}_{n,z,\epsilon}$.

\cref{ineq:Vrfy_B,ineq:Good} imply that
\begin{align}
\Pr[ 
\mathsf{SD}(\cD(1^n,z),\cD(1^n,h))\leq \frac{1}{\epsilon(n)}:
\begin{array}{ll}
&z\la\cS(1^n)\\
&\{x_i\}_{i\in[t(n)]}\la\cD(1^n,z)^{\otimes t(n)} \\
& h\la \cA(1^n,\{x_i\}_{i\in[t(n)]})
\end{array}
]\geq 1-\frac{1}{\delta(n)}
\end{align}
for infinitely many $n\in\N$.
This completes the proof.

Finally, we prove \cref{claim:useful,claim:Good}.
For showing \cref{claim:useful}, we will use the following \cref{claim:probabilistic_argument,claim:statistical_distance}.
\begin{claim}\label{claim:probabilistic_argument}
If 
\begin{align}
\mathsf{SD}(\cD(1^n,z), \cD(1^n,z^*)) > \frac{1}{2\epsilon(n)},
\end{align}
then we have
\begin{align}
\mathsf{SD}(\cD(1^n,z)^{\otimes t(n)},\cD(1^n,z^*)^{\otimes t(n)})> 1-2^{-\frac{t(n)}{8\epsilon(n)^2} +1}.
\end{align}
\end{claim}
\begin{claim}\label{claim:statistical_distance}
If $\mathsf{SD}(D(1^n,z),\cD(1^n,z^*))> \alpha$, then
\begin{align}
\Pr[\Pr[x\la\cD(1^n,z)]> \Pr[x\la\cD(1^n,z^*)] : x\la\cD(1^n,z) ]> \alpha.
\end{align}
\end{claim}
\cref{claim:probabilistic_argument,claim:statistical_distance} follows from standard probabilistic argument.
For clarity, we describe the proof of \cref{claim:probabilistic_argument,claim:statistical_distance} in the end of the proof.

\begin{proof}[Proof of \cref{claim:useful}]
Define
\begin{align}
\mathsf{Bad}_{n,z,\epsilon}\seteq\left\{z^*:\mathsf{SD}(\cD(1^n,z),\cD(1^n,z^*))> \frac{1}{2\epsilon(n)}\right\}.    
\end{align}
From \cref{claim:probabilistic_argument}, for any $z^*\in\mathsf{Bad}_{n,z,\epsilon}$, we have 
\begin{align}
\mathsf{SD}(\cD(1^n,z)^{\otimes t(n)},\cD(1^n,z^*)^{\otimes t(n)})> 1- 2^{-\frac{t(n)}{8\epsilon(n)^2}+1}.
\end{align}
From \cref{claim:statistical_distance}, for any $z^*\in\mathsf{Bad}_{n,z,\epsilon}$, 
we have
\begin{align}
&\Pr[\Pr[\{x_i\}_{i\in[t(n)]}\la\cD(1^n,z)^{\otimes t(n)}]>\Pr[\{x_i\}_{i\in[t(n)]}\la\cD(1^n,z^*)^{\otimes t(n)}]:\{x_i\}_{i\in[t(n)]}\la \cD(1^n,z)^{\otimes t(n)}]\\
&>
1- 2^{-\frac{t(n)}{8\epsilon(n)^2}+1}.
\end{align}
This implies that, for any $z^*\in\mathsf{Bad}_{n,z,\epsilon}$, we have 
\begin{align}
\Pr[z^*\la\mathsf{Eval}(1^n,\{x_i\}_{i\in[t(n)]}):\{x_i\}_{i\in[t(n)]}\la \cD(1^n,z)^{\otimes t(n)}]\leq 2^{-\frac{t(n)}{8\epsilon(n)^2}+1}
\end{align}
\if0
This implies that for any $z^*\in\mathsf{Bad}_{n,z,\epsilon}$, with probability at least $1-2^{-\frac{t(n)}{8\epsilon(n)^2}+1}$, over $\{x_i\}_{i\in[t(n)]}\la\cD(1^n,z)^{\otimes t(n)}$, we have
\begin{align}
\Pr[\{x_i\}_{i\in[t(n)]}\la\cD(1^n,z^*)^{\otimes t(n)}]\neq \max_{a\in\bit^n} \Pr[\{x_i\}_{i\in[t(n)]}\la\cD(1^n,a)^{\otimes t(n)}].
\end{align}
\fi
Therefore, we have
\begin{align}
&\sum_{z^*\in \mathsf{Bad}_{n,z,\epsilon}} \Pr[z^*\la\mathsf{Eval}(1^n,\{x_i\}_{i\in[t(n)]}):\{x_i\}_{i\in[t(n)]}\gets\cD(1^n,z)^{\otimes t(n)}]\leq\sum_{h\in \mathsf{Bad}_{n,z,\epsilon}}2^{-\frac{t(n)}{8\epsilon(n)^2}+1}\\
&\leq \abs{\mathsf{Bad}_{n,z,\epsilon}} 2^{-\frac{t(n)}{8\epsilon(n)^2}+1}\leq 2^{n-\frac{t(n)}{8\epsilon(n)^2}+1}\leq 2^{-n+1}.
\end{align}
Therefore, we have
\begin{align}
\Pr[\mathsf{SD}(\cD(1^n,z^*),\cD(1^n,z))\leq\frac{1}{2\epsilon(n)}:
\begin{array}{ll}
     & \{x_i\}_{i\in[t(n)]}\la\cD(1^n,z)^{\otimes t(n)}\\
     & z^*\la\mathsf{Eval}(1^n,\{x_i\}_{i\in[t(n)]})
\end{array}
]\geq 1-2^{-n+1}
\end{align}
for all $n\in\N$.
\end{proof}

\begin{proof}[Proof of \cref{claim:Good}]
If $\{x_i\}_{i\in [t(n)]}\in \mathsf{Good}_{n,z,\epsilon}$, then, by definition, $\mathsf{Eval}(1^n,\{x_i\}_{i\in[t(n)]})$ outputs $z^*$ such that
\begin{align}
\mathsf{SD}(\cD(1^n,z^*),\cD(1^n,z))< \frac{1}{2\epsilon(n)}
\end{align}
with probability $1$, and
from the definition of $\Vrfy$, $\Vrfy(\{x_i\}_{i\in[t(n)]},h)$ outputs $\bot$ for all $h$ such that
\begin{align}
    \mathsf{SD}(\cD(1^n,h),\cD(1^n,z^*))\geq \frac{3}{2\epsilon(n)},
\end{align}
where $z^*$ is the output of $\mathsf{Eval}(1^n,\{x_i\}_{i\in [t(n)]})$.
From the triangle inequality, this implies that $\Vrfy(\{x_i\}_{i\in[t(n)]},h)$ outputs $\bot$ for all $h$ such that
\begin{align}
\mathsf{SD}(\cD(1^n,h),\cD(1^n,z))\geq \mathsf{SD}(\cD(1^n,h),\cD(1^n,z^*))-\mathsf{SD}(\cD(1^n,z^*),\cD(1^n,z))> \frac{1}{\epsilon(n)}.
\end{align}
Therefore, for all $\{x_i\}_{i\in[t(n)]}\in\mathsf{Good}_{n,z,\epsilon}$, if $\Vrfy(\{x_i\}_{i\in[t(n)]},h)$ outputs $\top$, then we have
\begin{align}
\mathsf{SD}(\cD(1^n,h),\cD(1^n,z))\leq \frac{1}{\epsilon(n)}.
\end{align}
\end{proof}

\begin{proof}[Proof of \cref{claim:probabilistic_argument}]
Let 
\begin{align}
I_n\coloneqq\mathsf{SD}(\cD(1^n,z),\cD(1^n,z^*)).
\end{align}
Let $\cG_{n}$ be a set of $x\in\bit^*$ such that
\begin{align}
    \Pr[x\la\cD(1^n,z)]>\Pr[x\la\cD(1^n,z^*)].
\end{align}
Let $f$ be a function such that
\begin{align}
f(x)=
\left\{
\begin{array}{ll}
1& \mbox{\,\, if\,\, } x\in\cG_n\\
0& \mbox{\,\, if\,\, } x\notin\cG_n.
\end{array}
\right.
\end{align}
From Hoeffding inequality, we have
\begin{align}\label{ineq:D_z}
\Pr\left[\sum_{i\in[t(n)]}\frac{f(x_i)}{t(n)}\leq \Pr[x\in\cG_n :x\la\cD(1^n,z)]- \frac{I_n}{2} :\{x_i\}_{i\in[t(n)]} \la \cD(1^n,z)^{\otimes t(n)}\right]\leq 2^{-\frac{t(n)\cdot I_n^2}{2}}.
\end{align}
From Hoeffding inequality, we also have
\begin{align}
2^{-\frac{t(n)\cdot I_n^2}{2}}
&\geq 
\Pr\left[\sum_{i\in[t(n)]}\frac{f(x_i)}{t(n)} \geq\Pr[x\in\cG_n :x\la\cD(1^n,z^*)] + \frac{I_n}{2}:\{x_i\}_{i\in[t(n)]} \la \cD(1^n,z^*)^{\otimes t(n)}\right]\\
&=\Pr\left[\sum_{i\in[t(n)]}\frac{f(x_i)}{t(n)} \geq\Pr[x\in\cG_n :x\la\cD(1^n,z)] - \frac{I_n}{2}:\{x_i\}_{i\in[t(n)]} \la \cD(1^n,z^*)^{\otimes t(n)}\right].
\end{align}
This implies that
\begin{align}\label{ineq:D_z^*}
\Pr\left[\sum_{i\in[t(n)]}\frac{f(x_i)}{t(n)} \leq\Pr[x\in\cG_n :x\la\cD(1^n,z)] - \frac{I_n}{2}:\{x_i\}_{i\in[t(n)]} \la \cD(1^n,z^*)^{\otimes t(n)}\right]\geq 1-2^{-\frac{t(n)\cdot I_n^2 }{2}}.
\end{align}

\cref{ineq:D_z,ineq:D_z^*} imply that
\begin{align}
&\mathsf{SD}(\cD(1^n,z)^{\otimes t(n)},\cD(1^n,z^*)^{\otimes t(n)})\geq 1-2\cdot 2^{-\frac{t(n)\cdot I_n^2}{2}}.
\end{align}
Because $I_n=\mathsf{SD}(\cD(1^n,z),\cD(1^n,z^*))>\frac{1}{2\epsilon(n)}$, we have
\begin{align}
\mathsf{SD}(\cD(1^n,z)^{\otimes t(n)},\cD(1^n,z^*)^{\otimes t(n)})\geq1-2\cdot 2^{-\frac{t(n)\cdot I_n^2}{2}}> 1-2\cdot 2^{-\frac{t(n)}{8\epsilon(n)^2}}.
\end{align}

\if0
Let $\cG_{n}$ be a set of $x\in\bit^*$ such that
\begin{align}
    \Pr[x\la\cD(1^n,z)]>\Pr[x\la\cD(1^n,z^*)].
\end{align}
Let 
\begin{align}
p\seteq \Pr[x\in\cG_n:x\la\cD(1^n,z)]
\end{align}
and let $\cG_{*,n}^{T(n)}$ be a set of $\{x_i\}_{i \in [T(n)]}$ such that at least $(p - \frac{1}{4\epsilon(n)})T(n)$\mor{Is this always integer?} of the elements $x_i$ are contained in $\cG_n$.

From Hoeffding inequality, we have\mor{humei}
\begin{align}
\Pr[ \{x_i\}_{i\in[T(n)]}\in\cG_{*,n}^{T(n)}:\{x_i\}_{i\in[T(n)]}\la \cD(1^n,z)^{\otimes T(n)}]\geq 1-2^{-\frac{T(n)}{8\epsilon(n)^2}}    
\end{align}
and
\begin{align}
\Pr[\{x_i\}_{i\in[T(n)]}\in\cG_{*,n}^{T(n)}:\{x_i\}_{i\in[T(n)]}\la \cD(1^n,z^*)^{\otimes T(n)}]\leq 2^{-\frac{T(n)}{8\epsilon(n)^2}}.
\end{align}

This implies that
\begin{align}
\mathsf{SD}(\cD(1^n,z)^{\otimes T(n)},\cD(1^n,z^*)^{\otimes T(n)})
\geq 1-2^{-\frac{T(n)}{8\epsilon(n)^2}+1}.
\end{align}
\fi
\end{proof}

\begin{proof}[Proof of \cref{claim:statistical_distance}]
We have 
\begin{align}
&\Pr\left[\Pr[x\la\cD(1^n,z)] >\Pr[x\la\cD(1^n,z^*)]:x\la\cD(1^n,z) \right]\\
&=\sum_{x :\Pr[x\la\cD(1^n,z)] >\Pr[x\la\cD(1^n,z^*)] }\Pr[x\la\cD(1^n,z)]\\
&=\sum_{x :\Pr[x\la\cD(1^n,z)] >\Pr[x\la\cD(1^n,z^*)] }\left(\Pr[x\la\cD(1^n,z)]-\Pr[x\la\cD(1^n,z^*)]\right)\\
&\hspace{2cm}+\sum_{x :\Pr[x\la\cD(1^n,z)] >\Pr[x\la\cD(1^n,z^*)] } \Pr[x\la\cD(1^n,z^*)]\\
&\geq \sum_{x :\Pr[x\la\cD(1^n,z)] >\Pr[x\la\cD(1^n,z^*)] }\left(\Pr[x\la\cD(1^n,z)]-\Pr[x\la\cD(1^n,z^*)]\right)\\
&=\mathsf{SD}(\cD(1^n,z),\cD(1^n,z^*) )>\alpha.
\end{align}

\end{proof}

\end{proof}

\begin{proof}[Proof of \cref{lem:OWPuzz_to_Learn_avg}]
Suppose that OWPuzzs exist.
Then, from \cref{thm:QEFID_OWPuzz}, there exists a non-uniform QPRG, $\Gen$.
From $\Gen$, we construct $\cS$ and $\cD$, which is hard to learn on average, as follows.
\begin{description}
\item[$\cS(1^n)$:]$ $
\begin{enumerate}
\item Sample $\mu\la [n]$.
\item Sample $b\la\bit$.
\item Output $z\seteq (\mu,b)$.
\end{enumerate}
\item[$\cD(1^n,z)$:] $ $
\begin{enumerate}
\item Parse $z=(\mu,b)$.
\item If $b=0$, then run $\xi\la\Gen(1^n,\mu)$, and output $x\seteq (\xi,\mu)$.
\item If $b=1$, then run $\xi\la\bit^n$, and output $x\seteq (\xi,\mu)$.
\end{enumerate}
\end{description}
\if0
Let $\mu\in[n]$ be an advice such that we have 
\begin{align}
\mathsf{SD}((x)_{x\la\Gen(1^n,\mu,0)},(x)_{x\la \Gen(1^n,\mu,1)})\geq 1-\negl(n)
\end{align}
and   
\begin{align}
\abs{\Pr_{x\la\Gen(1^n,\mu,0)}[1\la\cB(1^n,x)]-\Pr_{x\la\Gen(1^n,\mu,1)}[1\la\cB(1^n,x)]}\leq\negl(n)
\end{align}
for any QPT algorithm $\cB$.
\fi

For contradiction, suppose that there exist a QPT algorithm $\cA$ and a polynomial $t$ such that
\begin{align}\label{ineq_3}
\Pr\left[ 
\mathsf{SD}(\cD(1^n,(\mu,b)),\cD(1^n,(\mu^*,b^*)))\leq 1/n:
\begin{array}{ll}
&(\mu,b)\la\cS(1^n)\\
&\{\xi_i,\mu\}_{i\in[t(n)]}\la\cD(1^n,(\mu,b))^{\otimes t(n)} \\
& (\mu^*,b^*)\la \cA(1^n,\{\xi_i\}_{i\in[t(n)]},\mu)
\end{array}
\right]\geq 1-\frac{1}{n^{100}}
\end{align}
for infinitely many $n\in\N$.

From $\cA$, we construct a QPT algorithm $\cB$ that breaks the security of $\Gen$. 
More formally, we construct a QPT algorithm $\cB$ such that, for all $\mu\in[n]$ with
\begin{align}
    \mathsf{SD}(\Gen(1^n,\mu),U_n)\geq 1-\negl(n),
\end{align}
we have
\begin{align}
&\Bigg|\Pr[1\la\cB(1^n,\{\xi_i\}_{i\in[t(n)]}):\{\xi_i\}_{i\in[t(n)]}\la\Gen(1^n,\mu)^{\otimes t(n)}]\\
&\hspace{2cm}-\Pr[1\la\cB(1^n,\{\xi_i\}_{i\in[t(n)]}):\{\xi_i\}_{i\in[t(n)]}\la U_n^{\otimes t(n)}]
\Bigg|\geq \frac{1}{\poly(n)}
\end{align}
for infinitely many $n\in\N$.
Here, $U_n$ is a uniform distribution over $\bit^n$.

For describing the algorithm $\cB$, let us introduce the following $\mathsf{Check}$ algorithm.
\begin{description}
\item[$\mathsf{Check}(1^n,\mu)$:] $ $
\begin{enumerate}
\item Receive $1^n$ and $\mu\in[n]$ as input.
%\item Run $\{X_i(j)\}_{i\in[t(n)]}\la \Gen(1^n,\mu)^{\otimes t(n)}$ for all $j\in[n^{25}]$. 
\item Run $\{X_i(j)\}_{i\in[t(n)]}\la U_n^{\otimes t(n)}$ for all $j\in[n^{25}]$.
\item Run $(\mu',b_{\mu}(j) )\la\cA(1^n,\{X_i(j)\}_{i\in[t(n)]},\mu)$ for all $j\in[n^{25}]$.
%\item 
%Run $(\mu',\widetilde{b_{\mu}}(j) )\la\cA(1^n,\{Y_i(j)\}_{j\in[t(n)]},\mu )$ for all $j\in[n^{25}]$.
\item If $b_{\mu}(j)=1$ for all $j\in[n^{25}]$, then output $\top$.
Otherwise, output $\bot$.
\end{enumerate}
\end{description}

Now, we give the construction of $\cB$.
\begin{description}
\item[$\cB(1^n,\{\xi_i\}_{i\in[t(n)]})$:]$ $
\begin{enumerate}
\item Receive $1^n$ and $\{\xi_i\}_{i\in[t(n)]}$ as input.
\item For all $\mu\in[n]$, run $\mathsf{flag}_{\mu}\la\mathsf{Check}(1^n,\mu)$.
\item For all $\mu\in[n]$, 
run $(\mu^*,b_\mu)\la\cA(1^n,\{\xi_i\}_{i\in[t(n)]},\mu)$.
\item Output $0$ if $b_\mu=0$ and $\mathsf{flag}_{\mu}=\top$ for some $\mu\in[n]$.
Otherwise, output $1$.
%Output $1$ if $b_\mu=1$ for all $\mu\in[n]$ such that $\mathsf{flag}_{\mu}=\top$.
\end{enumerate}
\end{description}

We have the following \cref{claim:B_far,claim:B_uniform}.
\begin{claim}\label{claim:B_far}
For all $\mu\in[n]$ such that
\begin{align}
    \mathsf{SD}(\Gen(1^n,\mu),U_n)\geq 1-\negl(n),
\end{align}
we have
\begin{align}
\Pr[0\la\cB(1^n,\{\xi_i\}_{i\in[t(n)]}):\{\xi_i\}_{i\in[t(n)]}\la\Gen(1^n,\mu)^{\otimes t(n)}]\geq 1-\frac{1}{n}
\end{align}
for infinitely many $n\in\N$.
\end{claim}

\begin{claim}\label{claim:B_uniform}
We have
\begin{align}
\Pr[1\la\cB(1^n,\{\xi_i\}_{i\in[t(n)]}):\{\xi_i\}_{i\in[t(n)]}\la U_n^{\otimes t(n)}]\geq 1-\frac{1}{n}
\end{align}
for infinitely many $n\in\N$.
\end{claim}
We defer the proof of \cref{claim:B_far,claim:B_uniform}.
From \cref{claim:B_far,claim:B_uniform}, for all $\mu\in[n]$ such that
\begin{align}
    \mathsf{SD}(\Gen(1^n,\mu),U_n)\geq 1-\negl(n),
\end{align}
we have
\begin{align}
&\Pr[1\la\cB(1^n,\{\xi_i\}_{i\in[t(n)]}):\{\xi_i\}_{i\in[t(n)]}\la U_n^{\otimes t(n)}]\\
&\hspace{2cm}-\Pr[1\la\cB(1^n,\{\xi_i\}_{i\in[t(n)]}):\{\xi_i\}_{i\in[t(n)]}\la \Gen(1^n,\mu)^{\otimes t(n)}]\geq 1-\frac{2}{n}
\end{align}
for infinitely many $n\in\N$.
This contradicts the security of $\Gen$.

Now, we prove \cref{claim:B_far,claim:B_uniform}.
For showing them, we will use the following \cref{claim:check,claim:markov_EFID}.

\begin{claim}\label{claim:check}
For all sufficiently large $n\in\N$,
if $\mu\in[n]$ satisfies 
\if0
either
\begin{align}
\Pr[b^*=0:
\begin{array}{ll}
     & \{\xi_i\}_{i\in [t(n)]}\la \Gen(1^n,\mu)^{\otimes t(n)} \\
     & (\mu^*,b^*)\la\cA(1^n,\mu)
\end{array}]
\leq 1-\frac{1}{n^5}
\end{align}
or
\fi
\begin{align}
\Pr[b^*=1:
\begin{array}{ll}
     & \{\xi_i\}_{i\in [t(n)]}\la U_n^{\otimes t(n)} \\
     & (\mu^*,b^*)\la\cA(1^n,\{\xi_i\}_{i\in [t(n)]},\mu)
\end{array}]
\leq 1-\frac{1}{n^5},
\end{align}
then 
\begin{align}
\Pr[\top\la\mathsf{Check}(1^n,\mu)]\leq \frac{1}{n^{10}}.
\end{align}
\end{claim}

\begin{claim}\label{claim:markov_EFID}
For all $\mu\in[n]$ such that
\begin{align}
\mathsf{SD}(\Gen(1^n,\mu),U_n)\geq 1-\negl(n),
\end{align}
we have
\begin{align}\label{ineq_4}
\Pr\left[ 
b'=0:
\begin{array}{ll}
&\{\xi_i\}_{i\in[t(n)]}\la\Gen(1^n,\mu)^{\otimes t(n)} \\
&(\mu',b')\la \cA(1^n,\{\xi_i\}_{i\in[t(n)]},\mu)
\end{array}
\right]\geq 1-\frac{1}{n^{80}}
\end{align}
and 
\begin{align}
\Pr\left[ 
b'=1:
\begin{array}{ll}
&\{\xi_i\}_{i\in[t(n)]}\la U_n^{\otimes t(n)} \\
&(\mu',b')\la \cA(1^n,\{\xi_i\}_{i\in[t(n)]},\mu)
\end{array}
\right]\geq 1-\frac{1}{n^{80}}
\end{align}
for infinitely many $n\in\N$.
\end{claim}
We defer the proof of \cref{claim:check,claim:markov_EFID} in the end of the proof.

\begin{proof}[Proof of \cref{claim:B_far}]
From the construction of $\cB$, we have
\begin{align}
&\Pr[0\la\cB(1^n,\{\xi_i\}_{i\in[t(n)]}):\{\xi_i\}_{i\in[t(n)]}\la\Gen(1^n,\mu)^{\otimes t(n)}]\\
&=\Pr[\left(\mathsf{flag}_{\mu^*}=\top \wedge b_{\mu^*}=0\right) \mbox{for some }\mu^*\in[n] :
\begin{array}{ll}
& \{\xi_i\}_{i\in[t(n)]}\la\Gen(1^n,\mu)^{\otimes t(n)}\\
& \mathsf{flag}_{\mu^*}\la \mathsf{Check}(1^n,\mu^*) \mbox{ for all  } \mu^*\in[n] \\
&  (\mu',b_{\mu^*})\la\cA(1^n,\{\xi_i\}_{i\in[t(n)]},\mu^*) \mbox{ for all }\mu^*\in[n]
\end{array}
]\\
&\geq 
\Pr[\left(\mathsf{flag}_{\mu}=\top \wedge b_{\mu}=0\right)  :
\begin{array}{ll}
& \{\xi_i\}_{i\in[t(n)]}\la\Gen(1^n,\mu)^{\otimes t(n)}\\
& \mathsf{flag}_{\mu}\la \mathsf{Check}(1^n,\mu) \\
&  (\mu',b_{\mu})\la\cA(1^n,\{\xi_i\}_{i\in[t(n)]},\mu)
\end{array}
]\\
&=\Pr[ b_{\mu}=0  :
\begin{array}{ll}
& \{\xi_i\}_{i\in[t(n)]}\la\Gen(1^n,\mu)^{\otimes t(n)}\\
&  (\mu',b_{\mu})\la\cA(1^n,\{\xi_i\}_{i\in[t(n)]},\mu)
\end{array}
]\cdot \Pr[\top\la\mathsf{Check}(1^n,\mu)]\\
&=\Pr[ b_{\mu}=0  :
\begin{array}{ll}
& \{\xi_i\}_{i\in[t(n)]}\la\Gen(1^n,\mu)^{\otimes t(n)}\\
&  (\mu',b_{\mu})\la\cA(1^n,\{\xi_i\}_{i\in[t(n)]},\mu)
\end{array}
]\\
&\hspace{1cm}\cdot 
\left(
\Pr[ b_{\mu}=1  :
\begin{array}{ll}
& \{\xi_i\}_{i\in[t(n)]}\la U_n^{\otimes t(n)}\\
&  (\mu',b_{\mu})\la\cA(1^n,\{\xi_i\}_{i\in[t(n)]},\mu)
\end{array}
]
\right)^{n^{25}}
\end{align}
On the other hand, from \cref{claim:markov_EFID}, for all $\mu$ such that
\begin{align}
\mathsf{SD}(\Gen(1^n,\mu),U_n )\geq 1-\negl(n),
\end{align}
we have
\begin{align}
\Pr\left[ 
b'=0:
\begin{array}{ll}
&\{\xi_i\}_{i\in[t(n)]}\la\Gen(1^n,\mu)^{\otimes t(n)} \\
&(\mu',b')\la \cA(1^n,\{\xi_i\}_{i\in[t(n)]},\mu)
\end{array}
\right]\geq 1-\frac{1}{n^{80}}
\end{align}
and 
\begin{align}
\Pr\left[ 
b'=1:
\begin{array}{ll}
&\{\xi_i\}_{i\in[t(n)]}\la U_n^{\otimes t(n)} \\
&(\mu',b')\la \cA(1^n,\{\xi_i\}_{i\in[t(n)]},\mu)
\end{array}
\right]\geq 1-\frac{1}{n^{80}}
\end{align}
for infinitely many $n\in\N$.
Therefore,
for all $\mu$ such that 
\begin{align}
\mathsf{SD}(\Gen(1^n,\mu),U_n )\geq 1-\negl(n),
\end{align}
we have
\begin{align}
&\Pr[0\la\cB(1^n,\{\xi_i\}_{i\in[t(n)]}):\{\xi_i\}_{i\in[t(n)]}\la\Gen(1^n,\mu)^{\otimes t(n)}]\geq\left(1-\frac{1}{n^{80}}\right)^{n^{26}}\geq 1-\frac{1}{n}
\end{align}
for infinitely many $n\in\N$.
\end{proof}

\begin{proof}[Proof of \cref{claim:B_uniform}]

For any $n\in\N$, let $\mathsf{Bad}_n$ be a set of $\mu$ such that the following inequality  holds
\if0
\begin{align}\label{Bad_1}
\Pr[b^*=0:
\begin{array}{ll}
     & \{\xi_i\}_{i\in [t(n)]}\la \Gen(1^n,\mu)^{\otimes t(n)} \\
     & (\mu^*,b^*)\la\cA(1^n,\mu)
\end{array}]
\leq 1-\frac{1}{n^5}
\end{align}
\fi
\begin{align}\label{Bad_2}
\Pr[b^*=1:
\begin{array}{ll}
     & \{\xi_i\}_{i\in [t(n)]}\la U_n^{\otimes t(n)} \\
     & (\mu^*,b^*)\la\cA(1^n,\{\xi_i\}_{i\in [t(n)]},\mu)
\end{array}]
\leq 1-\frac{1}{n^5}.
\end{align}

From the construction of $\cB$, we have
\begin{align}
&\Pr[0\la\cB(1^n,\{\xi_i\}_{i\in[t(n)]}):\{\xi_i\}_{i\in[t(n)]}\la U_n^{\otimes t(n)}]\\
&= \Pr[\left(\mathsf{flag}_{\mu^*}=\top \wedge b_{\mu^*}=0\right) \mbox{for some }\mu^*\in[n] :
\begin{array}{ll}
& \{\xi_i\}_{i\in[t(n)]}\la U_n^{\otimes t(n)}\\
& \mathsf{flag}_{\mu^*}\la \mathsf{Check}(1^n,\mu^*) \mbox{ for all } \mu^*\in[n] \\
&  (\mu',b_{\mu^*})\la\cA(1^n,\{\xi_i\}_{i\in[t(n)]},\mu^*) \mbox{ for all }\mu^*\in[n]
\end{array}
]\\
&
=
\Pr[
\begin{array}{ll}
&\left(\mathsf{flag}_{\mu^*}=\top \wedge b_{\mu^*}=0\right) \mbox{ for some }\mu^*\in[n] \\
&\mathsf{flag}_{\mu^*}=\top \mbox{ for some } \mu^*\in\mathsf{Bad}_n
\end{array}
:
\begin{array}{ll}
& \{\xi_i\}_{i\in[t(n)]}\la U_n^{\otimes t(n)}\\
& \mathsf{flag}_{\mu^*}\la \mathsf{Check}(1^n,\mu^*) \mbox{ for all  } \mu^*\in[n] \\
&  (\mu',b_{\mu^*})\la\cA(1^n,\{\xi_i\}_{i\in[t(n)]},\mu^*) \mbox{ for all }\mu^*\in[n]
\end{array}
]\label{ineq:long_1}
\\
&
+\Pr[
\begin{array}{ll}
&\left(\mathsf{flag}_{\mu^*}=\top \wedge b_{\mu^*}=0\right) \mbox{ for some }\mu^*\in[n] \\
&\mathsf{flag}_{\mu^*}=\bot \mbox{ for all } \mu^*\in\mathsf{Bad}_n
\end{array}
:
\begin{array}{ll}
& \{\xi_i\}_{i\in[t(n)]}\la U_n^{\otimes t(n)}\\
& \mathsf{flag}_{\mu^*}\la \mathsf{Check}(1^n,\mu^*) \mbox{ for all  } \mu^*\in[n] \\
&  (\mu',b_{\mu^*})\la\cA(1^n,\{\xi_i\}_{i\in[t(n)]},\mu^*) \mbox{ for all }\mu^*\in[n]
\end{array}
]\label{ineq:long_2}.
\end{align}

In the following, we give the upper bound of \cref{ineq:long_1,ineq:long_2}.
We have
\begin{align}
&\mbox{\cref{ineq:long_1}}\\
&\leq 
\Pr[
\begin{array}{ll}
\mathsf{flag}_{\mu^*}=\top \mbox{ for some } \mu^*\in\mathsf{Bad}_n
\end{array}
:
\begin{array}{ll}
& \{\xi_i\}_{i\in[t(n)]}\la U_n^{\otimes t(n)}\\
& \mathsf{flag}_{\mu^*}\la \mathsf{Check}(1^n,\mu^*) \mbox{ for all  } \mu^*\in[n] \\
&  (\mu',b_{\mu^*})\la\cA(1^n,\{\xi_i\}_{i\in[t(n)]},\mu^*) \mbox{ for all }\mu^*\in[n]
\end{array}
]\\
&\leq 
\sum_{\mu^*\in\mathsf{Bad}_n}\Pr[
\begin{array}{ll}
\mathsf{flag}_{\mu^*}=\top 
\end{array}
:
\begin{array}{ll}
& \mathsf{flag}_{\mu^*}\la \mathsf{Check}(1^n,\mu^*) 
\end{array}
]\\
&\leq\abs{\mathsf{Bad}_n}\cdot \frac{1}{n^{10}} \leq\frac{1}{n^9}\leq\frac{1}{2n}
\end{align}
for all sufficiently large $n\in\N$.
Here, in the second inequality, we have used the union bound,
and in the third inequality, we have used \cref{claim:check} and the definition of $\mathsf{Bad}_n$.

We also have
\begin{align}
&\mbox{\cref{ineq:long_2}}\\
&\leq\Pr[
\begin{array}{ll}
\left(\mathsf{flag}_{\mu^*}=\top \wedge b_{\mu^*}=0\right) \mbox{ for some }\mu^*\notin\mathsf{Bad}_n
\end{array}
:
\begin{array}{ll}
& \{\xi_i\}_{i\in[t(n)]}\la U_n^{\otimes t(n)}\\
& \mathsf{flag}_{\mu^*}\la \mathsf{Check}(1^n,\mu^*) \mbox{ for all  } \mu^*\in[n] \\
&  (\mu',b_{\mu^*})\la\cA(1^n,\{\xi_i\}_{i\in[t(n)]},\mu^*) \mbox{ for all }\mu^*\in[n]
\end{array}
]\\
&\leq\Pr[
 b_{\mu^*}=0 \mbox{ for some }\mu^*\notin\mathsf{Bad}_n
:
\begin{array}{ll}
& \{\xi_i\}_{i\in[t(n)]}\la U_n^{\otimes t(n)}\\
&  (\mu',b_{\mu^*})\la\cA(1^n,\{\xi_i\}_{i\in[t(n)]},\mu^*) \mbox{ for all }\mu^*\in[n]
\end{array}
]\label{ineq_Bad}\\
&\leq \left(n-\abs{\mathsf{Bad}_n}\right)\cdot \frac{1}{n^5}\leq \frac{1}{n^4}\leq \frac{1}{2n}
\end{align}
for all sufficiently large $n\in\N$.
Here, in the third inequality, we have used that, for all $\mu^*\notin\mathsf{Bad}_n$, we have
\begin{align}
\Pr[
 b^*=0 
:
\begin{array}{ll}
& \{\xi_i\}_{i\in[t(n)]}\la U_n^{\otimes t(n)}\\
&  (\mu',b^*)\la\cA(1^n,\{\xi_i\}_{i\in[t(n)]},\mu^*) 
\end{array}
]\leq\frac{1}{n^5}.
\end{align}

Therefore, we have
\begin{align}
\Pr[0\la\cB(1^n,\{\xi_i\}_{i\in[t(n)]} ): \{\xi_i\}_{i\in[t(n)]}\la U_n^{\otimes t(n)}]\leq \frac{1}{n}
\end{align}
for all sufficiently large $n\in\N$.
\end{proof}

\begin{proof}[Proof of \cref{claim:check}]
From the construction of $\mathsf{Check}$, we have
\begin{align}
\Pr[\top\la\mathsf{Check}(1^n,\mu)]
=
\left(
\Pr[b^*=1:
\begin{array}{ll}
     & \{\xi_i\}_{i\in [t(n)]}\la U_n^{\otimes t(n)} \\
     & (\mu^*,b^*)\la\cA(1^n,\{\xi_i\}_{i\in [t(n)]},\mu)
\end{array}]
\right)^{n^{25}}.
\end{align}
Therefore, if 
we have
\begin{align}
    \Pr[b^*=1:
\begin{array}{ll}
     & \{\xi_i\}_{i\in [t(n)]}\la U_n^{\otimes t(n)} \\
     & (\mu^*,b^*)\la\cA(1^n,\{\xi_i\}_{i\in [t(n)]},\mu)
\end{array}
]\leq 1-\frac{1}{n^5},
\end{align}
then we have 
\begin{align}
\Pr[\top\la\mathsf{Check}(1^n,\mu)]\leq \left(1-\frac{1}{n^{5}}\right)^{n^{25}}\leq \frac{1}{n^{10}}
\end{align}
for all sufficiently large $n\in\N$.

\end{proof}

\begin{proof}[Proof of \cref{claim:markov_EFID}]
For all $\mu\in[n]$ such that 
\begin{align}
    \mathsf{SD}(\Gen(1^n,\mu), U_n)\geq 1-\negl(n),
\end{align}
we have
\begin{align}
&\Pr\left[ 
\mathsf{SD}(\cD(1^n,(\mu',b)),\cD(1^n,(\mu,0)))\leq 1/n:
\begin{array}{ll}
&\{\xi_i\}_{i\in[t(n)]}\la\cD(1^n,(\mu,0))^{\otimes t(n)} \\
&(\mu',b)\la \cA(1^n,\{\xi_i\}_{i\in[t(n)]},\mu)
\end{array}
\right]\\
&=\Pr\left[ 
b=0:
\begin{array}{ll}
&\{\xi_i\}_{i\in[t(n)]}\la\Gen(1^n,\mu)^{\otimes t(n)} \\
&(\mu',b)\la \cA(1^n,\{\xi_i\}_{i\in[t(n)]},\mu)
\end{array}
\right]\label{eq:A_1}
\end{align}
and
\begin{align}
&\Pr\left[ 
\mathsf{SD}(\cD(1^n,(\mu',b)),\cD(1^n,(\mu,1)))\leq 1/n:
\begin{array}{ll}
&\{\xi_i\}_{i\in[t(n)]}\la\cD(1^n,(\mu,1))^{\otimes t(n)} \\
&(\mu',b)\la \cA(1^n,\{\xi_i\}_{i\in[t(n)]},\mu)
\end{array}
\right]\\
&=\Pr\left[ 
b=1:
\begin{array}{ll}
&\{\xi_i\}_{i\in[t(n)]}\la U_n^{\otimes t(n)} \\
&(\mu',b)\la \cA(1^n,\{\xi_i\}_{i\in[t(n)]},\mu)
\end{array}
\right]\label{eq:A_2}.
\end{align}
Here, in the \cref{eq:A_1}, we have used the fact that $\cD(1^n,(\mu,0))$ is far from $\cD(1^n,(\mu^*,b^*))$ for all $\mu^*\neq \mu$ and $b^*\neq 0$, and in the \cref{eq:A_2}, we have used the fact that $\cD(1^n,(\mu,1)))$ is far from $\cD(1^n,(\mu^*,b^*))$ for all $\mu^*\neq \mu$ and $b^*\neq 1$.

In the following, we show that
for all $(\mu,b)\in[n]\times\bit$, we have
\begin{align}
\Pr\left[ 
\mathsf{SD}(\cD(1^n,(\mu,b)),\cD(1^n,(\mu',b')))\leq 1/n:
\begin{array}{ll}
&\{\xi_i\}_{i\in[t(n)]}\la\cD(1^n,(\mu,b))^{\otimes t(n)} \\
&(\mu',b')\la \cA(1^n,\{\xi_i\}_{i\in[t(n)]},\mu)
\end{array}
\right]\geq 1-\frac{1}{n^{80}}
\end{align}
for infinitely many $n\in\N$.
By combining it with the above argument, we obatin the claim.

To show it,    
let $\mathsf{Good}_{n}$ be a set of $(\mu,b)$ such that
\begin{align}
\Pr\left[ 
\mathsf{SD}(\cD(1^n,(\mu,b)),\cD(1^n,(\mu',b')))\leq 1/n:
\begin{array}{ll}
&\{\xi_i\}_{i\in[t(n)]}\la\cD(1^n,(\mu,b))^{\otimes t(n)} \\
&(\mu',b')\la \cA(1^n,\{\xi_i\}_{i\in[t(n)]},\mu)
\end{array}
\right]\geq 1-\frac{1}{n^{80}}.
\end{align}
We will show that $|\mathsf{Good}|=2n$.

From \cref{ineq_3}, we have
\begin{align}
1-\frac{1}{n^{100}}&\leq \Pr\left[ 
\mathsf{SD}(\cD(1^n,(\mu,b)),\cD(1^n,h))\leq 1/n:
\begin{array}{ll}
&(\mu,b)\la\cS(1^n)\\
&\{\xi_i,\mu\}_{i\in[t(n)]}\la\cD(1^n,(\mu,b))^{\otimes t(n)} \\
& h\la \cA(1^n,\{\xi_i\}_{i\in[t(n)]},\mu)
\end{array}
\right]\\
&=\sum_{(\mu,b)\in\mathsf{Good}_n}\frac{1}{2n}\Pr\left[ 
\mathsf{SD}(\cD(1^n,(\mu,b)),\cD(1^n,(\mu',b')))\leq 1/n:
\begin{array}{ll}
&\{\xi_i\}_{i\in[t(n)]}\la\cD(1^n,(\mu,b))^{\otimes t(n)} \\
&(\mu',b')\la \cA(1^n,\{\xi_i\}_{i\in[t(n)]},\mu)
\end{array}
\right]\\
&+\sum_{(\mu,b)\notin\mathsf{Good}_n}\frac{1}{2n}\Pr\left[ 
\mathsf{SD}(\cD(1^n,(\mu,b)),\cD(1^n,(\mu',b')))\leq 1/n:
\begin{array}{ll}
&\{\xi_i\}_{i\in[t(n)]}\la\cD(1^n,(\mu,b))^{\otimes t(n)} \\
&(\mu',b')\la \cA(1^n,\{\xi_i\}_{i\in[t(n)]},\mu)
\end{array}
\right]\\
&\leq\sum_{(\mu,b)\in\mathsf{Good}_n}\frac{1}{2n}+\sum_{(\mu,b)\notin\mathsf{Good}_n}\frac{1}{2n}\left(1-\frac{1}{n^{80}}\right)\\
&=\abs{\mathsf{Good}_{n}}\frac{1}{2n}+\left(2n-\abs{\mathsf{Good}_n}\right)\frac{1}{2n}\left(1-\frac{1}{n^{80}}\right)
=1-\frac{1}{n^{80}}+ \frac{\abs{\mathsf{Good}_n}}{2n^{81}}
\end{align}
for infinitely many $n\in\N$.

This implies that
\begin{align}
2n-\frac{2}{n^{19}}\leq\abs{\mathsf{Good}_n}
\end{align}
for infinitely many $n\in\N$.
Because $|\mathsf{Good}_n|$ is an integer, this means that
$|\mathsf{Good}_n|=2n$ for infinitely many $n\in\mathbb{N}$.

\end{proof}
\end{proof}

\subsection{Classical Case}
In this subsection, we apply the previous quantum results to the classical case.
\if
\taiga{
In this section, we show that the existence of OWFs is equivalent to the average-case hardness of proper classical distribution learning.  
\cite{FOCS:HirNan23} establishes this equivalence in the case of \emph{improper} learning.  
Their techniques inherently rely on improper learning, and it is unclear how to extend them to the \emph{proper} setting.  
Therefore, we present the classical result for proper learning in this section.
Our technique parallels that used to establish the equivalence between OWPuzzs and the average-case hardness of proper quantum distribution learning.
}
\fi

\begin{definition}[Average-Case Hardness of Proper Classical Distribution Learning~\cite{STOC:KMRRSS94,FOCS:HirNan23}]
We say that average-case hardness of proper classical distribution learning holds if the following holds.
There exist some polynomials $\epsilon$ and $\delta$, a PPT algorithm $\cS$, which takes $1^n$ as input, and outputs $z\in\bit^n$, and a PPT algorithm $\cD$, which takes $1^n$ and $z\in\bit^n$ as input, and outputs $x\in\bit^{\poly(n)}$, such that, for all PPT algorithms $\cA$ and any polynomial $t$,
\begin{align}
\Pr[ 
\mathsf{SD}(\cD(1^n,z),\cD(1^n,h))\leq 1/\epsilon(n):
\begin{array}{ll}
&z\la\cS(1^n)\\
&\{x_i\}_{i\in[t(n)]}\la\cD(1^n,z)^{\otimes t(n)} \\
& h\la \cA(1^n,\{x_i\}_{i\in[t(n)]})
\end{array}
]\leq 1-1/\delta(n)
\end{align}
for all sufficiently large $n\in\N$.
\end{definition}

\begin{definition}[Average-Case Hardness of Improper Classical Distribution Learning~\cite{STOC:KMRRSS94,FOCS:HirNan23}]
We say that average-case hardness of improper classical distribution learning holds if the following holds.
There exist some polynomials $\epsilon$ and $\delta$, a PPT algorithm $\cS$, which takes $1^n$ as input, and outputs $z\in\bit^n$, and a PPT algorithm $\cD$, which takes $1^n$ and $z\in\bit^n$ as input, and outputs $x\in\bit^{\poly(n)}$, such that, for all PPT algorithms $\cA$ and any polynomial $t$, 
\begin{align}
\Pr[ 
\mathsf{SD}(\cD(1^n,z),h)\leq 1/\epsilon(n):
\begin{array}{ll}
&z\la\cS(1^n)\\
&\{x_i\}_{i\in[t(n)]}\la\cD(1^n,z)^{\otimes t(n)} \\
& h\la \cA(1^n,\{x_i\}_{i\in[t(n)]})
\end{array}
]\leq 1-1/\delta(n)
\end{align}
for all sufficiently large $n\in\N$.
Here, $h$ is a description of a PPT algorithm.
\end{definition}

\if0
\begin{definition}[Average-Case Hardness of Classical Maximum Likelihood]
We say that average-case hardness of classical maximum likelihood holds if the following holds.
There exist some polynomials $\epsilon$ and $\delta$, a PPT algorithm $\cS$, which takes $1^n$ as input, and outputs $z\in\bit^n$, and a PPT algorithm $\cD$, which takes $1^n$ and $z\in\bit^n$ as input, and outputs $x\in\bit^{\poly(n)}$, such that, for all PPT algorithms $\cA$
\begin{align}
\Pr[ 
\frac{\max_{z\in \bit^n}\{\Pr[x \la\cD(1^n,z)]\} }{\Pr[x\la \cD(1^n,h) ]}\leq 2^{\epsilon(n)}:
\begin{array}{ll}
&z\la\cS(1^n)\\
&x\la\cD(1^n,z) \\
& h\la \cA(1^n,x)
\end{array}
]\leq 1-1/\delta(n)-2^{n-\epsilon(n)}
\end{align}
for all sufficiently large $n\in\N$.
\end{definition}
\fi

\begin{theorem}\label{thm:OWF}
The following four conditions are equivalent:
\begin{enumerate}
\item OWFs exist.
\item Classical OWPuzzs exist.
\item Average-case hardness of proper classical distribution learning holds.
\item Average-case hardness of improper classical distribution learning holds.
\end{enumerate}
\end{theorem}

\if0
\begin{remark}
\color{red}
Note that it is open if we can show the equivalence between OWPuzzs and average case hardness of improper quantum distribution learning.
\cite{FOCS:HirNan23} shows the equivalence between OWFs and average-case hardness of improper classical distribution learning.
They show the hardness of improper learning from OWFs based on the equivalence between OWFs and PRFs.
In the quantum case, it is an open problem how to show the equivalence between OWPuzzs and a quantum analog of PRFs.
Therefore, we do not know how to show the hardness of improper quantum distribution learning from OWPuzzs.
\color{black}
\mor{Ueto onaji?}
\end{remark}
\fi

Because its proof is similar to that of \cref{thm:OWPuzz}, we give only a sketch.
\begin{proof}[Proof sketch of \cref{thm:OWF}]
\cite{FOCS:HirNan23} shows the equivalence between $(1)$ and $(4)$.
$(3)\Leftarrow (4)$ directly follows.
$(2)\Leftarrow (3) $ follows in the same way as \cref{lem:OWPuzz}.
The difference is that we consider classical OWPuzzs instead of OWPuzzs.
$(1)\Leftarrow (2)$ follows from \cref{thm:classical_OWPuzz}.
\end{proof}
\section{Agnostic Distribution Learning with Respect to KL Divergence}\label{sec:KL_divergence}
\subsection{Quantum Case}\label{sec:quantum_KL}
In this section, we introduce hardness of agnostic quantum distribution learning with respect to KL divergence based on \cite{ML:AW92}.
Then, we show that $\mathbf{BQP}\neq \mathbf{PP}$ holds if and only if hardness of agnostic quantum distribution learning with respect to KL divergence holds.
\begin{definition}[Hardness of Agnostic Quantum Distribution Learning with Respect to KL Divergence]\label{def:Learn_KL}
We say that hardness of agnostic quantum distribution learning with respect to KL divergence holds if the following holds:
There exist some polynomials $\epsilon$, $\delta$ and $p$, and a QPT algorithm $\cD$, which takes $1^n$ and $z\in\bit^n$ as input, and outputs $x\in\bit^{p(n)}$ such that the following two conditions are satisfied:

\begin{enumerate}
\item For all $n\in\N$, $z\in\bit^n$, and $x\in\bit^{p(n)}$, we have
\begin{align}
\Pr[x\la\cD(1^n,z)]\neq 0.
\end{align}
\item 
For any QPT algorithm $\cA$ and any polynomial $t$, there exists an algorithm $\cT$, which takes $1^n$ as input and outputs $x\in\bit^{p(n)}$ such that 
\begin{align}
\label{ineq:q_KL_learn}
\Pr[ 
D_{KL}(\cT(1^n)\,\|\,\cD(1^n,h))\leq \mathsf{Opt}_n+1/\epsilon(n):
\begin{array}{ll}
&\{x_i\}_{i\in[t(n)]}\la\cT(1^n)^{\otimes t(n)}\\
& h\la \cA(1^n,\{x_i\}_{i\in[t(n)]})
\end{array}
]\leq 1-1/\delta(n)
\end{align}
for infinitely many $n\in\N$.
Here, $\mathsf{Opt}_n\seteq \min_{a\in\bit^n}\{D_{KL}(\cT(1^n)\,\|\,\cD(1^n,a))\}$.
\end{enumerate}
\end{definition}

\begin{definition}[Worst-Case Hardness of Quantum Maximum Likelihood Estimation]\label{def:QML}
We say that worst-case hardness of quantum maximum likelihood estimation holds if the following holds.
There exist some polynomials $\epsilon$, $\delta$ and $p$, and a QPT algorithm $\cD$, which takes $1^n$ and $z\in\bit^n$ as input, and outputs $x\in\bit^{p(n)}$ such that, for any QPT algorithm $\cA$, there exists $x\in\bit^{p(n)}$ with $\max_{a\in\bit^n}\Pr[x\la\cD(1^n,a)]>0$ such that 
\begin{align}\label{ineq:QML}
\Pr[ \frac{\max_{a\in\bit^n}\Pr[x\la\cD(1^n,a)]}{\Pr[x\la\cD(1^n,h)]} \leq 2^{1/\epsilon(n)} : h\la\cA(1^n,x)]\leq 1-1/\delta(n)
\end{align}
for infinitely many $n\in\N$.
\end{definition}
\begin{remark}
In the above definitions, we require that \cref{ineq:QML,ineq:q_KL_learn} hold \emph{for infinitely many security parameters} instead of \emph{for all sufficiently large security parameters}.
This is because we want to show its equivalence to $\mathbf{PP}\neq \mathbf{BQP}$ instead of $\mathbf{PP}\neq \mathbf{ioBQP}$.
If we use the ``for all sufficiently large security parameters'' security, we obtain the equivalence to
$\mathbf{PP}\neq \mathbf{ioBQP}$
in a similar proof.
\end{remark}

\begin{theorem}\label{thm:PP_completeness}
The following three conditions are equivalent:
\begin{enumerate}
\item $\mathbf{PP}\neq \mathbf{BQP}$.
\item Hardness of agnostic quantum distribution learning with respect to KL divergence holds.
\item Worst-case hardness of quantum maximum likelihood estimation holds.
\end{enumerate}
\end{theorem}

\cref{thm:PP_completeness} follows from the following \cref{lem:PP_hardness,lem:PP_easiness,lem:KL_learn_to_QML,lem:QML_to_KL_Learn}.
\begin{lemma}\label{lem:PP_hardness}
If $\mathbf{PP}\neq \mathbf{BQP}$, then worst-case hardness of quantum maximum likelihood estimation holds.
\end{lemma}

\begin{lemma}\label{lem:PP_easiness}
If $\mathbf{PP}=\mathbf{BQP}$, then worst-case hardness of quantum maximum likelihood estimation does not hold.
\end{lemma}
\cref{lem:PP_easiness} follows from Lemma 4.2 in \cite{HH24}.
Therefore, we omit the proof and refer it to \cite{HH24}.

\begin{lemma}\label{lem:QML_to_KL_Learn}
If worst-case hardness of quantum maximum likelihood estimation holds, then hardness of agnostic quantum distribution learning with respect to KL divergence holds.
\end{lemma}

\begin{lemma}\label{lem:KL_learn_to_QML}
If hardness of agnostic quantum distribution learning with respect to KL divergence holds, then worst-case hardness of quantum maximum likelihood estimation holds.
\end{lemma}

In the following, we give the proofs.

\begin{proof}[Proof of \cref{lem:PP_hardness}]

For contradiction, assume that for any polynomial $\epsilon$, $\delta$ and $p$, and any QPT algorithm $\cD$, which takes $1^n$ and $z\in\bit^{n}$ as input, and outputs $x\in\bit^{p(n)}$, there exists a QPT algorithm $\cA$ such that for all $x\in\bit^{p(n)}$ with $\max_{a\in\bit^n}\Pr[x\la\cD(1^n,a)]>0$, we have
\begin{align}
\Pr[ \frac{\max_{a\in\bit^n} \Pr[x\la\cD(1^n,a)]}{\Pr[x\la\cD(1^n,h)]} \leq 2^{1/\epsilon(n)} : h\la\cA(1^n,x)]\geq 1-1/\delta(n)
\end{align}
for all sufficiently large $n\in\N$.
Then, we show that, for any $\cL\in\mathbf{PP}$, we have $\cL\in\mathbf{BQP}$.

Because $\mathbf{PP}=\mathbf{PostBQP}$, for any $\cL\in\mathbf{PP}$,
there exist a polynomial $q$ and a QPT algorithm $\cB$, which takes $x$ as input, and outputs $(b,b^*)\in\bit\times\bit$ such that the following three conditions are satisfied:
\begin{enumerate}
\item For all $x\in\bit^*$, $\Pr[b^*=1:(b,b^*)\la\cB(x)]\geq 2^{-q(\abs{x})}$.
\item For all $x\in\cL$, we have
\begin{align}
\Pr_{(b,b^*)\la\cB(x)}[b=1|b^*=1]\geq \frac{3}{4}.
\end{align}
\item For all $x\notin\cL$, we have
\begin{align}
\Pr_{(b,b^*)\la\cB(x)}[b=1|b^*=1]\leq \frac{1}{4}.
\end{align}
\end{enumerate}

From the QPT algorithm $\cB$, we construct the following QPT algorithm $\cB^*$.
\begin{description}
\item[$\cB^*(1^{n+1},x,c)$:]$ $
\begin{enumerate}
\item Take $1^{n+1}$, $x\in\bit^n$, and $c\in\bit$ as input.
\item Run $(b,b^*)\la\cB(x)$.
\item Output $(x,1)$ if $b=c$ and $b^*=1$.
      Otherwise, output $\bot$.
\end{enumerate}
\end{description}
From the construction of $\cB^*$, we have
\begin{align}
\Pr[(x,1)\la \cB^*(1^{\abs{x}+1},x,c)]=\Pr_{(b,b^*)\la \cB(x)}[b=c\wedge b^*=1]
\end{align}
for all $x\in\bit^*$ and $c\in\bit$.
Therefore, the following two conditions are satisfied:
\begin{enumerate}
\item For all $x\in\cL$, we have
\begin{align}
\frac{\Pr[(x,1)\la\cB^*(1^{\abs{x}+1},x,1)]}{\Pr[(x,1)\la\cB^*(1^{\abs{x}+1},x,0)]}
&=
\frac{\Pr_{(b,b^*)\la\cB(x)}[b=1\wedge b^*=1]}{\Pr_{(b,b^*)\la\cB(x)}[b=0\wedge b^*=1]}\\
&=\frac{\Pr_{(b,b^*)\la\cB(x)}[b=1| b^*=1]}{\Pr_{(b,b^*)\la\cB(x)}[b=0| b^*=1]}
\geq 3.
\end{align}
\item 
For all $x\notin\cL$, we have
\begin{align}
\frac{\Pr[(x,1)\la\cB^*(1^{\abs{x}+1},x,1)]}{\Pr[(x,1)\la\cB^*(1^{\abs{x}+1},x,0)]}
&=
\frac{\Pr_{(b,b^*)\la\cB(x)}[b=1\wedge b^*=1]}{\Pr_{(b,b^*)\la\cB(x)}[b=0\wedge b^*=1]}\\
&=\frac{\Pr_{(b,b^*)\la\cB(x)}[b=1| b^*=1]}{\Pr_{(b,b^*)\la\cB(x)}[b=0| b^*=1]}
\leq \frac{1}{3}.
\end{align}
\end{enumerate}

On the other hand, we assume that worst-case hardness of quantum maximum likelihood estimation does not hold.
This means that for any polynomial $\epsilon$, $\delta$ and $p$, and any QPT algorithm $\cD$, which takes $1^n$ and $z\in\bit^{n}$ as input, and outputs $x\in\bit^{p(n)}$, there exists a QPT algorithm $\cA$ such that for all $x\in\bit^{p(n)}$ with $\max_{a\in\bit^n}\Pr[x\la\cD(1^n,a)]>0$, we have
\begin{align}
\Pr[ \frac{\max_{a\in\bit^n} \Pr[x\la\cD(1^n,a)]}{\Pr[x\la\cD(1^n,h)]} \leq 2^{1/\epsilon(n)} : h\la\cA(1^n,x)]\geq 1-1/\delta(n)
\end{align}
for all sufficiently large $n\in\N$.
Therefore, there exists a QPT algorithm $\cA$ such that, 
for all $x\in\bit^n$ with $\max_{b\in\bit}\Pr[(x,1)\la \cB^*(1^{n+1}, (x,b))] > 0$,
\begin{align}
\Pr[ \frac{\max_{b\in\bit} \Pr[(x,1)\la\cB^*(1^{n+1},(x,b))]}{\Pr[(x,1)\la\cB^*(1^{n+1},(x,b^*))]} \leq 2 : (x,b^*)\la\cA(1^{n+1},(x,1))]\geq \frac{5}{6}
\end{align}
for all sufficiently large $n\in\N$.

This implies that $\cA$ satisfies the following:
\begin{enumerate}
\item For all $x\in\cL$, we have
\begin{align}
\Pr[(x,1) \la\cA(1^{\abs{x}},x)]\geq \frac{5}{6}.
\end{align}
\item 
For all $x\notin\cL$, we have
\begin{align}
\Pr[(x,1)\la\cA(1^{\abs{x}},x)]\leq \frac{1}{6}.
\end{align}
\end{enumerate}
This implies that $\cL\in\mathbf{BQP}$.

\end{proof}

\begin{proof}[Proof of \cref{lem:QML_to_KL_Learn}]

Let $\epsilon$, $\delta$, and $p$ be some polynomials, and let $\cD$ be a QPT algorithm given in \cref{def:QML}.
From $\cD$, we construct a QPT algorithm $\cD^*$ such that the following two conditions are satisfied:
\begin{enumerate}
\item For all $n\in\N$, $z\in\bit^n$ and $x\in\bit^{p(n)}$, we have 
\begin{align}
    \Pr[x\la\cD^*(1^n,z)]\neq 0\label{ineq:non_zero}.
\end{align}
\item For any QPT algorithm $\cA$ and any polynomial $t$, there exists an algorithm $\cT$, which takes $1^n$ as input and outputs $x\in\bit^{p(n)}$ such that 
\begin{align}
\Pr[ 
D_{KL}(\cT(1^n)\,\|\,\cD^*(1^n,h))\leq \mathsf{Opt}_n+1/2\epsilon(n):
\begin{array}{ll}
&\{x_i\}_{i\in[t(n)]}\la\cT(1^n)^{\otimes t(n)}\\
& h\la \cA(1^n,\{x_i\}_{i\in[t(n)]})
\end{array}
]\leq 1-1/\delta(n)
\end{align}
for infinitely many $n\in\N$.
Here, $\mathsf{Opt}_n\seteq \min_{a\in\bit^n}\{D_{KL}(\cT(1^n))\,\|\,\cD^*(1^n,a)\}$.
\end{enumerate}

Let $m$ be a polynomial such that $m(n)\geq p(n)$ and for all $z\in\bit^n$ and for all $x\in\bit^{p(n)}$ with $\Pr[x\la\cD(1^n,z)]\neq 0$, 
it holds that
\begin{align}
\Pr[x\la\cD(1^n,z)]\geq 2^{-m(n)}.    
\end{align}
Let us define $C_n$ as follows: 
\begin{align}\label{eqn:A}
C_n\coloneqq(2^{-1/2\epsilon(n)}-2^{-1/\epsilon(n)})\cdot2^{-\left(m(n)+1-p(n)\right)}.
\end{align}

Now, from $\cD$, we give the construct of $\cD^*$ as follows.
\begin{description}
    \item[$\cD^*(1^n,z):$] $ $
    \begin{enumerate}
        \item Run $x\la\cD(1^n,z)$. Sample $x^*\la\bit^{p(n)}$.
        \item Output $x$ with probability $1-C_n$, and output $x^*$ with probability $C_n$.
    \end{enumerate}
\end{description}

It is clear that $\cD^*$ satisfies \cref{ineq:non_zero}.
Next, for contradiction, let us assume that there exists a polynomial $t$ and a QPT algorithm $\cA$ such that for all algorithms $\cT$, which takes $1^n$ as input, and outputs $x\in\bit^{p(n)}$, we have 
\begin{align}\label{eqn:A^*}
\Pr[ 
D_{KL}(\cT(1^n)\,\|\,\cD^*(1^n,h))\leq \mathsf{Opt}_n+1/2\epsilon(n):
\begin{array}{ll}
&\{x_i\}_{i\in[t(n)]}\la\cT(1^n)^{\otimes t(n)}\\
& h\la \cA(1^n,\{x_i\}_{i\in[t(n)]})
\end{array}
]\geq 1-1/\delta(n)
\end{align}
for all sufficiently large $n\in\N$.
Here, $\mathsf{Opt}_n\seteq \min_{a\in\bit^n}\{D_{KL}(\cT(1^n))\,\|\,\cD^*(1^n,a)\}$.
Then, we show that $\cA$ satisfies 
\begin{align}
\Pr[ \frac{\max_{a\in\bit^n} \Pr[x\la\cD(1^n,a)]}{\Pr[x\la\cD(1^n,h)]} \leq 2^{1/\epsilon(n)} : h\la\cA(1^n,x)]\geq 1-1/\delta(n)
\end{align}
for all $x$ with $\max_{a\in\bit^n}\Pr[x\la\cD(1^n,a)]>0$
and
for all sufficiently large $n\in\N$, which is a contradiction.

For an arbitrary $x\in\bit^*$, we write $\cT_x$ to mean an algorithm that outputs $x$.
\cref{eqn:A^*} implies that for all $x\in\bit^{p(n)}$, with probability at least $1-1/\delta(n)$ over $h\la\cA(1^{n},x)$, $h$ satisfies
\begin{align}
\frac{1}{2\epsilon(n)}
&\geq D_{KL}(\cT_x\,\|\,\cD^*(1^n,h))- \mathsf{Opt}_n\\
&=D_{KL}(\cT_x\,\|\,\cD^*(1^n,h)) -\min_{a\in\bit^n}D_{KL}(\cT_x\,\|\,\cD^*(1^n,a))\\
&=\log\left(\frac{1}{\Pr[x\la\cD^*(1^n,h)]}\right)-\min_{a\in\bit^n}\log\left(\frac{1}{\Pr[x\la\cD^*(1^n,a)]}\right)\\
&=\max_{a\in\bit^n}\log\left(\frac{\Pr[x\la\cD^*(1^n,a)]}{\Pr[x\la\cD^*(1^n,h)]}\right)
\end{align}
for all sufficiently large $n\in\N$.

This implies that
$h$ satisfies
\begin{align}
2^{1/2\epsilon(n)}&\geq \frac{\max_{a\in\bit^n}\{\Pr[x\la\cD^*(1^n,a)]\}}{\Pr[x\la\cD^*(1^n,h)]}\\
&=\frac{\{(1-C_n)\max_{a\in\bit^n}\{\Pr[x\la\cD(1^n,a)]\}+C_n\cdot2^{-p(n)}}{(1-C_n)\Pr[x\la\cD(1^n,h)]+C_n\cdot 2^{-p(n)}}
\end{align}
for all sufficiently large $n\in\N$.
This implies that,
for all $x\in\bit^{p(n)}$ with $\max_{a\in\bit^n}\{\Pr[x\la\cD(1^n,a)]\}> 0$,
$h$ satisfies
\begin{align}
\Pr[x\la\cD(1^n,h)]
&\geq 2^{-1/2\epsilon(n)}\max_{a\in\bit^n}\{\Pr[x\la\cD(1^n,a)]\}- \frac{1-2^{-1/2\epsilon(n)}}{1-C_n }C_n\cdot2^{-p(n)}\\
&\geq 2^{-1/2\epsilon(n)}\max_{a\in\bit^n}\{\Pr[x\la\cD(1^n,a)]\}- 2\cdot C_n\cdot 2^{-p(n)}\label{ineq:C_n}\\
&=2^{-1/2\epsilon(n)}\max_{a\in\bit^n}\{\Pr[x\la\cD(1^n,a)]\}- (2^{-1/2\epsilon(n)}-2^{-1/\epsilon(n)}) \cdot 2^{-m(n)}\label{eq:C_n_2}\\
&\geq2^{-1/2\epsilon(n)}\max_{a\in\bit^n}\{\Pr[x\la\cD(1^n,a)]\}- (2^{-1/2\epsilon(n)}-2^{-1/\epsilon(n)}) \max_{a\in\bit^n}\{\Pr[x\la\cD(1^n,a)]\} \label{ineq_5}\\
&= 2^{-1/\epsilon(n)} \cdot \max_{a\in\bit^n}\{\Pr[x\la\cD(1^n,a)]\}.
\end{align}
Here, in the \cref{ineq:C_n}, we have used $1-C_n\geq \frac{1}{2}$, and in the \cref{eq:C_n_2}, we have used the definition of $C_n$, and in the \cref{ineq_5}, we have used that $\max_{a\in\bit^n}\{\Pr[x\la\cD(1^n,a)]\}\geq 2^{-m(n)}$.

\if0
Here, in the \cref{eq:C_n_2}, we have used \cref{eqn:A}, and, in the \cref{ineq_5}, we have used 
$\max_{a\in\bit^n}\{\Pr[x\la\cD(1^n,a)]\}\geq 2^{-m(n)}$.
\fi
This implies that
there exists a QPT $\cA$ such that for all $x\in\bit^{p(n)}$ with $\max_{a\in\bit^n}\Pr[x\gets\cD(1^n,a)]>0$, we have
\begin{align}
\Pr[ \frac{\max_{a\in\bit^n} \Pr[x\la\cD(1^n,a)]}{\Pr[x\la\cD(1^n,h)]} \leq 2^{1/\epsilon(n)} : h\la\cA(1^n,x)]\geq 1-1/\delta(n)
\end{align}
for all sufficiently large $n\in\N$, which is a contradiction.

\end{proof}

\begin{proof}[Proof of \cref{lem:KL_learn_to_QML}]

Let $\epsilon$, $\delta$, and $p$ be polynomials, and let $\cD$ be a QPT algorithm given in \cref{def:Learn_KL}.
From $\cD$, we show that there exist a polynomial $t$ and a QPT algorithm $\cD^*$ such that, for any QPT algorithm $\cA$, there exists $x\in\bit^{t(n)\cdot p(n)}$ with $\max_{a\in\bit^n}\Pr[x\la\cD^*(1^n,a)]>0$ such that
\begin{align}\label{ineq:A_QML}
\Pr[ \frac{\max_{a\in\bit^n}\Pr[x\la\cD^*(1^n,a)]}{\Pr[x\la\cD^*(1^n,h)]} \leq 2^{\frac{t(n)}{2\epsilon(n)}} : h\la\cA(1^n,x)]\leq 1-1/\delta(n)
\end{align}
for infinitely many $n\in\N$.

\if0
Let $\epsilon$ and $\delta$ be polynomials, and let $\cD$ be a QPT algorithm that takes $1^n$ and $z\in\bit^n$ as input, and outputs $x\in\bit^{\poly(n)}$ such that $\Pr[x\la\cD(1^n,z)]\neq 0$ for all $n\in\N$, $z\in\bit^n$, and $x\in\bit^{\poly(n)}$, and
for all QPT algorithms $\cA$ and all polynomials $t$, there exists an algorithm $\cT$, which takes $1^n$ and $z\in\bit^n$ as input, and outputs $x\in\bit^{\poly(n)}$, we have
\begin{align}
\Pr[ 
D_{KL}(\cT(1^n)\,\|\,\cD(1^n,h))\leq \mathsf{Opt}_n+1/\epsilon(n):
\begin{array}{ll}
&\{x_i\}_{i\in[T]}\la\cT(1^n)^{\otimes t(n)}\\
& h\la \cA(1^n,\{x_i\}_{i\in[t(n)]})
\end{array}
]\leq 1-1/\delta(n)
\end{align}
for infinitely many $n\in\N$.
Here, $\mathsf{Opt}_n\seteq \min_{a\in\bit^n}\{D_{KL}(\cT(1^n)\,\|\, \cD(1^n,a))\}$.
\fi

There exists a polynomial $m$ such that $\Pr[x\la\cD(1^n,z)]\geq2^{-m(n)}$ for all $x\in\bit^{p(n)}$ and $z\in\bit^n$.
Let $t(n)\coloneqq \lceil1000m^2(n)\epsilon^2(n)\left(n+\log(2\delta(n))\right)\rceil$.
Let us consider the following QPT algorithm $\cD^*$.
\begin{description}
\item[$\cD^*(1^n,z)$:]$ $
\begin{enumerate}
\item Run $x_i\la\cD(1^n,z)$ for all $i\in[t(n)]$.
\item Output $x\seteq\{x_i\}_{i\in[t(n)]}$.
\end{enumerate}
\end{description}

In the following, we will show that $\cD^*$ satisfies \cref{ineq:A_QML}.
For contradiction, let us assume that there exists a QPT algorithm $\cA$ such that for all $x\in\bit^{t(n)\cdot p(n)}$ with 
$\max_{a\in\bit^n}\Pr[x\la\cD^*(1^n,a)]>0$, we have
\begin{align}
\Pr[ \frac{\max_{a\in\bit^n}\Pr[x\la\cD^*(1^n,a)]}{\Pr[x\la\cD^*(1^n,h)]} \leq 2^{\frac{t(n)}{2\epsilon(n)}} : h\la\cA(1^n,x)]\geq 1-1/2\delta(n)
\end{align}
for all sufficiently large $n\in\N$.
Then, we show that $\cA$ contradicts the security of $\cD$.

For showing it, we use the following claim.
\begin{claim}\label{claim:KL}
Let $\cT$ be any algorithm that, on input $1^n$, outputs a $p(n)$-bit string.
For all sufficiently large $n\in\N$,
with probability at least $1-1/\delta(n)$ over $\{x_i\}_{i\in[t(n)]} \la\cT(1^n)^{\otimes t(n)}$ and $h\la\cA(1^n,\{x_i\}_{i\in[t(n)]})$, we have
\begin{align}
\mathbb{E}_{x\la\cT(1^n)}\left[\log\left(\frac{\Pr[x\la \cD(1^n,z)]}{\Pr[x\la \cD(1^n,h)]}\right)\right]\leq 1/\epsilon(n)
\end{align}
for all $z\in\bit^n$.
\end{claim}
We give the proof of \cref{claim:KL} in the end of the proof.
From \cref{claim:KL}, for all sufficiently large $n\in\N$, with probability at least $1-1/\delta(n)$ over $\{x_i\}_{i\in[t(n)]} \la\cT(1^n)^{\otimes t(n)}$ and $h\la\cA(1^n,\{x_i\}_{i\in[t(n)]})$,  we have
\begin{align}
1/\epsilon(n)&\geq\sum_{x\in\bit^{p(n)}}\Pr[x\la\cT(1^n)]\left(
\log\left(\frac{1}{\Pr[x\la \cD(1^n,h)]}\right)
-\log\left(\frac{1}{\Pr[x\la \cD(1^n,z)]}\right)
\right) \\
& =\sum_{x\in\bit^{p(n)}}\Pr[x\la\cT(1^n)]\left(
\log\left(\frac{\Pr[x\la\cT(1^n)]}{\Pr[x\la \cD(1^n,h)]}\right)
-\log\left(\frac{\Pr[x\la\cT(1^n)]}{\Pr[x\la \cD(1^n,z)]}\right)
\right)\\
&=D_{KL}(\cT(1^n)\,\|\, \cD(1^n,h))- D_{KL}(\cT(1^n)\,\|\,\cD(1^n,z))
\end{align}
for all $z\in\bit^n$.
This implies that $h$ satisfies
\begin{align}
D_{KL}(\cT(1^n)\,\|\,\cD(1^n,h))\leq \mathsf{Opt}_n+1/\epsilon(n)
\end{align}
for all sufficiently large $n\in\N$.
This contradicts the security of $\cD$.

\begin{proof}[Proof of \cref{claim:KL}]
From the definition of $\cA$,
for all sufficiently large $n\in\N$,
with probability at least $1-\frac{1}{2\delta(n)}$ over $h\la\cA(1^n,\{x_i\}_{i\in[t(n)]})$, we have
\begin{align}
\frac{1}{t(n)}\sum_{i\in[t(n)]}\log\left(\frac{1}{\Pr[x_i\la\cD(1^n,h)]}\right)+\max_{a\in\bit^n}\frac{1}{t(n)}\sum_{i\in[t(n)]}\log\left(\Pr[x_i\la\cD(1^n,a)]\right)
\leq \frac{1}{2\epsilon(n)}
\end{align}
for all $x\in\bit^{t(n)\cdot p(n)}$ such that $\max_{a\in\bit^n}\Pr[x\la\cD^*(1^n,a)]>0$.
Therefore, for all sufficiently large $n\in\N$, with probability at least $1-\frac{1}{2\delta(n)}$ over $h\la\cA(1^n,\{x_i\}_{i\in[t(n)]})$, for all $z\in\bit^n$, we have
\begin{align}
&\frac{1}{t(n)}\sum_{i\in[t(n)]}\log\left(\frac{1}{\Pr[x_i\la\cD(1^n,h)]}\right)-\frac{1}{t(n)}\sum_{i\in[t(n)]}\log\left(\frac{1}{\Pr[x_i\la\cD(1^n,z)]}\right)\\
&\leq \frac{1}{t(n)}\sum_{i\in[t(n)]}\log\left(\frac{1}{\Pr[x_i\la\cD(1^n,h)]}\right)+\max_{a\in\bit^n}\frac{1}{t(n)}\sum_{i\in[t(n)]}\log\left(\Pr[x_i\la\cD(1^n,a)]\right)\leq \frac{1}{2\epsilon(n)}\label{ineq_6}.
\end{align}

On the other hand,
from \cref{lem:Hoeffding}, with probability at least $1-\frac{1}{2\delta(n)}$ over $\{x_i\}_{i\in[t(n)]}\la \cT(1^n)^{\otimes t(n)}$, 
\begin{align}\label{ineq_7}
\abs{\mathbb{E}_{x\la \cT(1^n)}\left[\log\left(\frac{1}{\Pr[x\la\cD(1^n,z)]}\right) \right]-\frac{1}{t(n)}\sum_{i\in[t(n)]} \log\left(\frac{1}{\Pr[x_i\la\cD(1^n,z)]}\right)}\leq \frac{1}{4\epsilon(n)}
\end{align}
for all $z\in\bit^n$ at the same time.
Here, we have used the fact that
\begin{align}
\log\left(\frac{1}{\Pr[x\la \cD(1^n,z)]}\right)\leq m(n)
\end{align}
for all $x\in\bit^{p(n)}$ and $z\in\bit^n$.
From union bound, \cref{ineq_6,ineq_7} implies that, for all sufficiently large $n\in\N$, with probability at least $1-\frac{1}{\delta(n)}$, over $\{x_i\}_{i\in[t(n)]}\la \cT(1^n)^{\otimes t(n)}$ and $h\la\cA(1^n,\{x_i\}_{i\in[t(n)]})$, we have 
\begin{align}
\mathbb{E}_{x\la \cT(1^n) }\left[\log\left(\frac{1}{\Pr[x\la\cD(1^n,h)]}\right) \right]
-
\mathbb{E}_{x\la \cT(1^n) }\left[\log\left(\frac{1}{\Pr[x\la\cD(1^n,z)]}\right) \right]\leq \frac{1}{\epsilon(n)}
\end{align}
for all $z\in\bit^n$.
Therefore, with probability at least $1-\frac{1}{\delta(n)}$ over $\{x_i\}_{i\in[t(n)]}\la \cT(1^n)^{\otimes t(n)}$ and $h\la\cA(1^n,\{x_i\}_{i\in[t(n)]})$, we have 
\begin{align}
\mathbb{E}_{x\la \cT(1^n) }\left[\log\left(\frac{\Pr[x\la\cD(1^n,z)]}{\Pr[x\la\cD(1^n,h)]}\right) \right]\leq \frac{1}{\epsilon(n)}
\end{align}
for all $z\in\bit^n$.
\end{proof}

\end{proof}

\subsection{Classical Case}
\if0
\color{red}
In this section, we show that $\mathbf{NP} \nsubseteq \mathbf{BPP}$ holds if and only if agnostic classical distribution learning is hard with respect to KL divergence.  
While it is known that $\mathbf{NP} \nsubseteq \mathbf{BPP}$ implies this hardness, the proof is complex, and to the best of our knowledge, the reverse implication is folklore.  
Moreover, the existing techniques are tailored to the classical setting and do not easily extend to the quantum case. 
Therefore, we choose to provide the statement.
Our proof follows the same approach as in \cref{sec:quantum_KL}.

First, we describe the hardness of agnostic classical distribution learning with respect to KL divergence based on \cite{ML:AW92}.
\color{black}

\mor{Intro niwa kaku?}
\taiga{Ueno akaha Intro ni kaku.}
\fi

In this subsection, we apply the previous quantum results to the classical case.

\begin{definition}[Hardness of Agnostic Classical Distribution Learning with respect to KL divergence~\cite{ML:AW92}]
We say that hardness of agnostic classical distribution learning with respect to KL divergence holds if the following holds.
There exist some polynomial $\epsilon$, $\delta$ and $p$, and a PPT algorithm $\cD$, which takes $1^n$ and $z\in\bit^n$ as input, and outputs $x\in\bit^{p(n)}$ such that the following
two conditions are satisfied:
\begin{enumerate}
\item For all $n\in\N$, $z\in\bit^n$, and $x\in\bit^{p(n)}$, we have
\begin{align}
\Pr[x\la\cD(1^n,z)]\neq 0.
\end{align}
\item 
For any PPT algorithm $\cA$ and any polynomial $t$, there exists an algorithm $\cT$, which takes $1^n$ as input and outputs $x\in\bit^{p(n)}$, such that
\begin{align}
\Pr[ 
D_{KL}(\cT(1^n)\,\|\,\cD(1^n,h))\leq \mathsf{Opt}_n+1/\epsilon(n):
\begin{array}{ll}
&\{x_i\}_{i\in[t(n)]}\la\cT(1^n)^{\otimes t(n)}\\
& h\la \cA(1^n,\{x_i\}_{i\in[t(n)]})
\end{array}
]\leq 1-1/\delta(n)
\end{align}
for infinitely many $n\in\N$.
Here, $\mathsf{Opt}_n\seteq \min_{a\in\bit^n}\{D_{KL}(\cT(1^n)\,\|\,\cD(1^n,a))\}$.
\end{enumerate}
\end{definition}

\begin{definition}[Worst-Case Hardness of Classical Maximum Likelihood Estimation]
We say that worst-case hardness of classical maximum likelihood estimation holds if the following holds.
There exist some polynomial $\epsilon$, $\delta$ and $p$, and a PPT algorithm $\cD$, which takes $1^n$ and $z\in\bit^n$ as input, and outputs $x\in\bit^{p(n)}$ such that, for any PPT algorithm $\cA$, there exist $x\in\bit^{p(n)}$ with $\max_{a\in\bit^n}\Pr[x\la\cD(1^n,a)]>0$ such that 
\begin{align}
\Pr[ \frac{\max_{a\in\bit^n} \Pr[x\la\cD(1^n,a)]}{\Pr[x\la\cD(1^n,h)]} \leq 2^{\epsilon(n)} : h\la\cA(1^n,x)]\leq 1-1/\delta(n)
\end{align}
for infinitely many $n\in\N$.
\end{definition}

\begin{theorem}\label{thm:NP_completeness}
The following three conditions are equivalent:
\begin{enumerate}
\item $\mathbf{NP}\nsubseteq \mathbf{BPP}$.
\item Hardness of agnostic classical distribution learning with respect to KL divergence holds.
\item Worst-case hardness of classical maximum likelihood estimation holds.
\end{enumerate}
\end{theorem}
The proof of \cref{thm:NP_completeness} is similar to \cref{thm:PP_completeness}. Therefore, we give only a sketch of the proof.
\begin{proof}[Proof sketch of \cref{thm:NP_completeness}]
We can show that $(2)\Rightarrow(3)$ and $(2)\Leftarrow (3)$ in the same way as \cref{lem:QML_to_KL_Learn} and \cref{lem:KL_learn_to_QML}, respectively.

$(3)\Leftarrow(1)$ follows in the same way as \cref{lem:PP_hardness}.
We can show that $\mathbf{BPP_{path}}\subseteq \mathbf{BPP}$ if worst-case hardness of classical maximum likelihood estimation does not hold.
$\mathbf{BPP_{path}}\subseteq \mathbf{BPP}$ implies $\mathbf{NP}\subseteq\mathbf{BPP}$.

$(1)\Leftarrow (3)$ follows in the same way as \cref{lem:PP_easiness} (For details, see Lemma 4.2 in \cite{HH24}.)
\cite{HH24} proves \cref{lem:PP_easiness} by using \cref{lem:Extrapolation_PP}.
$(1)\Leftarrow (3)$ follows by using \cref{claim:stockmeyer} instead.
\end{proof}

\color{black}

\section{Agnostic Quantum Distribution Learning with Respect to Statistical Distance}\label{sec:total_variation_distance}

\subsection{Upper Bounds of Distribution Learning with Respect to Statistical Distance}

In this section, we introduce the agnostic distribution learning with respect to statistical distance.
Then, we state that if $\mathbf{PP}=\mathbf{BQP}$ (resp. $ \mathbf{NP}\subseteq\mathbf{BPP}$), then a QPT (resp. PPT) algorithm can achieve agnostic quantum (resp. classical) distribution learning with respect to statistical distance.

\begin{definition}[Agnostic Distribution Learning with Respect to Statistical Distance~\cite{STOC:KMRRSS94}]\label{def:ag_SD}
We say that an algorithm $\cA$ achieves agnostic quantum (resp. classical) distribution learning with respect to statistical distance if the following holds:
For any polynomial $\epsilon$ and $\delta$, there exists a polynomial $t$ such that for any QPT (resp. PPT) algorithm $\cD$, which takes $1^n$ and $z\in\bit^n$ as input, and outputs $x\in\bit^{\poly(n)}$ and an algorithm $\cT$, which takes $1^n$ as input and outputs $x\in\bit^{\poly(n)}$, we have
\begin{align}
\Pr[ 
\mathsf{SD}(\cD(1^n,h),\cT(1^n))\leq (3+1/\epsilon(n))\cdot\mathsf{Opt}_n+1/\epsilon(n):
\begin{array}{ll}
&\{x_i\}_{i\in[t(n)]}\la\cT(1^n)^{\otimes t(n)} \\
& h\la \cA(1^n,\{x_i\}_{i\in[t(n)]})
\end{array}
]\geq 1-1/\delta(n)    
\end{align}
for all sufficiently large $n\in\N$.
Here, $\mathsf{Opt}_n\seteq \min_{a\in\bit^n}\mathsf{SD}(\cD(1^n,a),\cT(1^n))$.
\end{definition}
\begin{remark}
In our definition of agnostic distribution learning, the learner needs to output a hypothesis $h$ such that
$
\mathsf{SD}(\cD(1^n,h), \cT(1^n)) \leq (3+1/\epsilon(n)) \cdot \min_{a} \mathsf{SD}(\cD(1^n,a), \cT(1^n)) +1/\epsilon(n).
$
One might wonder why we consider $(3+1/\epsilon(n)) \cdot \min_{a} \mathsf{SD}(\cD(1^n,a), \cT(1^n))$ instead of $\min_{a} \mathsf{SD}(\cD(1^n,a), \cT(1^n))$.
This is because \cite{myBKM19} shows that no (even unbounded) algorithm can achieve a factor strictly smaller than \(3\).
Therefore, we consider $\mathsf{SD}(\cD(1^n,h), \cT(1^n)) \leq (3+1/\epsilon(n)) \cdot \min_{a} \mathsf{SD}(\cD(1^n,a), \cT(1^n)) +1/\epsilon(n)$ in this work.
\end{remark}

We show that a PPT algorithm with access to a $\Sigma_3^{\mathbf{P}}$ oracle achieves agnostic classical distribution learning with respect to statistical distance.
\begin{theorem}\label{thm:NP_learning}
There exists a PPT algorithm with access to a $\Sigma_3^{\mathbf{P}}$ oracle which achieves agnostic classical distribution learning with respect to statistical distance.
\end{theorem}

We also show that a PPT algorithm with access to a $\Sigma_2^{\mathbf{PP}}$ oracle achieves agnostic quantum distribution learning with respect to statistical distance.
\begin{theorem}\label{thm:PP_learn_quantum}
There exists a PPT algorithm with access to a ${\Sigma_2^{\mathbf{PP}}}$ oracle which achieves agnostic quantum distribution learning with respect to statistical distance.
\end{theorem}
We give the proof of \cref{thm:NP_learning} in \cref{sec:Proof_NP_learning}.
The proof of \cref{thm:PP_learn_quantum} is similar to that of \cref{thm:NP_learning}.
For clarity, we give the sketch of the proof in \cref{sec:Proof_PP_learning}.
\cref{thm:NP_learning,thm:PP_learn_quantum} imply the following \cref{cor:NP_learn_classical,cor:PP_learn_quantum}, respectively.

\begin{corollary}\label{cor:NP_learn_classical}
If $\mathbf{NP}\subseteq\mathbf{BPP}$, then there exists a PPT algorithm which achieves agnostic classical distribution learning with respect to the statistical distance.
\end{corollary}

\begin{corollary}\label{cor:PP_learn_quantum}
If $\mathbf{PP}\subseteq\mathbf{BQP}$, then there exists a QPT algorithm which achieves agnostic quantum distribution learning with respect to the statistical distance.
\end{corollary}

\begin{proof}[Proof of \cref{cor:NP_learn_classical}]
If $\mathbf{NP}\subseteq \mathbf{BPP}$, then we have $\Sigma_3^{\mathbf{P}}\subseteq\mathbf{BPP}$.
Therefore, a PPT algorithm can simulate $\Sigma_3^{\mathbf{P}}$ oracle.
Hence, from \cref{thm:NP_learning}, a PPT algorithm can achieve agnostic classical distribution learning with respect to statistical distance assuming $\mathbf{NP}\subseteq \mathbf{BPP}$.
\end{proof}

\begin{proof}[Proof of \cref{cor:PP_learn_quantum}]
This is the same as \cref{cor:NP_learn_classical}.
If $\mathbf{PP}\subseteq\mathbf{BQP}$, then we have $\Sigma_2^{\mathbf{PP}}\subseteq\mathbf{BQP}$.
Therefore, a QPT algorithm can simulate $\Sigma_2^{\mathbf{PP}}$ oracle.
Hence, from \cref{thm:NP_learning}, a QPT algorithm can achieve agnostic quantum distribution learning with respect to statistical distance assuming $\mathbf{PP}\subseteq \mathbf{BQP}$.
\end{proof}

\subsection{Quantum Advantage and Hardness of Agnostic Quantum Distribution Learning}

In this section, we study relation between hardness of agnostic quantum distribution learning and sampling-based quantum advantage.

\begin{assumption}[Hardness of Agnostic Quantum Distribution Learning against PPT with $\Sigma_3^{\mathbf{P}}$ Oracle]\label{assumption:PP_hardness_of_Dis_learn}
We say that hardness of agnostic quantum distribution learning against PPT with access to a $\Sigma_3^{\mathbf{P}}$ oracle holds if the following holds.
There exist some polynomials $\epsilon$ and $\delta$, and a QPT algorithm $\cD$, which takes $1^n$ and $z\in\bit^n$ as input, and outputs $x\in\bit^{\poly(n)}$ such that, for any PPT algorithm $\cA$ with access to a $\Sigma_3^{\mathbf{P}}$ oracle and any polynomial $t$, there exist an algorithm $\cT$, which takes $1^n$ as input and outputs $x\in\bit^{\poly(n)}$, such that
\begin{align}\label{ineq:Sigma_3}
\Pr[ 
\mathsf{SD}(\cD(1^n,h),\cT(1^n))\leq (3+1/\epsilon(n))\cdot\mathsf{Opt}_n+1/\epsilon(n):
\begin{array}{ll}
&\{x_i\}_{i\in[t(n)]}\la\cT(1^n)^{\otimes t(n)} \\
& h\la \cA^{\Sigma_3^{\mathbf{P}}}(1^n,\{x_i\}_{i\in[t(n)]})
\end{array}
]\leq 1-1/\delta(n)
\end{align}
for infinitely many $n\in\N$.
Here, $\mathsf{Opt}_n\seteq \min_{a\in\bit^n}\mathsf{SD}(\cD(1^n,a),\cT(1^n))$.
\end{assumption}

\begin{assumption}[PP-hardness of Agnostic Quantum Distribution Learning with Respect to Statistical Distance]\label{ass:PP-hard}
We say that agnostic quantum distribution learning with respect to statistical distance is $\mathbf{PP}$-hard in a PPT black-box reduction if there exist polynomials $\epsilon$ and $\delta$, and a QPT algorithm $\cD$, which takes $1^n$ and $z\in\bit^n$ as input, and outputs $x\in\bit^{\poly(n)}$, such that the following condition is satisfied:
For any oracle $\mathcal{O}$ such that there exists a polynomial $t$ such that for any algorithm $\cT$, which takes $1^n$ as input, and outputs $x\in\bit^{\poly(n)}$,
\begin{align}
\Pr[ 
\mathsf{SD}(\cD(1^n,h),\cT(1^n))\leq (3+1/\epsilon(n))\cdot\mathsf{Opt}_n+1/\epsilon(n):
\begin{array}{ll}
&\{x_i\}_{i\in[t(n)]}\la\cT(1^n)^{\otimes t(n)} \\
& h\la \cO(1^n,\{x_i\}_{i\in[t(n)]})
\end{array}
]\geq 1-1/\delta(n)
\end{align}
for all sufficiently large $n\in\N$,
we have $\mathbf{PP}\subseteq \mathbf{BPP}^{\mathcal{O}}$.
\end{assumption}

Our first result is that the hardness of learning is derived from 
the standard quantum advantage assumption (\cref{def:QAA}) plus
$\mathbf{PP}\nsubseteq \mathbf{BQP}^{\Sigma_3^{\mathbf{P}}}$.
\begin{theorem}\label{thm:Learning_hardness_is_weaker}
$\mathbf{PP}\nsubseteq \mathbf{BQP}^{\Sigma_3^{\mathbf{P}}}$ and quantum advantage assumption (\cref{def:QAA}) imply
\cref{assumption:PP_hardness_of_Dis_learn}.
\end{theorem}
\begin{remark}
In the above Theorem, we established hardness against a PPT algorithm with access to a $\Sigma_3^{\mathbf{P}}$ oracle.
The same hardness result holds even for a QPT algorithm with access to a $\Sigma_3^{\mathbf{P}}$ oracle.
\end{remark}

We give the proof of \cref{thm:Learning_hardness_is_weaker} in \cref{app:proof_1}.

Our next result is that hardness of agnostic quantum distribution learning with respect to statistical distance implies sampling-based quantum advantage.

\begin{theorem}\label{cor:Q_AD}
\cref{assumption:PP_hardness_of_Dis_learn} implies $\mathbf{SampBQP}\neq \mathbf{SampBPP}$.
\end{theorem}

Our third result is that showing $\mathbf{PP}$-hardness of agnostic quantum distribution learning (in a black-box way) is difficult.
\begin{theorem}\label{cor:PP_harndess}
If agnostic quantum distribution learning with respect to statistical distance is $\mathbf{PP}$-hard in a PPT black-box reduction (\cref{ass:PP-hard}), then $\mathbf{PP}\nsubseteq\mathbf{BPP}^{\Sigma_3^{\mathbf{P}}}$ implies $\mathbf{SampBQP}\neq\mathbf{SampBPP}$.    
\end{theorem}

\begin{proof}[Proof of \cref{cor:Q_AD}]
For contradiction, suppose $\mathbf{SampBQP}=\mathbf{SampBPP}$.
Then, we will show that \cref{assumption:PP_hardness_of_Dis_learn} does not hold.

From the assumption of $\mathbf{SampBQP}=\mathbf{SampBPP}$,
for any polynomial $\epsilon$ and any QPT algorithm $\cD$, which takes $1^n$ and $z\in\bit^n$ as input, and outputs $x\in\bit^{\poly(n)}$, there exists a PPT algorithm $\cC_{\cD,\epsilon}$ that takes $1^n$ and $z\in\bit^n$ as input, and outputs $x\in\bit^{\poly(n)}$ such that
\begin{align}
\mathsf{SD}(\cD(1^n,z),\cC_{\cD,\epsilon}(1^n,z))\leq\frac{1}{100\epsilon(n)}    
\end{align}
for all $z\in\bit^n$ and all $n\in\N$.
From \cref{thm:NP_learning}, for any PPT algorithm $\cC_{\cD,\epsilon}(1^n,\cdot)$,
there exists a PPT algorithm $\cA$ with access to a $\Sigma_3^{\mathbf{P}}$ oracle and a polynomial $t$ such that for any algorithm $\cT(1^n)$, which takes $1^n$ as input and outputs $x\in\bit^{\poly(n)}$, we have 
\begin{align}
 \Pr[
\begin{array}{ll}
\mathsf{SD}(\cT(1^n), \cC_{\cD,\epsilon}(1^n,h))\leq (3+\frac{1}{100\epsilon(n)})\cdot\mathsf{Opt}_n^*+\frac{1}{100\epsilon(n)} 
\end{array}
:
\begin{array}{ll}
\{x_i\}_{i\in[t(n)]}\la \cT(1^n)^{\otimes t(n)}\\
h\la \cA^{\Sigma_3^{\mathbf{P}}}(1^n,\{x_i\}_{i\in[t(n)]})
\end{array}
]
\geq 1-1/\delta(n)
\end{align}
for all sufficiently large $n\in\N$. 
Here, $\mathsf{Opt}_n^*\seteq \min_{a\in\bit^n}\mathsf{SD}(\cT(1^n),\cC_{\cD,\epsilon}(1^n,a))$.

This implies that with probability at least $1-\frac{1}{\delta(n)}$ over $\{x_i\}_{i\in[t(n)]}\la\cT(1^n)^{\otimes t(n)}$ and $h\la\cA^{\Sigma_3^{\mathbf{P}}}(1^n,\{x_i\}_{i\in[t(n)]})$,
\begin{align}
\mathsf{SD}(\cT(1^n),\cD(1^n,h))
&\leq \mathsf{SD}(\cT(1^n),\cC_{\cD,\epsilon}(1^n,h))+\mathsf{SD}(\cD(1^n,h),\cC_{\cD,\epsilon}(1^n,h) )\\
&\leq \left(3+\frac{1}{100\epsilon(n)}\right)\cdot \mathsf{Opt}_n^*+\frac{2}{100\epsilon(n)}\\
&\leq \left(3+\frac{1}{100\epsilon(n)}\right)\cdot \left(\mathsf{Opt}_n +\frac{1}{100\epsilon(n)}\right)+\frac{2}{100\epsilon(n)}\\
&\leq \left(3+\frac{1}{\epsilon(n)}\right)\cdot \mathsf{Opt}_n+\frac{1}{\epsilon(n)}
\end{align}
for all sufficiently large $n\in\N$, where $\mathsf{Opt}_n\seteq \min_{a\in\bit^n}\{\mathsf{SD}\left(\cT(1^n),\cD(1^n,a) \right) \}$.
Here, in the third inequality, we have used $\mathsf{Opt}_n^*\leq \mathsf{Opt}_{n}+\frac{1}{100\epsilon(n)}$.
\if0
%kokoni tyokuzenno hutousikino syoumei aru.
%
This holds because if we denote $\ell$ (resp. $m$) to mean $\mathsf{Opt}_n^*\seteq\mathsf{SD}(\cT(1^n),\cC_{\cD,\epsilon}(1^n,\ell))$ (resp. $\mathsf{Opt}_n\seteq\mathsf{SD}(\cT(1^n),\cD(1^n,m)) $), then we have
\begin{align}
\mathsf{Opt}_n^*&=\mathsf{SD}(\cT(1^n),\cC_{\cD,\epsilon}(1^n,\ell))\\
                &\leq \mathsf{SD}(\cT(1^n),\cC_{\cD,\epsilon}(1^n,m))\\
                &\leq\mathsf{SD}(\cT(1^n),\cD(1^n,m))+\mathsf{SD}(\cC_{\cD,\epsilon}(1^n,m),\cD(1^n,m))\\
                &\leq \mathsf{Opt}_n+\frac{1}{100\epsilon(n)}.
\end{align}
\fi

This is a contradiction to \cref{assumption:PP_hardness_of_Dis_learn}.
\end{proof}

\begin{proof}[Proof of \cref{cor:PP_harndess}]
From \cref{cor:Q_AD}, it is sufficient to prove that \cref{ass:PP-hard} and $\mathbf{PP}\nsubseteq\mathbf{BPP}^{\Sigma_3^{\mathbf{P}}}$ imply \cref{assumption:PP_hardness_of_Dis_learn}.
For contradiction, suppose that \cref{assumption:PP_hardness_of_Dis_learn} does not hold, then we show that either \cref{ass:PP-hard} or $\mathbf{PP}\nsubseteq\mathbf{BPP}^{\Sigma_3^{\mathbf{P}}}$ does not hold.
For showing this, it is sufficient to show that if \cref{assumption:PP_hardness_of_Dis_learn} does not hold, then \cref{ass:PP-hard} implies $\mathbf{PP}\subseteq \mathbf{BPP}^{{\Sigma_3^{\mathbf{P}}}}$.

Because we assume that \cref{assumption:PP_hardness_of_Dis_learn} does not hold,
for any polynomial $\epsilon$ and $\delta$,
there exists a PPT algorithm $\cA$ with access to a $\Sigma_3^{\mathbf{P}}$ oracle such that \cref{ineq:Sigma_3} does not hold for all sufficiently large $n\in\N$.
Furthermore, from \cref{ass:PP-hard},
for any language $\cL\in\mathbf{PP}$,
there exists a PPT algorithm $\cR$ such that the following two conditions are satisfied:
\begin{enumerate}
\item For all $x\in\cL$, we have 
\begin{align}
   \Pr[1\la \cR^{\cA^{\Sigma_3^{\mathbf{P}}}}(x)]\geq \frac{2}{3}.
\end{align}
\item For all $x\notin \cL$, we have
\begin{align}
    \Pr[1\la \cR^{\cA^{\Sigma_3^{\mathbf{P}}}}(x)]\leq \frac{1}{3}.
\end{align}
\end{enumerate}
Then, by simulating $\cR^{\cA^{\Sigma_3^{\mathbf{P}}}}(x)$, a PPT algorithm $\cM$ with access to a $\Sigma_3^{\mathbf{P}}$ oracle can decide if $x$ is in $\cL$ or not.
This implies that $\mathbf{PP}\subseteq\mathbf{BPP}^{\Sigma_3^{\mathbf{P}}}$.
\end{proof}

\if0
\begin{remark}
In the same way as \cref{thm:NP_learning}, we can show \cref{thm:PP_learning}.
\begin{theorem}[Agnostic Distribution Learning for QPT Algorithms]\label{thm:PP_learning}
For any QPT algorithm $\cD$ which takes $1^n$ and $z\in\bit^n$ as input and outputs $x\in\bit^n$, there exists a PPT algorithm with access to a $\mathbf{PP^{NP}}$ oracle $\cL$ and a polynomial $T=\poly(n,\lfloor1/\epsilon\rfloor,\lfloor 1/\delta\rfloor)$ such that, for any distributions $\cT_n$ over $\bit^n$, we have
\begin{align}
    &\Pr[
        \begin{array}{ll}
        \mathsf{SD}(\cT_n, \cD(1^n,h))\leq 3\cdot\mathsf{Opt}+\epsilon
        \end{array}
        :
    \begin{array}{ll}
       \{x_1,...,x_{T} \}\la \cT_n^{\otimes T}\\
        h\la \cL(1^n,1^{\lfloor1/\epsilon\rfloor},1^{\lfloor1/\delta\rfloor},\{x_i\}_{i\in[T]},\mathsf{Opt})
    \end{array}
        ]\\
        &\geq 1-\delta
\end{align}
for all $\epsilon\in\N$ and $\delta\in\N$ and all sufficiently large $n\in\N$.
Here, $\mathsf{Opt}\seteq \min_{a\in\bit^n}\left\{\mathsf{SD}(\cT_n,  \cD(1^n,a))\right\}$.
\end{theorem}
\end{remark}
\fi

\subsection{Proof of \cref{thm:NP_learning}}\label{sec:Proof_NP_learning}
In this section, we give the proof of \cref{thm:NP_learning}.
For showing \cref{thm:NP_learning}, we use the following \cref{thm:NP_distinguish}, which we prove later.
\begin{theorem}\label{thm:NP_distinguish}
For any polynomial $\epsilon$, any PPT algorithm $\cD$ that takes $1^n$ and $z\in\bit^n$ as input and outputs $x\in\bit^{\poly(n)}$, there exists a PPT algorithm $\mathsf{Dis}_{\epsilon}$ with access to an $\mathbf{NP}$ oracle, which takes $1^n$, $\ell\in\bit^n$, $m\in\bit^n$, and $x\in\bit^{\poly(n)}$ as input, and outputs $b\in\bit$, such that for all $(\ell,m)\in\bit^n\times\bit^n$ with $\mathsf{SD}(\cD(1^n,\ell),\cD(1^n,m))\geq 1/\epsilon(n)$, we have
\begin{align}
&\mathsf{SD}(\cD(1^n,\ell),\cD(1^n,m))\\
&\leq \left(1+\frac{1}{4\epsilon(n)}\right) \Bigg|\Pr[1\la\mathsf{Dis}_{\epsilon}^{\mathbf{NP}}(1^n,\ell,m,x):x\la\cD(1^n,\ell)]\\
&\hspace{4cm}-\Pr[1\la\mathsf{Dis}_{\epsilon}^{\mathbf{NP}}(1^n,\ell,m,x):x\la\cD(1^{n},m)]\Bigg|
   +\frac{1}{4\epsilon(n)}
\end{align}
\if0
\begin{align}
   &\mathsf{SD}(\cD(1^n,\ell),\cD(1^n,m))\\
   &\leq \left(1+\frac{1}{4\epsilon(n)}\right) \abs{\Pr_{x\la\cD(1^n,\ell)}[1\la\mathsf{Dis}_{\epsilon}^{\mathbf{NP}}(1^n,\ell,m,x)]-\Pr_{x\la\cD(1^{n},m)}[1\la\mathsf{Dis}_{\epsilon}^{\mathbf{NP}}(1^n,\ell,m,x)]}
   +\frac{1}{4\epsilon(n)}
\end{align}
\fi
for all sufficiently large $n\in\N$.
\end{theorem}

\begin{proof}[Proof of \cref{thm:NP_learning}]
Let $\cD$ be a PPT algorithm that takes $1^n$ and $z\in\bit^n$ as input, and outputs $x\in\bit^{\poly(n)}$.
Let $\epsilon$ and $\delta$ be polynomials.
In the following, we show that there exists a polynomial $t$, and a PPT algorithm $\cA$ with access to a $\Sigma_3^{\mathbf{P}}$ oracle such that, for all algorithms $\cT$, which takes $1^n$ and outputs $x\in\bit^{\poly(n)}$,
we have
\begin{align}
\Pr[ 
\mathsf{SD}(\cD(1^n,h),\cT(1^n))\leq (3+1/\epsilon(n))\cdot\mathsf{Opt}_n+1/\epsilon(n):
\begin{array}{ll}
&\{x_i\}_{i\in[t(n)]}\la\cT(1^n)^{\otimes t(n)} \\
& h\la \cA^{\Sigma_3^{\mathbf{P}}}(1^n,\{x_i\}_{i\in[t(n)]})
\end{array}
]\geq 1-1/\delta(n)    
\end{align}
for all sufficiently large $n\in\N$.
Here, $\mathsf{Opt}_n\seteq \min_{a\in\bit^n}\mathsf{SD}(\cT(1^n),\cD(1^n,a))$.

For simplicity, in the following, we often omit $1^n$ of $\cT(1^n)$, $\cD(1^n,a)$, and $\mathsf{Dis}_{\epsilon}^{\mathbf{NP}}(1^n,a,b,x)$.
We set $t(n)\seteq 100n^2\cdot\epsilon(n)^2\left(n +\log(\delta(n))\right)$.
To construct $\cA$, for any polynomial-time computable function $\omega$, let us consider the following deterministic polynomial-time algorithm $\cM_{\omega}$ with access to an $\mathbf{NP}$ oracle.
\begin{description}
\item[$\cM_{\omega}^{\mathbf{NP}}(1^n,a,b,\{x_i,r_i(1),r_i(2),r_i(3)\}_{i\in[t(n)]})$:] $ $
\begin{enumerate}
\item Run $\mathsf{Dis}_{\epsilon}^{\mathbf{NP}}(a,b, x_i ;r_i(1))$ for all $i\in[t(n)]$.
%Let $\widetilde{\cT}(a,b;\{x_i,r_i(1)\}_{i\in[t(n)]})\seteq\sum_{i\in [t(n)]}\frac{A_i}{t(n)}$.
\item For all $i\in[t(n)]$, run $X_{a,r_i(2)}\la\cD(a;r_i(2))$ and $\mathsf{Dis}_{\epsilon}^{\mathbf{NP}}(a,b, X_{a,r_i(2)} ;r_i(3))$.
\item Output $1$ if 
\begin{align}
\abs{\sum_{i\in[t(n)]}\frac{\mathsf{Dis}_{\epsilon}^{\mathbf{NP}}(a,b, x_i ;r_i(1))}{t(n)}-\sum_{i\in[t(n)]}\frac{\mathsf{Dis}_{\epsilon}^{\mathbf{NP}}(a,b, X_{a,r_i(2)} ;r_i(3))}{t(n)}}\leq\omega(n).
\end{align}
\if0
\begin{align}
\abs{\widetilde{\cT}(a,b;\{x_i,r_i(1)\}_{i\in[t(n)]})-\widetilde{\cD}(a,b;\{r_i(2),r_i(3)\}_{i\in[t(n)]})}\leq\mathsf{Opt}_n+\frac{1}{2}\left(\epsilon_1(n)+\epsilon_2(n)\right).
\end{align}
\fi
Otherwise, output $0$.
\end{enumerate}
\end{description}

For any polynomial-time computable function $\omega$, and $(a_1,...,a_I)\in\bit^{I}$ with $I\in\N$, let us define
\begin{align}
&\cL_{\omega,(a_1,...,a_I)}\\
&\seteq 
\bigg\{ \{x_i,r_i(1),r_i(2),r_i(3)\}_{i\in[t(n)]}\in\bit^*,n\in\N:\exists a\in\bit^{n-I} \mbox{ such that } \forall b\in\bit^n \mbox{, we have}\\
&1\la \cM_{\omega}^{\mathbf{NP}}(1^n,a_1\|...\|a_{I}\|a,b,\{x_i,r_i(1),r_i(2),r_i(3)\}_{i\in[t(n)]})
\bigg\}.
\end{align} 
We have $\cL_{\omega,a_1\|...\|a_I}\in\Sigma_2^{\mathbf{NP}}=\Sigma_3^{\mathbf{P}}$.

Now, we construct the PPT algorithm $\cA$ with access to a $\Sigma_3^{\mathbf{P}}$ oracle.
\begin{description}
\item[$\cA^{\Sigma_3^{\mathbf{P}}}(1^n,\{x_i\}_{i\in[t(n)]})$:] $ $
\begin{enumerate}
\item Uniformly sample random strings $\{r_i(1), r_i(2), r_i(3)\}_{i \in [t(n)]}$.
\item For each index $I \in [n]$, determine the output bit $\mathsf{out}_I$ as follows:
    \begin{enumerate}
    \item Initialize $p(1) = \frac{1}{2}$.
    \item For $j = 1$ to $n$, do:
        \begin{enumerate}
        \item For each $\beta \in \bit$, check whether
        \begin{align}
        (\{x_i, r_i(1), r_i(2), r_i(3)\}_{i \in [t(n)]}, n)
        \in \cL_{p(j), \mathsf{out}_1\| \dots\| \mathsf{out}_{I-1} \beta}.
        \end{align}
        Let $Q_\beta = 1$ if the above holds, and $Q_\beta = 0$ otherwise.
        \item  If $Q_0 = Q_1 = 1$ and $j < n$, then set $p(j+1)=p(j)-\left(\frac{1}{2}\right)^{j+1} $ and continue to the next $j$.

        \item If $Q_0 = Q_1 = 0$ and $j < n$, then set $p(j+1) = p(j)+\left(\frac{1}{2}\right)^{j+1}$ and continue to the next $j$.

        \item  If $Q_b = 1$ and $Q_{b \oplus 1} = 0$ for some $b \in \bit$, then set $\mathsf{out}_I = b$ and break the $j$-loop.

        \item If $j = n$, then set $\mathsf{out}_I = 1$, and break the $j$-loop.
        \end{enumerate}
    \end{enumerate}
\item Output the string $\mathsf{out}_1\| \dots\| \mathsf{out}_n$.
\end{enumerate}
\end{description}

We use the following \cref{claim:A_output}.
\begin{claim}\label{claim:A_output}
With probability at least $(1-1/\delta(n))$ over $\{x_i\}_{i\in[t(n)]}\la\cT^{\otimes t(n)}$ and $\mathsf{out}\la\cA^{\Sigma_3^\mathbf{P}}(1^n,\{x_i\}_{i\in[t(n)]})$, $\mathsf{out}$ satisfies
\begin{align}\label{ineq:out}
\abs{\Pr[1\la\mathsf{Dis}_{\epsilon}^{\mathbf{NP}}(\mathsf{out},b,x):x\la\cT]-\Pr[1\la\mathsf{Dis}_{\epsilon}^{\mathbf{NP}}(\mathsf{out},b,x) :x\la\cD(\mathsf{out})]}\leq \mathsf{Opt}_n+\frac{1}{2\epsilon(n)}
\end{align}
for all $b\in\bit^n$.
Here, $\mathsf{Opt}_n\seteq \min_{a\in\bit^n}\mathsf{SD}(\cT,\cD(a))$.
\end{claim}
We defer the proof of \cref{claim:A_output}.

For simplicity, in the following, let us denote $\ell$ to mean $\ell\seteq\mathsf{argmin}_{a\in\bit^n}\{\mathsf{SD}(\cD(a),\cT)\}$ so that $\mathsf{Opt}_n =\mathsf{SD}(\cD(\ell),\cT)$.
In the following, we show that, for the $\mathsf{out}$ such that \cref{ineq:out} holds for all $b\in\bit^n$, we have
\begin{align}
\mathsf{SD}(\cT,\cD(\mathsf{out}))\leq (3+1/\epsilon(n))\cdot \mathsf{Opt}_n+1/\epsilon(n).
\end{align}

First, let us consider the case where $\mathsf{SD}(\cD(\ell),\cD(\mathsf{out}))<  1/\epsilon(n)$.
In this case, we have
\begin{align}
\mathsf{SD}(\cT,\cD(\mathsf{out}))\leq\mathsf{SD}(\cT,\cD(\ell))+\mathsf{SD}(\cD(\ell),\cD(\mathsf{out})) < \mathsf{Opt}_n+\frac{1}{\epsilon(n)}.
\end{align}
Here, in the first inequality, we have used triangle inequality.

Now, let us consider the case, where $\mathsf{SD}(\cD(\ell),\cD(\mathsf{out}))\geq  1/\epsilon(n)$.
In this case, we have
\begin{align}
    &\mathsf{SD}(\cT,\cD(\mathsf{out}))\\
    &\leq \mathsf{SD}(\cT,\cD(\ell))+\mathsf{SD}(\cD(\ell),\cD(\mathsf{out})))\\
    &\leq \mathsf{Opt}_n+\frac{1}{4\epsilon(n)}\label{ineq_8}\\
    &+\left(1+\frac{1}{4\epsilon(n)}\right)\abs{\Pr[1\la\mathsf{Dis}_{\epsilon}^{\mathbf{NP}}(\mathsf{out},\ell,x):x\la\cD(\ell)]-\Pr[1\la\mathsf{Dis}_{\epsilon}^{\mathbf{NP}}(\mathsf{out},\ell,x):x\la\cD(\mathsf{out})]}\\
    &\leq\mathsf{Opt}_n+\frac{1}{4\epsilon(n)}\label{ineq:triangle}\\
    &+\left(1+\frac{1}{4\epsilon(n)}\right)\abs{\Pr[1\la\mathsf{Dis}_{\epsilon}^{\mathbf{NP}}(\mathsf{out},\ell,x):x\la\cD(\mathsf{out})]-\Pr[1\la\mathsf{Dis}_{\epsilon}^{\mathbf{NP}}(\mathsf{out},\ell,x):x\la\cT]}\\
    &+\left(1+\frac{1}{4\epsilon(n)}\right)\abs{\Pr[1\la\mathsf{Dis}_{\epsilon}^{\mathbf{NP}}(\mathsf{out},\ell,x):x\la\cT]-\Pr[1\la\mathsf{Dis}_{\epsilon}^{\mathbf{NP}}(\mathsf{out},\ell,x):x\la\cD(\ell)]}\\
    &\leq \mathsf{Opt}_n+\frac{1}{4\epsilon(n)}+\left(1+\frac{1}{4\epsilon(n)}\right)\left(\mathsf{Opt}_n+\frac{1}{2\epsilon(n)}\right)+\left(1+\frac{1}{4\epsilon(n)}\right) \mathsf{Opt}_n\label{ineq_9}\\
    &\leq \left(3+\frac{1}{\epsilon(n)}\right)\cdot\mathsf{Opt}_n+\frac{1}{\epsilon(n)}.
\end{align}
Here, in \cref{ineq_8}, we have used \cref{thm:NP_distinguish}, in \cref{ineq:triangle}, we have used the triangle inequality, and
in \cref{ineq_9}, we have used $\mathsf{SD}(\cT(1^n),\cD(\ell))\leq \mathsf{Opt}_n$
and $\mathsf{out}$ satisfies the following inequality
\begin{align}
\abs{\Pr[1\la\mathsf{Dis}^{\mathbf{NP}}(\mathsf{out},b,x):x\la\cT]-\Pr[1\la\mathsf{Dis}^{\mathbf{NP}}(\mathsf{out},b,x):x\la\cD(\mathsf{out})]}\leq \mathsf{Opt}_n+\frac{1}{2\epsilon(n)}
\end{align}
for all $b\in\bit^n$.
This completes the proof.

Now, let us show the \cref{claim:A_output}, which we deferred the proof.
For showing it, we use the following \cref{claim:behave_well}, which can be shown by the standard average argument.
We defer the proof of \cref{claim:behave_well} in the end of the proof.
\begin{claim}\label{claim:behave_well}
Let us denote $X_{a,r_i(2)}$ to mean the output of $\cD(a;r_i(2))$.
Let $p(n)$ be the size of $\{r_i(1),r_i(2),r_i(3)\}_{i\in[t(n)]}$.
Let $\mathsf{Good}_n$ be a set of $\{x_i,r_i(1),r_i(2),r_i(3)\}_{i\in[t(n)]}$ such that, for all $I\in[n]$ and for all $(a_1,...,a_{I-1})\in\bit^{I-1}$, the following two conditions hold.
\begin{enumerate}
\item 
If there exists $(a_{I},...,a_{n})\in\bit^{n-I+1}$ such that
\begin{align}
&\Bigg|\Pr[1\la\mathsf{Dis}_{\epsilon}^{\mathbf{NP}}(a_1\|...\|a_n,b,x):x\la\cT]-\Pr[1\la\mathsf{Dis}_{\epsilon}^{\mathbf{NP}}(a_1\|...\|a_n,b,x):x\la\cD(a_1\|...\|a_n)]\Bigg|\\
&\leq \mathsf{Opt}_n+\frac{I-1}{2n\cdot\epsilon(n)}
\end{align}
for all $b\in\bit^n$,
then there exists $(a_{I+1},...,a_n)\in\bit^{n-I+1}$ such that
\begin{align}
\abs{\sum_{i\in[t(n)]}\frac{\mathsf{Dis}_{\epsilon}^{\mathbf{NP}}(a_1\|...\|a_n,b, x_i ;r_i(1))}{t(n)} -\sum_{i\in[t(n)]}\frac{\mathsf{Dis}_{\epsilon}^{\mathbf{NP}}(a_1\|...\|a_n,b, X_{a,r_i(2)} ;r_i(3))}{t(n)}}\leq \mathsf{Opt}_n+\frac{\left(4I-3\right)}{8n\cdot\epsilon(n)}
\end{align}
 for all $b\in\bit^n$.
\item 
If for all $(a_{I+1},...,a_{n})\in\bit^{n-I+1}$, there exists $b\in\bit^n$ such that 
\begin{align}
&\Bigg|\Pr[1\la\mathsf{Dis}_{\epsilon}^{\mathbf{NP}}(a_1\|...\|a_n,b,x):x\la\cT]-\Pr[1\la\mathsf{Dis}_{\epsilon}^{\mathbf{NP}}(a_1\|...\|a_n,b,x):x\la\cD(a_1\|...\|a_n)]\Bigg|\\
&\geq \mathsf{Opt}_n+\frac{I}{2n\cdot\epsilon(n)},
\end{align}
then for all $(a_{I+1},...,a_n)\in\bit^{n-I+1}$, there exists $b\in\bit^n$ such that
\begin{align}
\abs{\sum_{i\in[t(n)]}\frac{\mathsf{Dis}_{\epsilon}^{\mathbf{NP}}(a_1\|...\|a_n,b, x_i ;r_i(1))}{t(n)} -\sum_{i\in[t(n)]}\frac{\mathsf{Dis}_{\epsilon}^{\mathbf{NP}}(a_1\|...\|a_n,b, X_{a,r_i(2)} ;r_i(3))}{t(n)}}\geq \mathsf{Opt}_n+\frac{\left(4I-1\right)}{8n\cdot\epsilon(n)}.
\end{align}
\end{enumerate}
Then, we have
\begin{align}
\Pr[\{x_i,r_i(1),r_i(2),r_i(3)\}_{i\in[t(n)]}\in\mathsf{Good}_n : 
\begin{array}{ll}
     & \{x_i\}_{i\in[t(n)]}\la\cT^{\otimes t(n)} \\
     & \{r_i(1),r_i(2),r_i(3)\}_{i\in[t(n)]}\la\bit^{p(n)}
\end{array}
]\geq 1-1/\delta(n).
\end{align}
\end{claim}

\begin{proof}[Proof of \cref{claim:A_output}]

Let $\mathsf{Good}_n$ be a set of $\{x_i,r_i(1),r_i(2),r_i(3)\}_{i\in[t(n)]}$ defined in \cref{claim:behave_well}.
In the following, we show that, for all $\{x_i,r_i(1),r_i(2),r_i(3)\}_{i\in[t(n)]}\in\mathsf{Good}_n$,  $\cA^{\Sigma_3^{\mathbf{P}}}(1^n,\{x_i\}_{i\in[t(n)]};\{r_i(1),r_i(2),r_i(3)\}_{i\in[t(n)]})$ outputs $\mathsf{out}\seteq\mathsf{out}_1\|...\|\mathsf{out}_n$
such that 
\begin{align}
\abs{\Pr[1\la\mathsf{Dis}_{\epsilon}^{\mathbf{NP}}(\mathsf{out},b,x):x\la\cT]-\Pr[1\la\mathsf{Dis}_{\epsilon}^{\mathbf{NP}}(\mathsf{out},b,x) :x\la\cD(\mathsf{out})]}\leq \mathsf{Opt}_n+\frac{1}{2\epsilon(n)}
\end{align}
for all $b\in\bit^n$.
This can be shown by induction as follows.

First, we show that for the $\mathsf{out}_1$, there exists $(a_2,...,a_n)\in\bit^{n-1}$ such that
\begin{align}\label{ineq_induction}
&\abs{\Pr[1\la\mathsf{Dis}_{\epsilon}^{\mathbf{NP}}(\mathsf{out}_1\|a_2\|...\|a_n,b,x):x\la\cT]-\Pr[1\la\mathsf{Dis}_{\epsilon}^{\mathbf{NP}}(\mathsf{out}_1\|a_2\|...\|a_n,b,x) :x\la\cD(\mathsf{out}_1\|a_2\|...\|a_n)]}\\
&\leq \mathsf{Opt}_n+\frac{1}{2n\cdot\epsilon(n)}
\end{align}
for all $b\in\bit^n$.
From the definition of $\mathsf{Opt}_n$, 
there exists 
$a\in\bit^{n}$ such that
\begin{align}
\abs{\Pr[1\la\mathsf{Dis}_{\epsilon}^{\mathbf{NP}}(a,b,x):x\la\cT]-\Pr[1\la\mathsf{Dis}_{\epsilon}^{\mathbf{NP}}(a,b,x) :x\la\cD(a)]}\leq \mathsf{Opt}_n
\end{align}
for all $b\in\bit^n$.
From the first item of \cref{claim:behave_well}, this implies that
there exists 
$a\in\bit^{n}$ such that
\begin{align}
&\Bigg|\sum_{i\in[t(n)]}\frac{\mathsf{Dis}_{\epsilon}^{\mathbf{NP}}(a,b, x_i ;r_i(1))}{t(n)}-\sum_{i\in[t(n)]}\frac{\mathsf{Dis}_{\epsilon}^{\mathbf{NP}}(a,b, X_{a,r_i(2)} ;r_i(3))}{t(n)}
\Bigg|\leq \mathsf{Opt}_n+\frac{1}{8n\cdot\epsilon(n)}
\end{align}
for all $b\in\bit^n$.
From the construction of $\cA^{\Sigma_3^{\mathbf{P}}}$, for the $\mathsf{out}_1\in\bit$, there exists 
$(a_{2},...,a_n)\in\bit^{n-1}$ 
such that
\begin{align}
&\Bigg|\sum_{i\in[t(n)]}\frac{\mathsf{Dis}_{\epsilon}^{\mathbf{NP}}(\mathsf{out}_1\|a_2\|...\|a_n,b, x_i ;r_i(1))}{t(n)}\\
&\hspace{1cm}-\sum_{i\in[t(n)]}\frac{\mathsf{Dis}_{\epsilon}^{\mathbf{NP}}(\mathsf{out}_1\|a_2\|...\|a_n,b, X_{\mathsf{out}_1\|a_2\|...\|a_n,r_i(2)} ;r_i(3))}{t(n)}
\Bigg|\leq \mathsf{Opt}_n+\frac{3}{8n\cdot\epsilon(n)}
\end{align}
for all $b\in\bit^n$.
From the second item of \cref{claim:behave_well}, this implies that, for such $\mathsf{out}_1$, there exists $(a_{2},...,a_n)\in\bit^{n-1}$ such that 
\cref{ineq_induction} holds
for all $b\in\bit^n$.

For induction,
suppose that for the $(\mathsf{out}_1,...,\mathsf{out}_{I})\in\bit^{I}$, there exists $(a_{I+1}, ..., a_n)\in\bit^{n-I}$ such that 
\begin{align}
&\Biggl|\Pr[1\la\mathsf{Dis}_{\epsilon}^{\mathbf{NP}}(\mathsf{out}_1\|...\|\mathsf{out}_{I}\|a_{I+1}\|...\|a_n,b, x) :x\la\cT]\\
&\hspace{1cm}-\Pr[1\la\mathsf{Dis}_{\epsilon}^{\mathbf{NP}}(\mathsf{out}_1\|...\|\mathsf{out}_{I}\| a_{I+1}\|...\|a_n,b, x) :x\la\cD(\mathsf{out}_1\|...\|\mathsf{out}_{I}\|a_{I+1}\|...\|a_n)]
\Biggl|
\leq\mathsf{Opt}_n+\frac{I}{2n\cdot \epsilon(n)}
\end{align}
for all $b\in\bit^n$.
From the first item of \cref{claim:behave_well}, there exists 
$(a_{I+1},...,a_n)\in\bit^{n-I}$ such that
\begin{align}
&\Bigg|\sum_{i\in[t(n)]}\frac{\mathsf{Dis}_{\epsilon}^{\mathbf{NP}}(\mathsf{out}_1\|...\|\mathsf{out}_{I}\|a_{I+1}\|...\|a_n,b, x_i ;r_i(1))}{t(n)} \\
&\hspace{1cm}-\sum_{i\in[t(n)]}\frac{\mathsf{Dis}_{\epsilon}^{\mathbf{NP}}(\mathsf{out}_1\|...\|\mathsf{out}_{I}\|a_{I+1}\|...\|a_n,b, X_{\mathsf{out}_1\|...\|\mathsf{out}_{I}\|a_{I+1}\|...\|a_n,r_i(2)} ;r_i(3))}{t(n)}
\Bigg|
\leq \mathsf{Opt}_n+\frac{4I+1}{8n\cdot\epsilon(n)}
\end{align}
for all $b\in\bit^n$.
From the construction of $\cA^{\Sigma_3^{\mathbf{P}}}$, for the $(\mathsf{out}_1,...,\mathsf{out}_{I+1})$, there exist $(a_{I+2},...,a_n)\in\bit^{n-I-1}$ such that
\begin{align}
&\Bigg|\sum_{i\in[t(n)]}\frac{\mathsf{Dis}_{\epsilon}^{\mathbf{NP}}(\mathsf{out}_1\|...\|\mathsf{out}_{I+1}\|a_{I+2}\|...\|a_n,b, x_i ;r_i(1))}{t(n)} \\
&\hspace{1cm}-\sum_{i\in[t(n)]}\frac{\mathsf{Dis}_{\epsilon}^{\mathbf{NP}}(\mathsf{out}_1\|...\|\mathsf{out}_{I+1}\|a_{I+2}\|...\|a_n,b, X_{\mathsf{out}_1\|...\|\mathsf{out}_{I+1}\|a_{I+2}\|...\|a_n,r_i(2)} ;r_i(3))}{t(n)}
\Bigg|
\leq \mathsf{Opt}_n+\frac{4I+3}{8n\cdot\epsilon(n)}
\end{align}
for all $b\in\bit^n$.
From the second item of \cref{claim:behave_well}, this implies that, for the $\mathsf{out}_{1},...,\mathsf{out}_{I+1}$, there exists 
$(a_{I+2},...,a_n)\in\bit^{n-I-1}$ such that
\begin{align}
&\Biggl|\Pr[1\la\mathsf{Dis}_{\epsilon}^{\mathbf{NP}}(\mathsf{out}_1\|...\|\mathsf{out}_{I+1}\| a_{I+2}\|...\|a_n,b, x) :x\la\cT]\\
&\hspace{0.5cm}-\Pr[1\la\mathsf{Dis}_{\epsilon}^{\mathbf{NP}}(\mathsf{out}_1\|...\|\mathsf{out}_{I+1}\| a_{I+2}\|...\|a_n,b, x) :x\la\cD(\mathsf{out}_1\|...\|\mathsf{out}_{I+1}\|a_{I+2}\|...\|a_n)]
\Biggl|
\leq\mathsf{Opt}_n+\frac{I+1}{2n\cdot \epsilon(n)}
\end{align}
for all $b\in\bit^n$.

From induction, this implies that, for all $\{x_i,r_i(1),r_i(2),r_i(3)\}_{i\in[t(n)]}\in\mathsf{Good}_n$, $\cA^{\Sigma_3^{\mathbf{P}}}(1^n,\{x_i\}_{i\in[t(n)]}\allowbreak;\{r_i(1),r_i(2),r_i(3)\}_{i\in[t(n)]})$ outputs
$\mathsf{out}_n\seteq \mathsf{out}_1\|...\|\mathsf{out}_n$ such that
\begin{align}
\abs{\Pr[1\la\mathsf{Dis}_{\epsilon}^{\mathbf{NP}}(\mathsf{out},b,x):x\la\cT]-\Pr[1\la\mathsf{Dis}_{\epsilon}^{\mathbf{NP}}(\mathsf{out},b,x) :x\la\cD(\mathsf{out})]}\leq \mathsf{Opt}_n+\frac{1}{2\epsilon(n)}
\end{align}
for all $b\in\bit^n$.
Because we have
\begin{align}
\Pr[\{x_i,r_i(1),r_i(2),r_i(3)\}_{i\in[t(n)]}\in\mathsf{Good}_n : 
\begin{array}{ll}
     & \{x_i\}_{i\in[t(n)]}\la\cT^{\otimes t(n)} \\
     & \{r_i(1),r_i(2),r_i(3)\}_{i\in[t(n)]}\la\bit^{p(n)}
\end{array}
]\geq 1-1/\delta(n),
\end{align}
this completes the proof.
\end{proof}

\begin{proof}[Proof of \cref{claim:behave_well}]

From Hoeffding's inequality,
for each fixed $(a,b)\in\bit^n\times\bit^n$, with probability at least $1-2\cdot 2^{\left(\frac{-t(n)}{32n^2\cdot\epsilon(n)^2}\right)}$ over $\{x_i\}_{i\in[t(n)]}\la \cT^{\otimes t(n)}$ and $\{r_i(1),r_i(2),r_i(3)\}_{i\in[t(n)]}\la \bit^{p(n)}$, where $p(n)$ is the size of $\{r_i(1),r_i(2),r_i(3)\}_{i\in[t(n)]}$, we have
\begin{align}
&\Biggl|\left(\sum_{i\in[t(n)]}\frac{\mathsf{Dis}_{\epsilon}^{\mathbf{NP}}(a,b, x_i ;r_i(1))}{t(n)} -\sum_{i\in[t(n)]}\frac{\mathsf{Dis}_{\epsilon}^{\mathbf{NP}}(a,b, X_{a,r_i(2)} ;r_i(3))}{t(n)}\right)\\
&- \left(\Pr[1\la\mathsf{Dis}_{\epsilon}^{\mathbf{NP}}(a,b,x):x\la\cT]- \Pr[1\la\mathsf{Dis}_{\epsilon}^{\mathbf{NP}}(a,b,x):x\la\cD(a)] \right)\Biggl|\leq \frac{1}{8n\cdot\epsilon(n)}.
\end{align}
Therefore, from union bound, with probability at least $1-1/\delta(n)$ over $\{x_i\}_{i\in[t(n)]}\la \cT^{\otimes t(n)}$ and $\{r_i(1),r_i(2),r_i(3)\}_{i\in[t(n)]}\la \bit^{p(n)}$, we have
\begin{align}
&\Biggl|\left(\sum_{i\in[t(n)]}\frac{\mathsf{Dis}_{\epsilon}^{\mathbf{NP}}(a,b, x_i ;r_i(1))}{t(n)} -\sum_{i\in[t(n)]}\frac{\mathsf{Dis}_{\epsilon}^{\mathbf{NP}}(a,b, X_{a,r_i(2)} ;r_i(3))}{t(n)}\right)\\
&- \left(\Pr[1\la\mathsf{Dis}_{\epsilon}^{\mathbf{NP}}(a,b,x):x\la\cT]- \Pr[1\la\mathsf{Dis}_{\epsilon}^{\mathbf{NP}}(a,b,x):x\la\cD(a)] \right)\Biggl|\leq \frac{1}{8n\cdot\epsilon(n)}\label{ineq:good}
\end{align}
for all $(a,b)\in\bit^n\times\bit^n$.
Let $\mathsf{Good}_{n}$ be a set of $\{x_i,r_i(1),r_i(2),r_i(3)\}_{i\in[t(n)]}$ such that
\cref{ineq:good} holds for all $(a,b)\in\bit^n\times\bit^n$.
From \cref{ineq:good}, for all $\{x_i,r_i(1),r_i(2),r_i(3)\}_{i\in[t(n)]}\in\mathsf{Good}_{n}$ and 
for all $I\in[n]$,
if there exists $a\in\bit^n$ such that, for all $b\in\bit^n$,
\begin{align}
\abs{\Pr[1\la\mathsf{Dis}_{\epsilon}^{\mathbf{NP}}(a,b,x):x\la\cT]-\Pr[1\la\mathsf{Dis}_{\epsilon}^{\mathbf{NP}}(a,b,x):x\la\cD(a)]}\leq \mathsf{Opt}_n+\frac{I-1}{2n\cdot\epsilon(n)},
\end{align}
then there exists $a\in\bit^n$ such that, for all $b\in\bit^n$,
\begin{align}
\abs{\sum_{i\in[t(n)]}\frac{\mathsf{Dis}_{\epsilon}^{\mathbf{NP}}(a,b, x_i ;r_i(1))}{t(n)} -\sum_{i\in[t(n)]}\frac{\mathsf{Dis}_{\epsilon}^{\mathbf{NP}}(a,b, X_{a,r_i(2)} ;r_i(3))}{t(n)}}\leq \mathsf{Opt}_n+\frac{\left(4I-1\right)}{8n\cdot\epsilon(n)}.
\end{align}

From \cref{ineq:good}, for all $\{x_i,r_i(1),r_i(2),r_i(3)\}_{i\in[t(n)]}\in\mathsf{Good}_{n}$, and for all $I\in[n]$, if for all $a\in\bit^n$, there exists $b\in\bit^n$ such that
\begin{align}
\abs{\Pr[1\la\mathsf{Dis}_{\epsilon}^{\mathbf{NP}}(a,b,x):x\la\cT]-\Pr[1\la\mathsf{Dis}_{\epsilon}^{\mathbf{NP}}(a,b,x):x\la\cD(a)]}\geq \mathsf{Opt}_n+\frac{I}{2n\cdot\epsilon(n)},
\end{align}
then for all $a\in\bit^n$, there exists $b\in\bit^n$ such that
\begin{align}
\abs{\sum_{i\in[t(n)]}\frac{\mathsf{Dis}_{\epsilon}^{\mathbf{NP}}(a,b, x_i ;r_i(1))}{t(n)} -\sum_{i\in[t(n)]}\frac{\mathsf{Dis}_{\epsilon}^{\mathbf{NP}}(a,b, X_{a,r_i(2)} ;r_i(3))}{t(n)}}\geq \mathsf{Opt}_n+\frac{(4I-1)}{8n\cdot\epsilon(n)}.
\end{align}
This completes the proof.
\end{proof}
\end{proof}

\paragraph{Proof of \cref{thm:NP_distinguish}}
In the following, we give a proof of \cref{thm:NP_distinguish}.
\begin{proof}[Proof of \cref{thm:NP_distinguish}]
In the following, we often omit $1^n$ of $\cD(1^n,a)$, and $\mathsf{Estimate}_{\epsilon}(1^n,a,x)$.

Let $\epsilon$ be an arbitrary polynomial.
Let $\mathsf{Estimate}^{\mathbf{NP}}_{\epsilon}$ be a PPT algorithm with access to an $\mathbf{NP}$ oracle given in \cref{thm:estimate} such that
\begin{align}
\Pr\left[\left(1-\frac{1}{\epsilon(n)}\right)\Pr[x\la\cD(z)]\leq \mathsf{Estimate}_{\epsilon}^{\mathbf{NP}}(z,x) \leq \left(1+\frac{1}{\epsilon(n)}\right)\Pr[x\la\cD(z)]\right]\geq 1-\frac{1}{\epsilon(n)}
\end{align}
for all $(x,z)\in\bit^n\times\bit^n$, where the probability is taken over the internal randomness of $\mathsf{Estimate}^{\mathbf{NP}}_{\epsilon}$.

From $\mathsf{Estimate}^{\mathbf{NP}}_{\epsilon}$, 
we construct $\mathsf{Dis}^{\mathbf{NP}}_{\epsilon}$ as follows:
\begin{description}
\item[$\mathsf{Dis}_{\epsilon}^{\mathbf{NP}}(\ell,m,x)$:]$ $
\begin{enumerate}
\item Run $E_\ell\la\mathsf{Estimate}_{500\epsilon}^{\mathbf{NP}}(\ell,x)$ and $E_m\la\mathsf{Estimate}_{500\epsilon}^{\mathbf{NP}}(m,x)$.
\item 
Output $1$ if $E_\ell\geq (1+\frac{1}{16\epsilon(n)})E_{m}$.
Otherwise, output $0$.
\end{enumerate}
\end{description}
We will show that, when $\mathsf{SD}(\cD(\ell),\cD(m))\geq\frac{1}{\epsilon(n)}$, we have
\begin{align}
&\mathsf{SD}(\cD(\ell),\cD(m))\\
&\leq
\left(1+\frac{1}{4\epsilon(n)}\right)\abs{
\Pr[1\la\mathsf{Dis}_{\epsilon}^{\mathbf{NP}}(\ell,m,x):x\la\cD(\ell)] 
-
\Pr[1\la\mathsf{Dis}_{\epsilon}^{\mathbf{NP}}(\ell,m,x):x\la\cD(m)] }+\frac{1}{4\epsilon(n)}.
\end{align}

For showing this, let us define
\begin{align}
&\cS_n\seteq\left\{x\in\bit^{\poly(n)}: \Pr[x\la\cD(1^n,\ell)] \geq \left(1+\frac{1}{8\epsilon(n)}\right)\Pr[x\la\cD(1^n,m)]\right\}\\
&\cT_n\seteq\left\{x\in\bit^{\poly(n)}:\left(1+\frac{1}{8\epsilon(n)}\right)\cdot \Pr[x\la\cD(1^n,m)]> \Pr[x\la\cD(1^n,\ell)]\geq \Pr[x\la\cD(1^n,m)]\right\}\\
&\cU_n\seteq \{x\in\bit^{\poly(n)}: \Pr[x\la\cD(1^n,m)]> \Pr[x\la\cD(1^n,\ell)]\}.
\end{align}
We use the following \cref{fact:distance,claim:Distinguish_well_in_S,claim:Not_fooled_in_U}, which we prove later.

\begin{claim}\label{claim:Distinguish_well_in_S}
For all $x\in\cS_n$, we have
\begin{align}
\Pr[1\la\mathsf{Dis}_{\epsilon}^{\mathbf{NP}}(\ell,m,x)]\geq 1-\frac{1}{250\epsilon(n)}.
\end{align}
\end{claim}

\begin{claim}\label{claim:Not_fooled_in_U}
For all $x\in\cU_n$, we have
\begin{align}
\Pr[1\la\mathsf{Dis}_{\epsilon}^{\mathbf{NP}}(\ell,m,x)]\leq \frac{1}{250\epsilon(n)}.
\end{align}
\end{claim}

\begin{claim}\label{fact:distance}
For all $\cD(\ell)$ and $\cD(m)$ such that 
$
\mathsf{SD}(\cD(\ell),\cD(m))\geq \frac{1}{\epsilon(n)}
$
we have
\begin{align}
\sum_{x\in\cS_n}\Pr[x\la\cD(\ell)]-\Pr[x\la\cD(m)]\geq \mathsf{SD}(\cD(\ell),\cD(m))-\frac{1}{8\epsilon(n)}.
\end{align}
\end{claim}

From \cref{fact:distance,claim:Distinguish_well_in_S,claim:Not_fooled_in_U}, we have
\begin{align}
&\Pr[1\la\mathsf{Dis}_{\epsilon}^{\mathbf{NP}}(\ell,m,x):x\la\cD(\ell)]
-
\Pr[1\la\mathsf{Dis}_{\epsilon}^{\mathbf{NP}}
(\ell,m,x):x\la\cD(m)]\\
&=\sum_{x\in\cS_n}\left(\Pr[x\la\cD(\ell)]-\Pr[x\la\cD(m)]\right)\cdot\Pr[1\la\mathsf{Dis}_{\epsilon}^{\mathbf{NP}}(\ell,m,x)]\\
&+\sum_{x\in\cT_n}\left(\Pr[x\la\cD(\ell)]-\Pr[x\la\cD(m)]\right)\cdot\Pr[1\la\mathsf{Dis}_{\epsilon}^{\mathbf{NP}}(\ell,m,x)]\\
&+\sum_{x\in\cU_n}\left(\Pr[x\la\cD(\ell)]-\Pr[x\la\cD(m)]\right)\cdot\Pr[1\la\mathsf{Dis}_{\epsilon}^{\mathbf{NP}}(\ell,m,x)]\\
&\geq \sum_{x\in\cS_n}\left(\Pr[x\la\cD(\ell)]-\Pr[x\la\cD(m)]\right)\cdot\Pr[1\la\mathsf{Dis}_{\epsilon}^{\mathbf{NP}}(\ell,m,x)]\label{ineq_10}\\
&+\sum_{x\in\cU_n}\left(\Pr[x\la\cD(\ell)]-\Pr[x\la\cD(m)]\right)\cdot\Pr[1\la\mathsf{Dis}_{\epsilon}^{\mathbf{NP}}(\ell,m,x)]\\
&\geq \sum_{x\in\cS_n}\left(\Pr[x\la\cD(\ell)]-\Pr[x\la\cD(m)]\right)\cdot\left(1-\frac{1}{250\epsilon(n)}\right)\label{ineq_11} \\
&+\sum_{x\in\cU_n}\left(\Pr[x\la\cD(\ell)]-\Pr[x\la\cD(m)]\right)\cdot\Pr[1\la\mathsf{Dis}_{\epsilon}^{\mathbf{NP}}(\ell,m,x)]\\
&\geq \sum_{x\in\cS_n}\left(\Pr[x\la\cD(\ell)]-\Pr[x\la\cD(m)]\right)\cdot\left(1-\frac{1}{250\epsilon(n)}\right) 
-\sum_{x\in\cU_n}\Pr[x\la\cD(m)]\cdot\Pr[1\la\mathsf{Dis}_{\epsilon}^{\mathbf{NP}}(\ell,m,x)]\\
&\geq \sum_{x\in\cS_n}\left(\Pr[x\la\cD(\ell)]-\Pr[x\la\cD(m)]\right)\cdot\left(1-\frac{1}{250\epsilon(n)}\right) 
-\frac{1}{250\epsilon(n)}\label{ineq_12}\\
&\geq \left(\mathsf{SD}(\cD(\ell),\cD(m))-\frac{1}{8\epsilon(n)}\right)\left(1-\frac{1}{250\epsilon(n)}\right)-\frac{1}{250\epsilon(n)}\label{ineq_13}.
\end{align}
Here, 
in \cref{ineq_10}, we have used that $\Pr[x\la\cD(\ell)]-\Pr[x\la\cD(m)]\geq 0$ for all $x\in\cT_n$,
in \cref{ineq_11},
we have used \cref{claim:Distinguish_well_in_S},
and in \cref{ineq_12}, we have used \cref{claim:Not_fooled_in_U}, and in \cref{ineq_13}, we have used \cref{fact:distance}.
This implies that
\begin{align}
&\mathsf{SD}(\cD(\ell),\cD(m))\\
&\leq
\left(1+\frac{1}{4\epsilon(n)}\right)\abs{
\Pr[1\la\mathsf{Dis}_{\epsilon}^{\mathbf{NP}}(\ell,m,x):x\la\cD(\ell)] 
-
\Pr[1\la\mathsf{Dis}_{\epsilon}^{\mathbf{NP}}(\ell,m,x):x\la\cD(m)] }+\frac{1}{4\epsilon(n)}.
\end{align}

\begin{proof}[Proof of \cref{claim:Distinguish_well_in_S}]
In the $\mathsf{Dis}_{\epsilon}^{\mathbf{NP}}(\ell,m,x)$ algorithm,
it runs $\mathsf{Estimate}_{500\epsilon}^{\mathbf{NP}}(\ell,x)$ and $\mathsf{Estimate}_{500\epsilon}^{\mathbf{NP}}\allowbreak(m,x)$.
From union bound, with probability at least $1-\frac{1}{250\epsilon(n)}$ over the internal randomness of $\mathsf{Estimate}_{500\epsilon}^{\mathbf{NP}}(\ell,x)$ and $\mathsf{Estimate}_{500\epsilon}^{\mathbf{NP}}(m,x)$, we have
\begin{align}
\left(1+\frac{1}{500\epsilon(n)}\right)\Pr[x\la\cD(\ell)]
\geq \mathsf{Estimate}_{500\epsilon}^{\mathbf{NP}}(\ell,x)
\geq \left(1-\frac{1}{500\epsilon(n)}\right)\Pr[x\la\cD(\ell)]
\end{align}
and
\begin{align}
\left(1+\frac{1}{500\epsilon(n)}\right)\Pr[x\la\cD(m)]
\geq \mathsf{Estimate}_{500\epsilon}^{\mathbf{NP}}(m,x)
\geq \left(1-\frac{1}{500\epsilon(n)}\right)\Pr[x\la\cD(m)].
\end{align}
From the definition of $\cS_n$, for all $x\in\cS_n$, we have
\begin{align}
&\mathsf{Estimate}_{500\epsilon}^{\mathbf{NP}}(\ell,x)
\geq 
\left(1-\frac{1}{500\epsilon(n)}\right)\left(1+\frac{1}{8\epsilon(n)}\right)\Pr[x\la\cD(m)]\\ 
&\geq\frac{500\epsilon(n)-1}{500\epsilon(n)+1}\left(1+\frac{1}{8\epsilon(n)}\right)
\mathsf{Estimate}_{500\epsilon}^{\mathbf{NP}}(m,x)
\geq\left(1+\frac{1}{16\epsilon(n)}\right)\mathsf{Estimate}_{500\epsilon}^{\mathbf{NP}}(m,x).
\end{align}
This completes the proof.
\end{proof}

\begin{proof}[Proof of \cref{claim:Not_fooled_in_U}]
In the $\mathsf{Dis}_{\epsilon}^{\mathbf{NP}}(\ell,m,x)$ algorithm,
it runs $\mathsf{Estimate}_{500\epsilon}^{\mathbf{NP}}(\ell,x)$ and $\mathsf{Estimate}_{500\epsilon}^{\mathbf{NP}}\allowbreak(m,x)$.
From union bound, with probability at least $1-\frac{1}{250\epsilon(n)}$ over the internal randomness of $\mathsf{Estimate}_{500\epsilon}^{\mathbf{NP}}(\ell,x)$ and $\mathsf{Estimate}_{500\epsilon}^{\mathbf{NP}}(m,x)$, we have
\begin{align}
\left(1+\frac{1}{500\epsilon(n)}\right)\Pr[x\la\cD(\ell)]
\geq \mathsf{Estimate}_{500\epsilon}^{\mathbf{NP}}(\ell,x)
\geq \left(1-\frac{1}{500\epsilon(n)}\right)\Pr[x\la\cD(\ell)]
\end{align}
and
\begin{align}
\left(1+\frac{1}{500\epsilon(n)}\right)\Pr[x\la\cD(m)]
\geq \mathsf{Estimate}_{500\epsilon}^{\mathbf{NP}}(m,x)
\geq \left(1-\frac{1}{500\epsilon(n)}\right)\Pr[x\la\cD(m)].
\end{align}
From the definition of $\cU_n$, for all $x\in\cU_n$, we have
\begin{align}
&\mathsf{Estimate}_{500\epsilon}^{\mathbf{NP}}(m,x)\geq \left(1-\frac{1}{500\epsilon(n)}\right)\Pr[x\la\cD(m)]>\left(1-\frac{1}{500\epsilon(n)}\right)\Pr[x\la\cD(\ell)]\\
&\geq\left(\frac{500\epsilon(n)-1}{500\epsilon(n)+1}\right)\mathsf{Estimate}_{500\epsilon}^{\mathbf{NP}}(\ell,x).
\end{align}
This implies that
\begin{align}
\left(1+\frac{1}{16\epsilon(n)}\right)\mathsf{Estimate}_{500\epsilon}^{\mathbf{NP}}(m,x)>\mathsf{Estimate}_{500\epsilon}^{\mathbf{NP}}(\ell,x).
\end{align}
Therefore, 
the probability is at most $\frac{1}{250\epsilon(n)}$ that the following event occurs
\begin{align}
\mathsf{Estimate}_{500\epsilon}^{\mathbf{NP}}(\ell,x)\geq\left(1+\frac{1}{16\epsilon(n)}\right)\mathsf{Estimate}_{500\epsilon}^{\mathbf{NP}}(m,x).
\end{align}
\end{proof}
\end{proof}

\begin{proof}[Proof of \cref{fact:distance}]

We have 
\begin{align}
&\mathsf{SD}(\cD(m),\cD(\ell))\\
&=\sum_{x\in\cS_n}(\Pr[x\la\cD(\ell)]-\Pr[x\la\cD(m)])+\sum_{x\in\cT_n}(\Pr[x\la\cD(\ell)]-\Pr[x\la\cD(m)])\\
&<\sum_{x\in\cS_n}(\Pr[x\la\cD(\ell)]-\Pr[x\la\cD(m)])+\frac{1}{8\epsilon(n)}\sum_{x\in\cT_n}\Pr[x\la\cD(m)]\\
&\leq \sum_{x\in\cS_n}(\Pr[x\la\cD(\ell)]-\Pr[x\la\cD(m)])+\frac{1}{8\epsilon(n)}
\end{align}
Hence, we have
\begin{align}
\sum_{x\in\cS_n}\Pr[x\la\cD(\ell)]-\Pr[x\la\cD(m)]\geq\mathsf{SD}(\cD(\ell),\cD(m))-\frac{1}{8\epsilon(n)}.
\end{align}
\end{proof}

\if0
\section{Definition of Agnostic Distribution Learning}

In this section, we formalize our learning model, which we call agnostic distribution learning.

\begin{definition}[Agnostic Distribution Learning]\label{def:distribution_learn}
    Let $\cD$ be an algorithm that takes $1^n$ with $n\in\N$ and $z\in\bit^{n}$, and outputs $x\in\bit^n$.
    We say that an algorithm $\cL$ can agnostically learn $\cD$ 
    if, for all distribution $\cT=\{\cT_n\}_{n\in\N}$, there exists a polynomial $t$ such that 
    \begin{align}
        &\Pr[
        \begin{array}{ll}
        \mathsf{SD}(\cT_n, \cD(1^n,h))\leq 3\cdot\min_{a\in\bit^n}\left\{\mathsf{SD}(\cT_n,  \cD(1^n,a))\right\}+\epsilon
        \end{array}
        :
    \begin{array}{ll}
       \{x_1,...,x_{t(n)} \}\la \cT_n^{\otimes t(n,\lfloor1/\epsilon\rfloor,\lfloor1/\delta\rfloor)}\\
        h\la \cL(1^n,1^{\lfloor1/\epsilon\rfloor},1^{\lfloor1/\delta\rfloor},\{x_i\}_{i\in[t(n,\lfloor1/\epsilon\rfloor,\lfloor1/\epsilon\rfloor)]})
    \end{array}
        ]\\
        &\geq 1-\delta
    \end{align}
    for all sufficiently large $n\in\N$.
\end{definition}
\color{black}
\fi

\ifnum\anonymous=0
\paragraph{Acknowledgements.}

TM is supported by
JST CREST JPMJCR23I3,
JST Moonshot R\verb|&|D JPMJMS2061-5-1-1, 
JST FOREST, 
MEXT QLEAP, 
the Grant-in Aid for Transformative Research Areas (A) 21H05183,
and 
the Grant-in-Aid for Scientific Research (A) No.22H00522.
\else
\fi

\ifnum\llncs=1
\bibliographystyle{alpha} 
\bibliography{abbrev3,crypto,reference}
\else
\bibliographystyle{alpha} 
\bibliography{abbrev3,crypto,reference}
\fi

\appendix

\section{Proof of \cref{thm:PP_learn_quantum}}\label{sec:Proof_PP_learning}

For showing \cref{thm:PP_learn_quantum}, we will use the following \cref{thm:PP_distinguish,thm:PP_sample}.
\begin{theorem}\label{thm:PP_distinguish}
For any QPT algorithm $\cD$ that takes $1^n$ and $z\in\bit^n$ as input and outputs $x\in\bit^{\poly(n)}$, there exists a deterministic polynomial-time algorithm $\mathsf{Dis}$ with access to a $\mathbf{PP}$ oracle, which takes $1^n$, $\ell\in\bit^n$, $m\in\bit^n$ and $x\in\bit^{\poly(n)}$, such that for all $(\ell,m)\in\bit^n\times\bit^n$, we have
\begin{align}
   &\mathsf{SD}(\cD(1^n,\ell),\cD(1^n,m))\\
   &= \abs{\Pr[1\la\mathsf{Dis}^{\mathbf{PP}}(1^n,\ell,m,x):x\la\cD(1^n,\ell)]-\Pr[1\la\mathsf{Dis}^{\mathbf{PP}}(1^n,\ell,m,x):x\la\cD(1^{n},m)]}
\end{align}
for all sufficiently large $n\in\N$.
\end{theorem}

\begin{theorem}\label{thm:PP_sample}
For any QPT algorithm $\cD$ that takes $1^n$ and $z\in\bit^n$ as input and outputs $x\in\bit^{\poly(n)}$, there exists a PPT algorithm $\cB$ with access to a $\mathbf{PP}$ oracle such that for all $z\in\bit^n$, we have
\begin{align}
\mathsf{SD}(\cD(1^n,z),\cB^{\mathbf{PP}}(1^n,z))=0
\end{align}
for all sufficiently large $n\in\N$.
\end{theorem}

Note that \cref{thm:PP_distinguish,thm:PP_sample} directly follows from the following \cref{lem:Extrapolation_PP}.
Therefore, we omit the proof of them.
\begin{lemma}[\cite{FR99}]\label{lem:Extrapolation_PP}
    There exists a deterministic polynomial-time algorithm $\cC$ with access to a $\mathbf{PP}$ oracle such that the following holds.
    \begin{align}
        \cC^{\mathbf{PP}}(1^n,\cC,z)=\Pr[z\la\cD(1^n)].
    \end{align}
    Here, 
    $\cD$ is a quantum polynomial-time algorithm that takes $1^n$ as input, and outputs $z$.
\end{lemma}

\begin{proof}[Proof of \cref{thm:PP_learn_quantum}]
The following proof is similar to that of \cref{thm:NP_learning}.

Let $\mathsf{Dis}^{\mathbf{PP}}$ be a deterministic polynomial-time algorithm with access to a $\mathbf{PP}$ oracle given in \cref{thm:PP_distinguish}.
For a QPT algorithm $\cD$, which takes $1^n$ and $z\in\bit^n$ as input, let $\cB$ be a PPT algorithm with access to a $\mathbf{PP}$ oracle given in \cref{thm:PP_sample} such that
\begin{align}
\mathsf{SD}(\cD(1^n,z),\cB^{\mathbf{PP}}(1^n,z))=0.
\end{align}

We set $t(n)\seteq 100n^2\cdot \epsilon(n)^2\left(n +\log(\delta(n))\right)$.
Let us consider the following deterministic algorithm $\cM_{\omega}$ with access to a $\mathbf{PP}$ oracle.

\begin{description}
\item[$\cM_{\omega}^{\mathbf{PP}}(1^n,a,b,\{x_i,r_i\}_{i\in[t(n)]})$:] $ $
\begin{enumerate}
\item Run $\mathsf{Dis}_{\epsilon}^{\mathbf{PP}}(a,b, x_i)$ for all $i\in[t(n)]$.
\item For all $i\in[t(n)]$, run $X_{a,r_i(2)}\la\cB^{\mathbf{PP}}(a;r_i)$ and $\mathsf{Dis}_{\epsilon}^{\mathbf{PP}}(a,b, X_{a,r_i})$.
\item Output $1$ if 
\begin{align}
\abs{\sum_{i\in[t(n)]}\frac{\mathsf{Dis}_{\epsilon}^{\mathbf{PP}}(a,b, x_i )}{t(n)}-\sum_{i\in[t(n)]}\frac{\mathsf{Dis}_{\epsilon}^{\mathbf{PP}}(a,b, X_{a,r_i})}{t(n)}}\leq\mathsf{Opt}_n+\omega(n).
\end{align}
Otherwise, output $0$.
\end{enumerate}
\end{description}

For any function $\omega$, and any $a_1,...,a_i\in\bit^{i}$ with $i\in\N$, let us define
\begin{align}
&\cL_{\omega,(a_1,...,a_i)}\\
&\seteq 
\bigg\{\{x_i,r_i\}_{i\in[t(n)]}\in\bit^*,n\in\N:
\exists a\in\bit^{n-i} \mbox{ such that } \forall b\in\bit^n \mbox{, we have}\\
&1\la \cM_{\omega}^{\mathbf{PP}}(1^n,(a_1,...,a_{i}),a,b,\{x_i,r_i\}_{i\in[t(n)]})
\bigg\}.
\end{align} 
Clearly, we have $\cL_{\omega,(a_1,...,a_i)}\in\Sigma_2^{\mathbf{PP}}$.

We consider the following PPT algorithm  $\cA$ with access to a $\Sigma_2^{\mathbf{PP}}$ oracle:
\begin{description}
\item[$\cA^{\Sigma_2^{\mathbf{PP}}}(1^n,\{x_i\}_{i\in[t(n)]})$:] $ $ 
\begin{enumerate}
%\item Set $A=\min\{\epsilon,\mathsf{Opt}\}$.
\item Uniformly randomly sample $\{r_i\}_{i\in[t(n)]}$.
\item For all $I\in[n]$, run the following.
\begin{enumerate}
\item Check if there exists $a_{I+1},...,a_n\in\bit^{n-I}$ such that
\begin{align}
&\Biggl|\sum_{i\in[t(n)]}\frac{\mathsf{Dis}_{\epsilon}^{\mathbf{PP}}(\mathsf{out}_1,...,\mathsf{out}_{I-1},1,a_{I+1},...,a_n,b, x_i)}{t(n)}\\
&\hspace{1cm}-\sum_{i\in[t(n)]}\frac{\mathsf{Dis}_{\epsilon}^{\mathbf{PP}}(\mathsf{out}_1,...,\mathsf{out}_{I-1},1,a_{I+1},...,a_n,b, X_{(\mathsf{out}_1,...,\mathsf{out}_{I-1},1,a_{I+1},...,a_n),r_i} )}{t(n)}
\Biggl|\\
&
\hspace{5cm}\leq\mathsf{Opt}_n+\frac{2I-1}{4n\cdot \epsilon(n)}
\end{align}
for all $b\in\bit^n$
by checking that
$\left(\{x_i,r_i\}_{i\in[t(n)]},n\right)$ is contained in $ \cL_{\frac{2I-1}{4n\cdot\epsilon(n)},(\mathsf{out}_1,...,\mathsf{out}_{I-1},1)}$.
\item If so, set $\mathsf{out}_I=1$. Otherwise, set $\mathsf{out}_{I}=0$.
\end{enumerate}
\item Output $\mathsf{out}_1,...,\mathsf{out}_n$.
\end{enumerate}
\end{description}
We can show that
$\cA^{\Sigma_2^{\mathbf{PP}}}$ satisfies
\begin{align}
\Pr[ 
\mathsf{SD}(\cD(1^n,h),\cT(1^n))\leq \left(3+\frac{1}{\epsilon(n)}\right)\cdot\mathsf{Opt}_n+\frac{1}{\epsilon(n)}:
\begin{array}{ll}
&\{x_i\}_{i\in[t(n)]}\la\cT(1^n)^{\otimes t(n)} \\
& h\la \cA^{\Sigma_2^{\mathbf{PP}}}(1^n,\{x_i\}_{i\in[t(n)]})
\end{array}
]\geq 1-1/\delta(n)    
\end{align}
for all sufficiently large $n\in\N$.
Here, $\mathsf{Opt}_n\seteq \min_{a\in\bit^n}\{\mathsf{SD}(\cT(1^n),\cD(1^n,a))\}$.
The analysis is the same as \cref{thm:NP_learning}, and thus we omit it.
\end{proof}

\section{Proof of \cref{thm:Learning_hardness_is_weaker}}\label{app:proof_1}

In the following, we show \cref{thm:Learning_hardness_is_weaker}.
For showing \cref{thm:Learning_hardness_is_weaker}, let us introduce average-case hardness of quantum probability estimation secure against QPT algorithm with access to a $\Sigma_3^{\mathbf{P}}$ oracle as follows.

\begin{definition}[Average-case hardness of quantum probability estimation secure against QPT algorithm with $\Sigma_3^{\mathbf{P}}$ oracle]\label{QPE}
We say that average-case hardness of quantum probability estimation secure against QPT algorithm with access to a $\Sigma_3^{\mathbf{P}}$ oracle holds if the following holds:
There exists a QPT algorithm $\cS$ that takes $1^n$ as input and outputs $x\in\bit^n$ such that for all QPT algorithms $\cA$ with access to a $\Sigma_3^{\mathbf{P}}$ oracle and constants $c>0$
\begin{align}
\Pr_{x\la\cS(1^n)}\left[\abs{\cA^{\Sigma_3^{\mathbf{P}}}(x)-\Pr[x\la\cS(1^n)]}\leq \frac{\Pr[x\la\cS(1^n)]}{n^c}\right]\leq 1-\frac{1}{n^c}
\end{align}
for infinitely many $n\in\N$.
\end{definition}

\begin{proof}[Proof of \cref{thm:Learning_hardness_is_weaker}]
The proof follows from \cite{STOC:KT25,C:ChuGolGra24} and \cref{lem:OWPuzz_to_Learn_avg}.

In the similar way as Theorem 5.1 of \cite{STOC:KT25}, we can show that \cref{def:QAA} and $\mathbf{P^{\#P}}\nsubseteq \mathbf{BQP^{\Sigma_3^{\mathbf{P}}}}$ implies \cref{QPE}.
In the similar way as Theorem 5.2 of \cite{STOC:KT25}~\footnote{
See also Lemma 4.3 of \cite{myHM24}.
}, we can show that \cref{QPE} implies distributionally OWPuzzs secure against QPT algorithm with access to a $\Sigma_3^{\mathbf{P}}$ oracle.
In the similar way as Theorem 29 of \cite{C:ChuGolGra24}, we can show that distributionally OWPuzzs secure against QPT with access to a $\Sigma_3^{\mathbf{P}}$ oracle implies nuQEFID secure against QPT algorithm with access to a $\Sigma_3^{\mathbf{P}}$ oracle.
In the similar way as \cref{lem:OWPuzz_to_Learn_avg}, we can show that nuQEFID secure against QPT algorithm with access to a $\Sigma_3^{\mathbf{P}}$ oracle implies
average-case hardness of proper quantum distribution learning secure against QPT algorithm with access to a $\Sigma_3^{\mathbf{P}}$ oracle.
Average-case hardness of proper quantum distribution learning secure against QPT algorithm with access to a $\Sigma_3^{\mathbf{P}}$ oracle directly implies \cref{assumption:PP_hardness_of_Dis_learn}.
\end{proof}

\if0
\color{red}
\section{On polynomial Hierarchy}
\taiga{Prepare the following just in case.}

\begin{theorem}
If $\mathbf{PP}\subseteq\mathbf{BQP}$, then $\Sigma_2^{\mathbf{PP}}\subseteq\mathbf{BQP}$.
\end{theorem}

\begin{proof}
First, we show that if $\mathbf{PP}\subseteq \mathbf{BQP}$, then $\Sigma_1^{\mathbf{PP}}\subseteq\mathbf{BQP}$.
Let $\cL\in\Sigma_1^{\mathbf{PP}}$.
Then, there exists a deterministic polynomial-time algorithm $\cM$ and a polynomial $p$ such that the following two conditions are satisfied:
\begin{enumerate}
\item For any $x\in\cL$, there exists $w_1\in\bit^{p(\abs{x})}$ such that
\begin{align}
1\la\cM^{\mathbf{PP}}(x,w_1).
\end{align}
\item For any $x\notin\cL$, for all $w_1\in\bit^{p(\abs{x})}$, we have
\begin{align}
0\la\cM^{\mathbf{PP}}(x,w_1).
\end{align}
\end{enumerate}
We consider the following deterministic polynomial-time algorithm $\cM^*(x)$ with $\mathbf{PP}$ oracle.
\begin{description}
\item[$\cM^{*\mathbf{PP}}(x)$: ]$ $
\begin{enumerate}
\item Sample $w_1\la\bit^{p(\abs{x})}$.
\item Run $\cM^{\mathbf{PP}}(x,w_1)$, and output its output.
\end{enumerate}
\end{description}
From the assumption of $\mathbf{BQP}=\mathbf{PP}$, there exists a QPT algorithm $\cQ$ such that

$\cM^{*\mathbf{PP}}(x)$ satisfies the following conditions:
\begin{enumerate}
    \item If $x\in\cL$, then
\begin{align}
\Pr[1\la\cM^{*\mathbf{PP}}(x)]>1/2.
\end{align}
\item If $x\notin\cL$, then
\begin{align}
\Pr[1\la\cM^{*\mathbf{PP}}(x)]<1/2.
\end{align}
\end{enumerate}
\end{proof}
\fi

\ifnum\cameraready=1
\else
\ifnum\submission=1
\newpage
\setcounter{tocdepth}{1}
\tableofcontents
\else
\fi
\fi

\end{document}